\documentclass[11pt, lmargin=1.25in, rmargin=1.25in]{article}

\RequirePackage{amsthm,amsmath,amsfonts,amssymb}
\RequirePackage[authoryear]{natbib} 
\RequirePackage[colorlinks,citecolor=blue,urlcolor=blue,breaklinks]{hyperref}

\usepackage{bm}
\usepackage{mathrsfs}  
\usepackage{appendix}
\usepackage{commath}
\usepackage{dsfont}
\usepackage{subfig}
\usepackage{xcolor}
\usepackage{mathtools}
\usepackage{enumitem}
\usepackage{fancyhdr}
\usepackage{mathrsfs} 
\usepackage{geometry}
\usepackage{caption}

\usepackage{setspace}
\newcommand{\N}{\mathbb{N}}
\newcommand{\R}{\mathbb{R}}
\newcommand{\eps}{\varepsilon}
\renewcommand{\P}{\mathbb{P}}
\newcommand\E{\mathbb{E}}
\newcommand{\mc}{\mathcal}
\newcommand{\indep}{\perp \!\!\! \perp}
\DeclareMathOperator*{\argmax}{arg\,max}

\theoremstyle{plain}
\newtheorem{theorem}{Theorem}[section]
\newtheorem{lemma}[theorem]{Lemma}
\newtheorem{corollary}[theorem]{Corollary}

\theoremstyle{remark}
\newtheorem{remark}{Remark}[section]

\newtheorem{assumption}{Assumption}[section]

\begin{document}

\title{Enhanced power enhancements for testing many moment equalities: Beyond the~$2$- and~$\infty$-norm
}

\author{
\begin{tabular}{c}
Anders Bredahl Kock \\ 
\small	University of Oxford \\
\small Department of Economics\\
\small	10 Manor Rd, Oxford OX1 3UQ
\\
\small	{\small	\href{mailto:anders.kock@economics.ox.ac.uk}{anders.kock@economics.ox.ac.uk}} 
\end{tabular}
\and
\begin{tabular}{c}
David Preinerstorfer \\ 
{\small	WU Vienna University of Economics and Business} \\
{\small Institute for Statistics and Mathematics} \\
{\small	Welthandelsplatz 1, 1020 Vienna} \\ 
{\small	 \href{mailto:david.preinerstorfer@wu.ac.at}{david.preinerstorfer@wu.ac.at}}
\end{tabular}
}

\date{October 2024}

\maketitle	

\begin{abstract}
Contemporary testing problems in statistics are increasingly complex, i.e., high-dimensional. Tests based on the~$2$- and~$\infty$-norm have received considerable attention in such settings, as they are powerful against dense and sparse alternatives, respectively. The power enhancement principle of~\cite{fan2015power} combines these two norms to construct improved tests that are powerful against both types of alternatives. In the context of testing whether a candidate parameter satisfies a large number of moment equalities, we construct a test that harnesses the strength of \emph{all}~$p$-norms with $p\in[2, \infty]$. As a result, this test is consistent against strictly more alternatives than \emph{any} test based on a single~$p$-norm. In particular, our test is consistent against more alternatives than tests based on the~$2$- and~$\infty$-norm, which is what most implementations of the power enhancement principle target. 

We illustrate the scope of our general results by using them to construct a test that simultaneously dominates the Anderson-Rubin test (based on~$p=2$), tests based on the~$\infty$-norm and power enhancement based combinations of these in terms of consistency in the linear instrumental variable model with many instruments.

\end{abstract}

\section{Introduction}
In the era of big data, target parameters in statistical models are often (partially) identified by a large number of moment equalities. Then, one frequently wishes to test whether a candidate (structural) parameter $\bm{\beta}^*_n$ is an element of the identified set, i.e., whether~$\bm{\beta}^*_n$ satisfies~$\E h_n(\bm{X},\bm{\beta}^*_n)=\bm{0}_d$.\footnote{Here~$\bm{\beta}^*_n$ may or may not be point-identified, cf., e.g.,~\cite{molinari2020microeconometrics} for many examples of how lack of point-identification may arise.}  An important example which can be cast within this framework and that has received much recent attention is testing restrictions on coefficients in regression models in the presence of endogenous regressors and many (weak) instruments; recent references include \cite{mikusheva2020inference},  \cite{matsushita2022jackknife},~\cite{boot2023identification}, and \cite{dovi2022ridge}. Indeed, this is our running example and further details and references can be found in~Section \ref{sec:IVexample}. Testing whether a treatment has an effect on one of many outcome variables or one of several groups of individuals is another example of high practical relevance, cf.~Section~\ref{sec:exmanyoutcomes}. More generally, true parameter(s) in~$M$-,~$Z$-, and GMM-estimation problems satisfy such moment restrictions, and hypotheses on the mean vector(s) or covariance matrices in classic one- and two-sample testing problems can be accommodated too.


Given a sample $\bm{X}_{1,n}, \hdots, \bm{X}_{n,n}$ of random variables with the same marginal distribution as $\bm{X}$, denote the (scaled) empirical counterpart of the population moments~$\E h_n(\bm{X},\bm{\beta}^*_n)$ by
\begin{align*}
\bm{H}_{n}(\bm{\beta}^*_n)&:=\frac{1}{\sqrt{n}}\sum_{i=1}^nh_{n}(\bm{X}_{i,n},\bm{\beta}^*_n),
\end{align*}
and let~$\hat{\bm{\Sigma}}_n(\bm{\beta}^*_n)$ be an estimator of the covariance matrix~$\bm{\Sigma}_n(\bm{\beta}^*_n)$ of~$\bm{H}_n(\bm{\beta}^*_n)$. It is then natural to base a test of~$\E h_n(\bm{X},\bm{\beta}^*_n)=\bm{0}_d$ on the distance of~$\hat{\bm{\Sigma}}^{-1/2}_n(\bm{\beta}^*_n)\bm{H}_{n}(\bm{\beta}^*_n)$ from the origin. To measure this distance, one most commonly makes use of a~$p$-norm, which then requires the choice of the exponent~$p$. 
Most tests used in practice are based on~$p=2$ or~$p=\infty$ and it is well-understood that tests based on the former are relatively powerful against ``dense'' alternatives whereas tests based on the latter are relatively powerful against ``sparse'' alternatives, cf., e.g.,~\cite{ingster2003nonparametric}. To construct a test that is simultaneously powerful against dense as well as many sparse alternatives,~\cite{fan2015power} combined tests based on the~$2$- and $\infty$-norm via their \emph{power enhancement principle}. The idea of combining the~$2$- and $\infty$-norm based tests or a \emph{finite} number of $p$-norm based tests in order to construct a more powerful test has since gained considerable popularity, cf.~\cite{xu2016adaptive}, \cite{yang2017weighted}, \cite{tang2018testing}, \cite{kp},
\cite{liu2019statistical}, \cite{jammalamadaka2020sobolev},  \cite{he2021asymptotically}, \cite{feng2020max}, \cite{juodis2022incidental}, \cite{zhang2021adaptive}, \cite{yu2023power}, 
\cite{enhanceIV}, \cite{ge2022dynamic}, \cite{li2024power},  \cite{yu2020fisher}, \cite{yu2024power}, \cite{li2024detection}.

Despite the success of the power enhancement principle and related combination procedures, sparse and dense alternatives are merely two (conceptually useful) ``endpoints" between which a continuum of ``semi-sparse'' alternatives exist. This continuum of structures is mirrored by a continuum of~$p$-norms,~$p\in[2,\infty]$, between the two extremes~$p=2$ and~$p=\infty$ that most tests are based on. In the clean but restrictive testbed of the Gaussian sequence model,~\cite{kp2021consistency} showed that there exist semi-sparse alternatives against which tests based on the~$2$- and~$\infty$-norm are inconsistent, but against which tests based on any~$p\in(2,\infty)$ are consistent. Thus, it is important to harness the power from \emph{all}~$p$-norms,~$p\in[2,\infty]$, which \cite{kp2021consistency} exploited to construct a test that is ``dominant''  --- in the Gaussian sequence model --- in the following sense: If there exists \emph{some}~$p$ for which the corresponding~$p$-norm based test is consistent against a given alternative, then so is their test. In particular, as Figure~\ref{fig:intro} illustrates for a slight variation~$\psi_d$ (which is more convenient for our purposes) of their test,~$\psi_d$ is

\begin{enumerate}
	\item about as powerful against sparse alternatives as a test based on~$p=\infty$, but much more powerful than a test based on~$p=2$;	
	\item about as powerful against dense alternatives as a test based on~$p=2$, but much more powerful than a test based on~$p=\infty$;
	\item \emph{often more powerful} against semi-sparse alternatives than tests based on~$p=2$,~$p=\infty$ or the power enhancement principle. This power gain can be of practical relevance as there is often no reason to believe that alternatives (or, more relevant, the non-centrality parameter in~\eqref{eq:nc2} below) are exactly sparse or dense. The gain comes from harnessing the strengths of~$p$-norms beyond~$2$ and~$\infty$, cf.~also the discussion in Section~\ref{sec:domtest}.
	\end{enumerate}	

\begin{figure}

\begin{center}
	\footnotesize Additional power of~$\psi_d$ over each of the other tests studied
\end{center}
\vspace{-0.3cm}
\includegraphics[width=5.2cm]{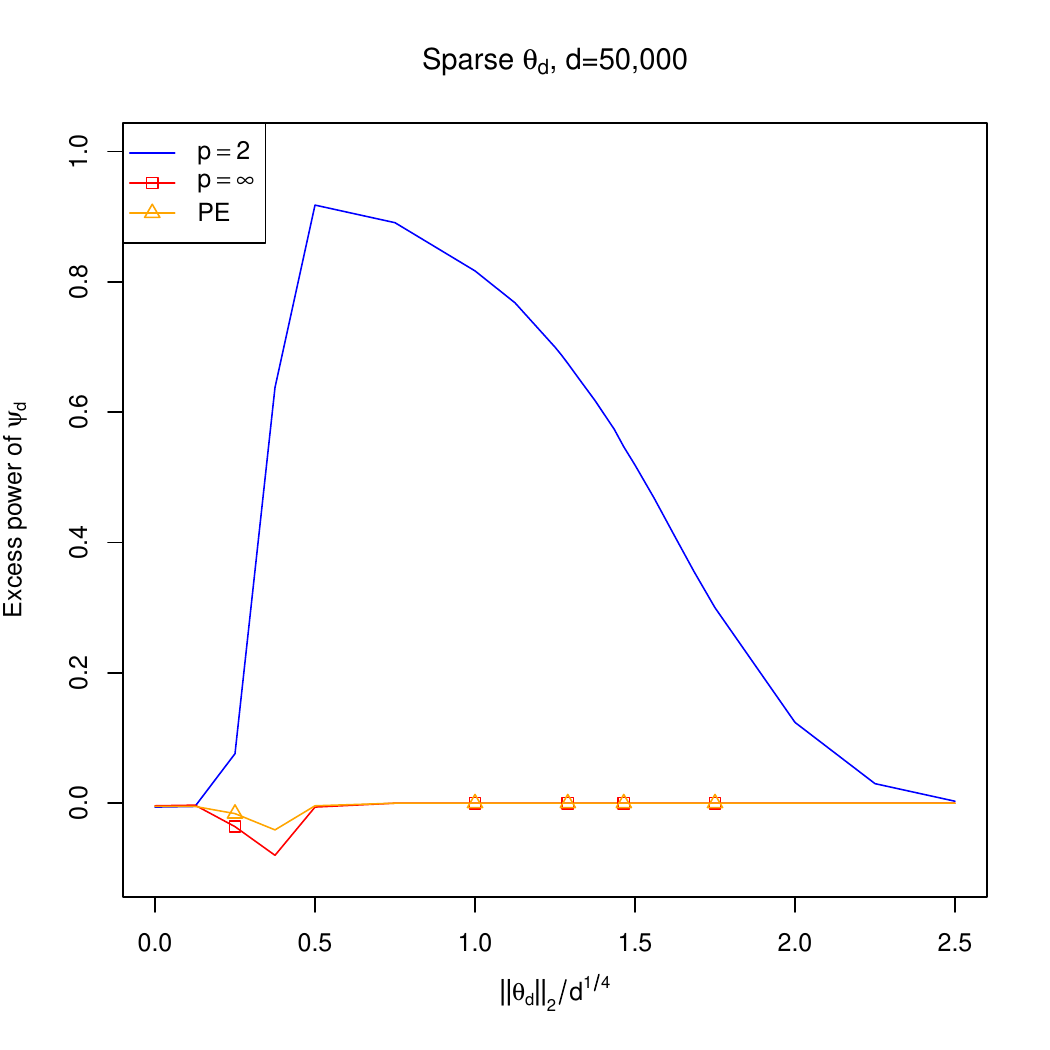}
\hspace{-0.6cm}
\includegraphics[width=5.2cm]{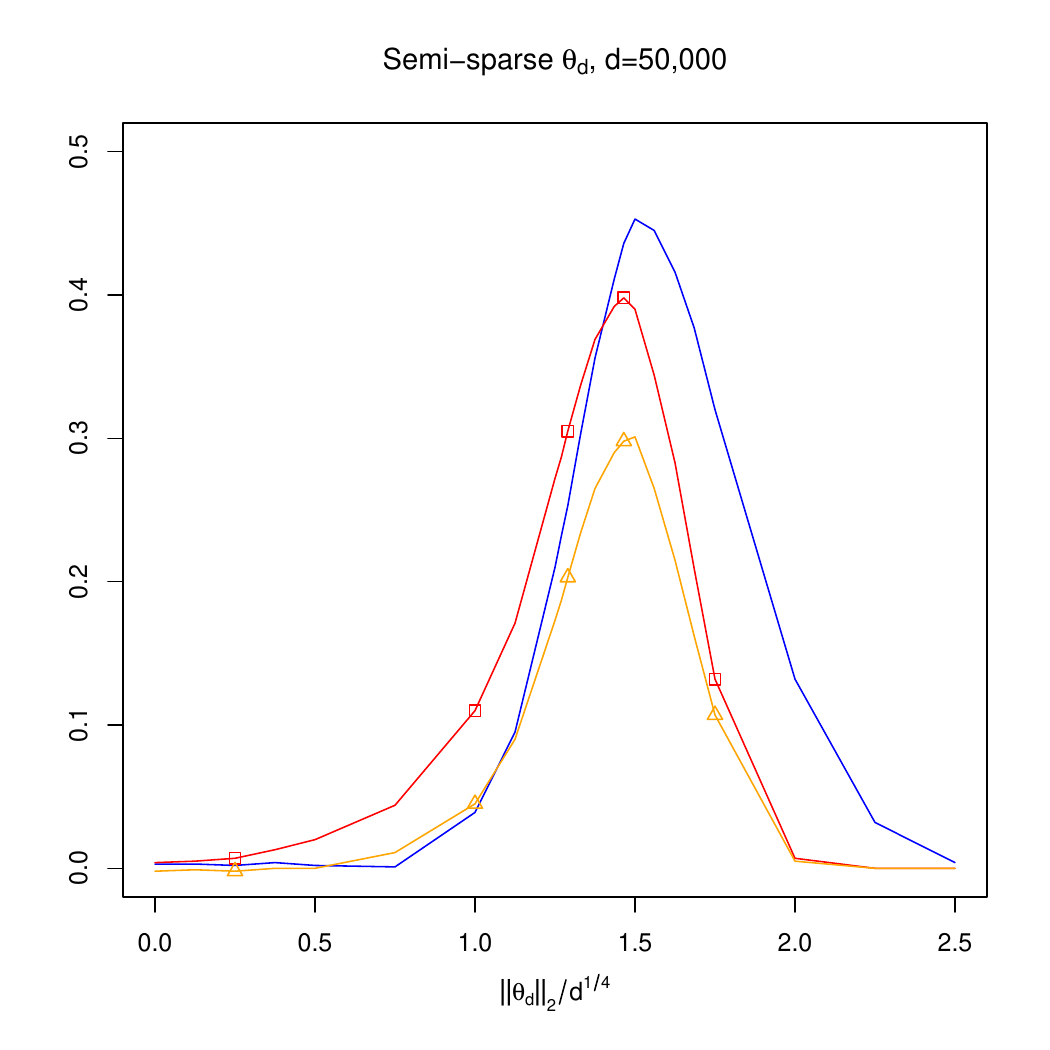}
\hspace{-0.6cm}
\includegraphics[width=5.2cm]{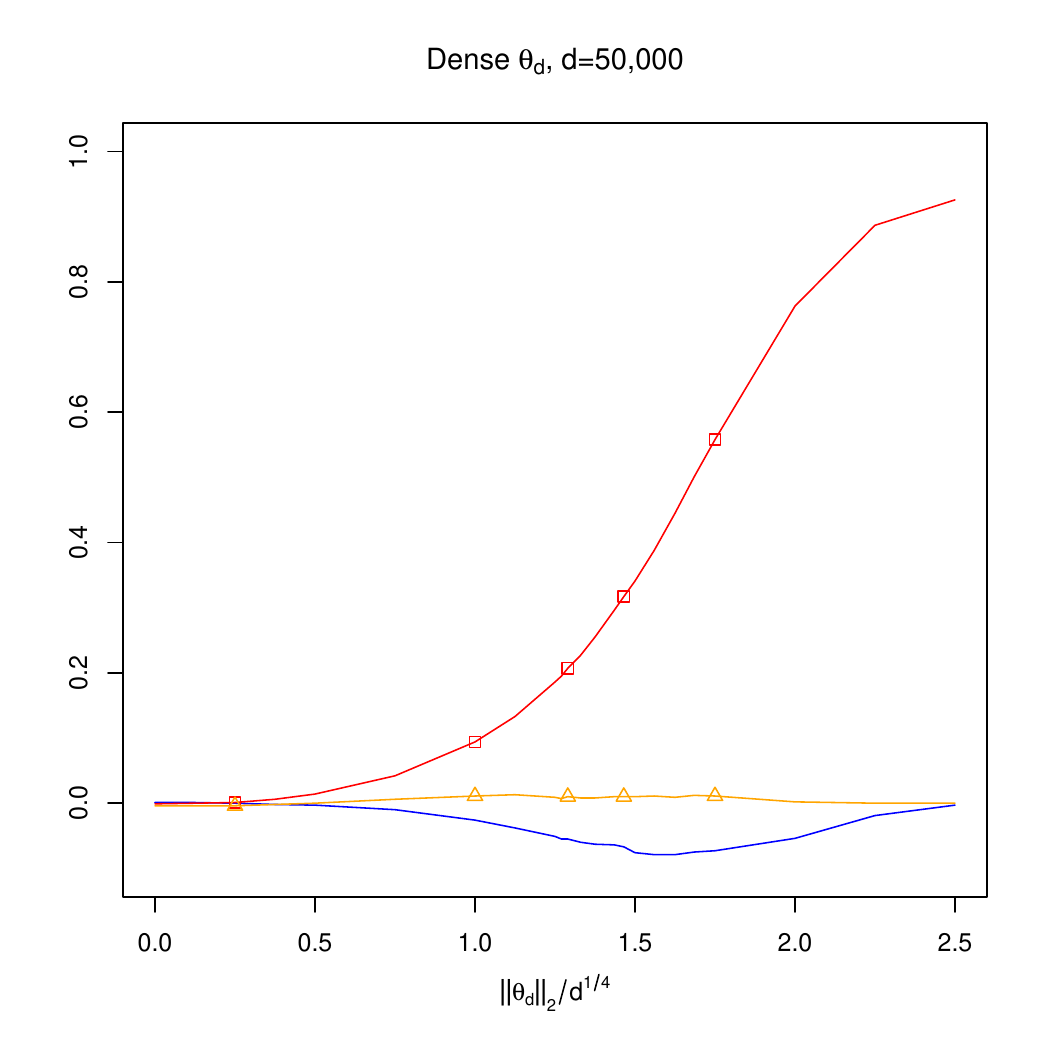}
\begin{center}
	\footnotesize Raw power functions of each of the tests studied
\end{center}
\vspace{-0.3cm}
\includegraphics[width=5.2cm]{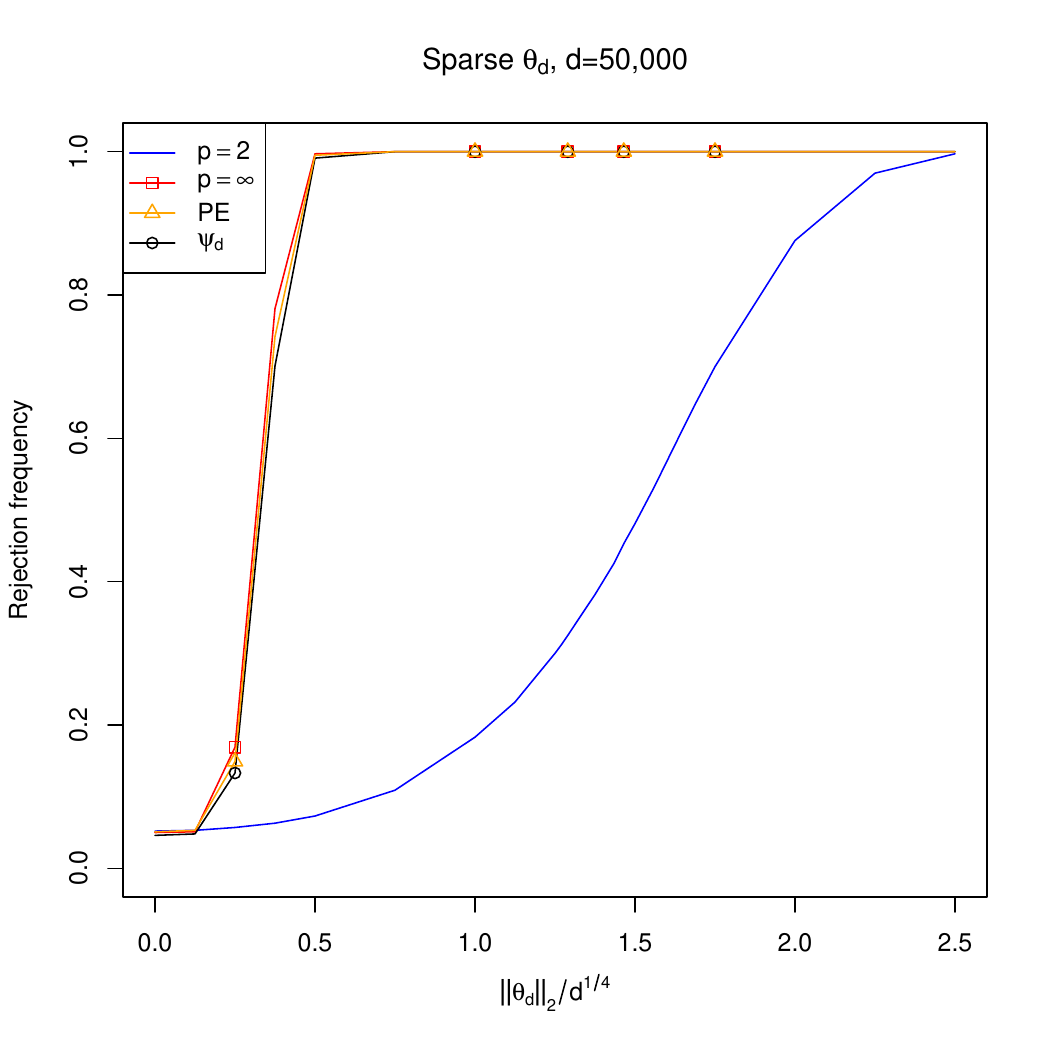}
\hspace{-0.6cm}
\includegraphics[width=5.2cm]{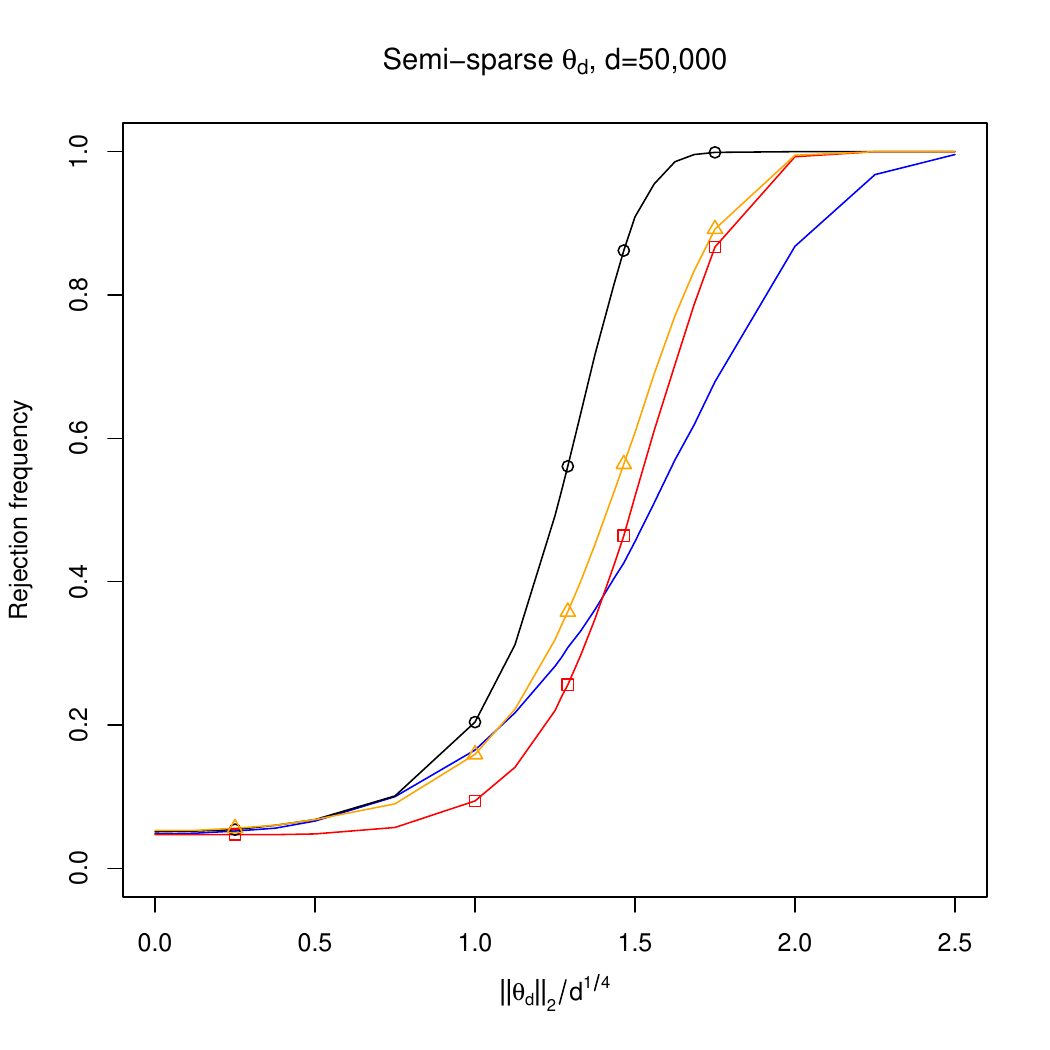}
\hspace{-0.6cm}
\includegraphics[width=5.2cm]{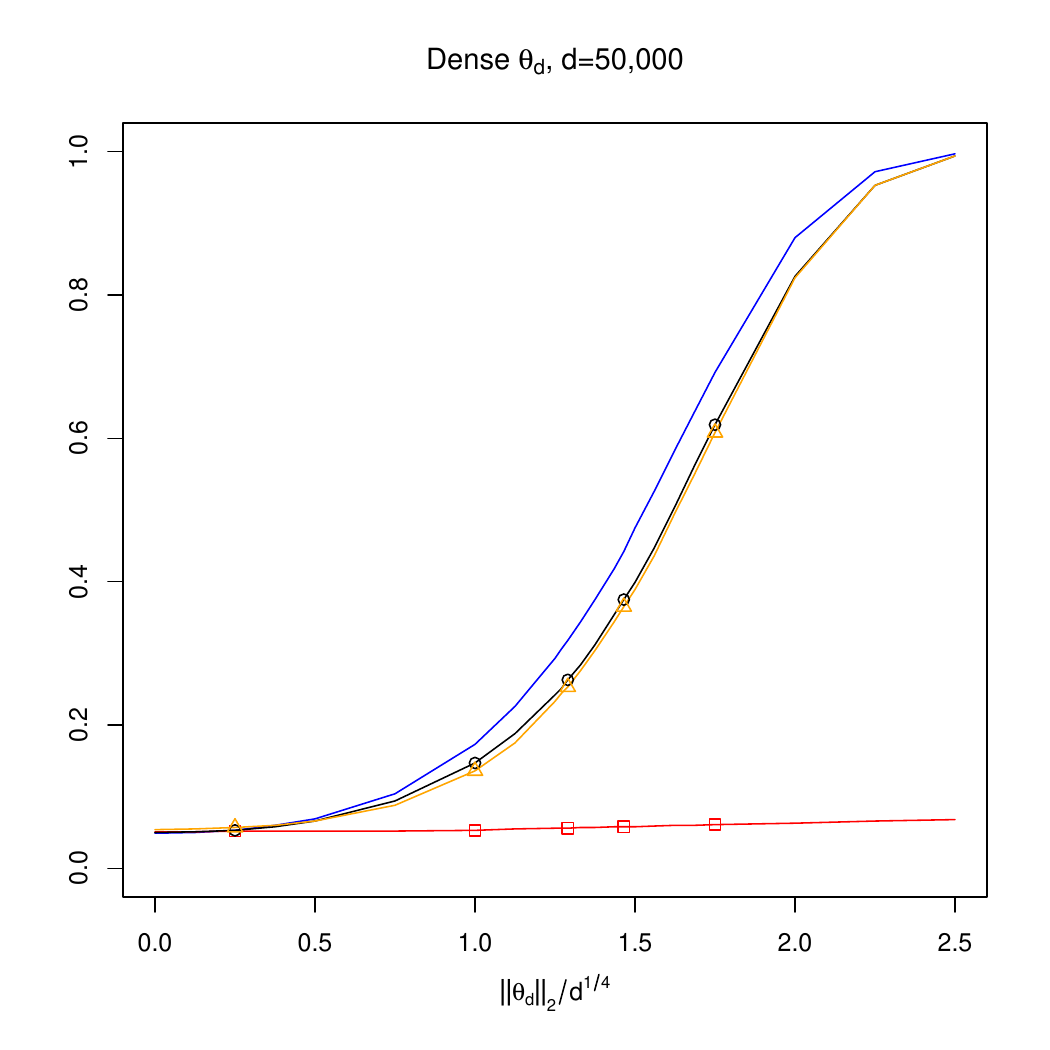}

\caption{\footnotesize For~$d=50{,}000$ and~$\bm{Z}_d\sim \mathsf{N}_d(\bm{\theta}_d,\mathbf{I}_d)$ the top row plots by how much the power  of~$\psi_d$ \emph{exceeds} that of~$\mathds{1}\del[0]{\enVert[0]{\bm{Z}_d}_2\geq \kappa_{d,2}}$ [``$p=2$''],~$\mathds{1}\del[0]{\enVert[0]{\bm{Z}_d}_\infty\geq \kappa_{d,\infty}}$ [``$p=\infty$''], and a power enhancement (PE) based test combining the $2$- and $\infty$-norm for testing~$H_{0,d}:\bm{\theta}_d=\bm{0}_d$ against~$H_{1,d}:\bm{\theta}_d\neq\bm{0}_d$. The bottom row plots the raw power functions. All tests have size~$0.05$. For the sparse alternatives,~$\bm{\theta}_d$ differs from zero in one entry only, for the dense alternatives all entries are equal and non-zero, and for the semi-sparse alternatives~$80$ entries are equal and non-zero. Full implementation details can be found in Section~\ref{sec:simGauss}. For sparse~$\bm{\theta}_d$, the power function of our test~$\psi_d$ is almost identical to that of the supremum-norm based test, but much higher than that of the~$2$-norm based test. For dense~$\bm{\theta}_d$, the power function of~$\psi_d$ is only marginally below that of the~$2$-norm based test, but much higher than that of~$\infty$-norm based test. Finally, for semi-sparse~$\bm{\theta}_d$, the power of~$\psi_d$ is i) nowhere below that of the~$2$- and~$\infty$-norm based tests and ii) up to twice as high as the power of these (the power of~$\psi_d$ is~$0.8$ where that of the~$2$- and~$\infty$-norm based tests is~$0.4$). The power of~$\psi_d$ is up to~$0.3$ higher than that of the PE based test.}
\label{fig:intro}
\end{figure}

In the present paper, we leverage the findings of~\cite{kp2021consistency} to \emph{construct a test simultaneously dominating all~$p$-norm based tests with~$p\in[2,\infty]$ in terms of consistency for testing~$H_{0,n}:\E h_n(\bm{X},\bm{\beta}^*_n)=\bm{0}_d$}.\footnote{The results in \cite{kp2021consistency} show that already in the Gaussian sequence model tests based on~$p$-norms with~$p \leq 2$ are dominated (in terms of consistency) by the $2$-norm based test. For simplicity, we therefore focus on $p \in [2, \infty]$ in the present paper.} This requires rigorous high-dimensional Gaussian approximations to justify that results obtained in the stylized Gaussian sequence model carry over to more complex empirically relevant models. Indeed, a large and vibrant literature on high-dimensional Gaussian approximations over hyperrectangles was sparked by the work of~\cite{chernozhukov2013gaussian} and is surveyed in~\cite{chernozhukov2023high}. Such  approximations over hyperrectangles are suitable for studying tests based on~$p=\infty$. However, since we study the full \emph{family} of~$p$-norm based tests and construct a test dominating all of its members, we cannot rely on these approximations, because the rejection regions of~$p$-norm based tests, albeit convex, are no hyperrectangles. Similarly, central limit theorems for quadratic forms (related to 2-norm type tests) do not suffice for our purpose. Gaussian approximations over convex sets have received considerable attention, cf.~\cite{nagaev2006estimate, senatov1981uniform, gotze1991rate, bentkus2003dependence, bentkus2005lyapunov, fang2020large}.

\subsection{Summary of results}
We now discuss our contributions in more detail. In Theorem \ref{thm:pnorm-char_GMM} we first exhibit high-level conditions under which the characterization of the consistency properties of~$p$-norm based tests in the Gaussian sequence model from \cite{kp2021consistency} carries over to the setting considered in the present paper. We then establish in Theorem \ref{thm:prim} that these high-level conditions are satisfied under
\begin{enumerate}[label=(\roman*)]
\item  a precision guarantee for~$\hat{\bm{\Sigma}}_n(\bm{\beta}^*_n)$ as an estimator of $\bm{\Sigma}_n(\bm{\beta}^*_n)$, and
\item a Gaussian approximation condition over convex sets for~$\bm{\Sigma}^{-1/2}_n(\bm{\beta}^*_n)\bm{H}_{n}(\bm{\beta}^*_n)$.
\end{enumerate}
In particular, although it is not our focus, dependent data can be accommodated. We also provide a catalogue of primitive sufficient conditions for (i) and (ii) in the i.i.d.~case and discuss their implication on the allowed growth rate of~$d$ in dependence of~$n$. To this end, because it is often not realistic to impose that the entries of~$h_{n}(\bm{X}_{i,n},\bm{\beta}^*_n)$ have light tails, we focus on a setting wherein (essentially) only four moments exist. 
Here, we need~$d=o(n^{2/5})$, cf.~Corollary \ref{cor:fourmoments}, and we emphasize that this rate cannot be improved by much, even if one restricts attention to tests based on~$p=\infty$ only, see also~\cite{kock2023moment}. Clearly,~$d$ larger than~$n$ is ruled out. Nevertheless, many empirically relevant settings involving a large~$d$ that is still smaller than~$n$ are covered. For example, \cite{angrist2022machine} based on the data from \cite{angrist1991does} consider a setting with~$n=329{,}509$ and up to~$d=1{,}530$ instruments.\footnote{\label{fn:motivation}Similarly, register-based data sets, such as those provided by~\emph{Statistics Denmark}, allow one to access data on entire populations yielding very large~$n$ (in the millions) and thus allowing~$d$ to be large as well.} Alternatively, one can verify (i) and (ii) under stronger moment assumptions or by imposing structure on~$\bm{\Sigma}_n(\bm{\beta}^*_n)$ to allow~$d$ to grow faster than~$n^{2/5}$. Furthermore, in Section~\ref{sec:rapidgrowth}, we introduce ``projection-based'' tests based on a sample-splitting device that overcome dimensionality constraints completely (without sacrificing asymptotic size).

On a technical level, the conditions (i) and (ii) allow us to establish a Gaussian approximation for~$\hat{\bm{\Sigma}}^{-1/2}_n(\bm{\beta}^*_n)\bm{H}_{n}(\bm{\beta}^*_n)$ over convex sets when the ``noncentrality term''
\begin{align}\label{eq:nc2}
	\bm{\theta}_n(\bm{\beta}^*_n)=\sqrt{n}\bm{\Sigma}^{-1/2}_n(\bm{\beta}^*_n)\E h_{n}(\bm{X},\bm{\beta}^*_n)
\end{align}
 is sufficiently close to~$\bm{0}_d$, cf.~Lemma~\ref{lem:convgauss}. This enables us to deduce the power properties of tests based on the statistics~$\enVert[1]{\hat{\bm{\Sigma}}^{-1/2}_n\bm{H}_{n}(\bm{\beta}^*_n)}_p$ for any~$p\in[2,\infty]$ from those of~$p$-norm based tests in the Gaussian sequence model, which was analyzed in~\cite{kp2021consistency}. However, the case of ``large''~$\bm{\theta}_n(\bm{\beta}^*_n)$ must be handled by a separate proof strategy tailored to this regime.
 
  We find that a~$p$-norm based test is consistent for~$p\in[2,\infty)$ if and only if~$$d^{-1/2}\sum_{j=1}^d \sbr[1]{\theta_{j,n}^2(\bm{\beta}^*_n)\vee |\theta_{j,n}(\bm{\beta}^*_n)|^p}\to\infty,$$ and obtain a structurally similar, but somewhat more complex, necessary and sufficient condition for consistency of the~$\infty$-norm based test. Thus, apart from the case of~$p=2$, consistency of a~$p$-norm based test is \emph{not} solely determined by the~$p$-norm of~$\bm{\theta}_n(\bm{\beta}^*_n)$. Furthermore, the characterization reveals the following monotonicity: For~$2\leq p<q<\infty$ tests based on the~$q$-norm are consistent against more alternatives than tests based on the~$p$-norm. This monotonicity does \emph{not} extend to~$p=\infty$.\footnote{Although sufficient, yet not necessary, conditions are known for consistency of tests based on~$p=\infty$ (based on maximin rates of testing, cf.~\cite{ingster2003nonparametric}), we are not aware of any full characterization providing necessary and sufficient conditions beyond the Gaussian setting considered in \cite{kp2021consistency}. Thus, our characterization also sheds new light on tests based on~$p=\infty$, which are frequently employed.} 

Based on the monotonicity property, we then construct a test~$\psi_n(\bm{\beta}^*_n)$, say, that harnesses the strength of all~$p$-norms with exponents in~$[2, \infty]$ in the sense that it has asymptotic size~$\alpha \in (0, 1)$ and is consistent against a sequence of alternatives whenever some $p$-norm based test with $p \in [2, \infty]$ and asymptotic size in $(0, 1)$ is consistent against that sequence of alternatives. Thus,~$\psi_n(\bm{\beta}^*_n)$ simultaneously dominates all~$p$-norm based tests in terms of consistency. In particular, it goes beyond current implementations of the power enhancement principle, which typically target consistency against alternatives that either tests based on the~$2$- or~$\infty$-norm are consistent against, i.e.,~only~\emph{two} norms are targeted.

Finally, we illustrate the scope of our results in the running example of the linear instrumental variable model with many (weak) instruments. Here our general construction yields a test that is powerful irrespectively of whether the first-stage is sparse, dense, or semi-sparse.

\section{The problem}\label{sec:problem}
We observe realizations of random variables~$\bm{X}_{1,n},\hdots,\bm{X}_{n,n}$, each of which is defined on an underlying probability space~$(\Omega,\mc{F},\P)$ and takes its values in the measurable space~$(\mc{X}_{n},\mc{A}_n)$. We wish to test whether a candidate structural parameter~$\bm{\beta}^*_n \in \mathbf{B}_n \neq \emptyset$ satisfies the \emph{moment equalities}
\begin{align}\label{eq:GMMind}
\E h_{j,n}(\bm{X}_{1,n},\bm{\beta}^*_n)=0,\qquad j=1,\hdots,d(n), 
\end{align}
for~$h_{j,n}:\mc{X}_{n}\times \mathbf{B}_n\to\R$ with~$h_{j,n}(\cdot,\bm{\beta}^*_n)$ being~$\P$-integrable and~$\E$ denoting the expectation with respect to~$\P$. The~$\bm{X}_{i,n}$ are assumed to be identically distributed throughout, but are allowed to be dependent (although our detailed examples are for independent variables).
\emph{In this paper we are interested in a setting wherein~$d(n)\to\infty$ as the sample size~$n\to\infty$}.  That is, there are ``many" moment equalities. To conserve on notation, we write~$d=d(n)$ whenever this causes no confusion. For our general results, no structure needs to be imposed on the parameter space~$\mathbf{B}_n$. Of course, the complexity of constructing confidence sets for the parameters via test inversion depends on the size and structure of~$\mathbf{B}_n$. 

For~$h_{n} = (h_{1,n},\hdots,h_{d,n})'$ and~$\bm{0}_{d}$ denoting the~$d\times 1$ vector of zeros, the requirement in~\eqref{eq:GMMind} is more conveniently expressed as
\begin{align*}
\E h_{n}(\bm{X}_{1,n},\bm{\beta}^*_n)=\bm{0}_{d}.
\end{align*}
We thus consider the testing problem
\begin{align}\label{eq:test_problem}
H_{0,n}:\E h_{n}(\bm{X}_{1,n},\bm{\beta}^*_n)=\bm{0}_{d}\quad \text{against}\quad H_{1,n}:\E h_{n}(\bm{X}_{1,n},\bm{\beta}^*_n)\neq \bm{0}_{d},
\end{align}
and are interested in asymptotic size and power properties of tests in situations in which~$d$ increases with~$n$. Finally, we stress that since the tests to be constructed build on plugging in a candidate~$\bm{\beta}^*_n$, no assumptions need to be made on (the degree of) identification of the parameter(s) satisfying~\eqref{eq:GMMind}. Thus, the tests are trivially robust to (weak) identification problems. 

\subsection{Examples}\label{sec:examples}
As pointed out already in~\cite{hansen1982generalized}, moment conditions as in~\eqref{eq:GMMind} frequently arise in economic models as first-order conditions to an agent's optimization problem.  Furthermore, as mentiond in the introduction, the true population parameters in~$M$- and $Z$-estimation problems, such as (non-linear) least squares and maximum likelihood satisfy first-order conditions of the form~\eqref{eq:GMMind}. Our running example will be inference in the presence of many (weak) instruments, and since the 2-norm based test there is just the Anderson-Rubin test, we begin by recalling how this fits into our general framework. We also provide details for treatment effect testing examples. For clarity, we assume i.i.d.~sampling in the examples below.

\subsubsection{Inference in the presence of many (weak) instruments}\label{sec:IVexample} Consider the classic linear instrumental variable (IV) setting in which~$\bm{X}_{i,n}=(y_{i,n},\bm{Y}'_{i,n},\bm{z}'_{i,n})'$ for~$y_{i,n}\in\R$ an outcome of interest,~$\bm{Y}_{i,n}\in\R^k$ a vector of endogenous explanatory variables and~$\bm{z}_{i,n}\in\R^{d}$ a vector of instruments. Thus,~$\mc{X}_n=\R^{1+k+d}$,~$\mathbf{B}_n=\R^k$ and~$h_{n}(\bm{X}_{1,n},\bm{\beta}^*_n)=(y_{1,n}-\bm{Y}'_{1,n} \bm{\beta}^*_n)\bm{z}_{1,n}$. In this context, we wish to test whether a candidate parameter vector $\bm{\beta}_n^*$ lies in the identified set, i.e., whether 
\begin{align*}
H_{0,n}:\E\sbr[1]{(y_{1,n}-\bm{Y}'_{1,n} \bm{\beta}^*_n)\bm{z}_{1,n}}=\bm{0}_{d}\qquad\text{against}\qquad H_{1,n}:\E\sbr[1]{(y_{1,n}-\bm{Y}'_{1,n} \bm{\beta}^*_n)\bm{z}_{1,n}}\neq\bm{0}_{d}. 
\end{align*}
Note that the number of moment equalities equals the number of instruments~$d$. This can be large even for~$k$ fixed. Inference on~$\bm{\beta}^*_n$ in the presence of many (weak) instruments has received considerable recent attention as witnessed by, e.g., the works of~\cite{andrews2007testing}, \cite{anatolyev2011specification}, \cite{belloni2012sparse}, \cite{tchuente2016regularization}, \cite{kaffo2017bootstrap}, \cite{crudu2021inference}, \cite{mikusheva2021many},  
\cite{matsushita2021second},
\cite{mikusheva2020inference},  \cite{matsushita2022jackknife}, \cite{dovi2022ridge}, \cite{boot2023identification}. There are several important practical reasons for this interest. First, when identification is weak, researchers may seek to obtain more precise inference by using a large number of instruments in order to capture more of the exogenous variation in the endogenous covariates. Second, approaches such as the granular IV approach of~\cite{gabaix2020granular}, the saturation approach of~\cite{blandhol2022tsls} or Mendelian Randomization (\cite{davey2003mendelian}), the latter using genetic variation as instruments, can lead to situations where the number of instruments is large compared to the sample size. The same is true for technical instruments such as transformations or interactions or empirical strategies such as ``judge designs,'' cf.~\cite{miksun23}. However, as there is no guarantee that even a large number of instruments is jointly informative, it is important to develop tests that remain valid under no or weak identification. 

Our general results allow us to construct a test that strictly dominates the Anderson-Rubin test, the sup-score test, and power enhancement combinations of these in terms of consistency properties when~$d\to\infty$, cf.~Section~\ref{ex:IVdom}. 

\bigskip 

The examples in the following two subsections revisit  Examples 1 and 2 in \cite{belloni2018high} and both consider testing for treatment effects in a randomized controlled trial.

\subsubsection{Randomized controlled trial with many outcomes}\label{sec:exmanyoutcomes}
Consider a randomized controlled trial in which~$\bm{X}_{i,n}=(D_{i,n},\bm{Y}'_{i,n}(0),\bm{Y}'_{i,n}(1))'$ for~$D_{i,n}\in\cbr[0]{0,1}$ indicating whether individual~$i$ is assigned to the treatment ($D_{i,n}=1$) or not ($D_{i,n}=0$). For each individual we observe a large number,~$d$,  of outcome variables and denote this by~$\bm{Y}_{i,n}(1) = (Y_{i1,n}(1), \hdots, Y_{id,n}(1))' \in \R^{d}$  in case individual~$i$ is treated and~$\bm{Y}_{i,n}(0) = (Y_{i1,n}(0), \hdots, Y_{id,n}(0))' \in\R^{d}$ if not. This framework of many outcome variables is relevant, because a researcher may, for example, be interested in whether a medication has an effect on any of a large number of health outcomes (blood pressure, weight, heart rate, ...) rather than only monitoring the effect on a single of these.\footnote{Although a treatment may be targeted to affect a single outcome of interest, it may nevertheless be important to scan for potential side effects on other outcome variables.}

Suppose that~$\P(D_{1,n}=1)=\pi$~for~$\pi\in(0,1)$ a known probability of being assigned to the treatment and that the treatments are assigned independently of the potential outcomes, that is~$D_{1,n} \indep (\bm{Y}_{1,n}(0),\bm{Y}_{1,n}(1))$. Furthermore,~$\mc{X}_n=\R^{1+2d}$. The treatment does not have an effect on any of the outcomes being measured if~$\E \bm{Y}_{1,n}(1)=\E \bm{Y}_{1,n}(0)$. More generally, for $\bm{\beta}^*_n = (\beta^*_{n,1}, \hdots, \beta^*_{n,d})' \in \R^d$ a vector of coordinate-wise effects, we are interested in testing whether $\E \bm{Y}_{1,n}(1) - \E \bm{Y}_{1,n}(0)=\bm{\beta}^*_n$. Because~$D_{1,n} \indep (\bm{Y}_{1,n}(0),\bm{Y}_{1,n}(1))$ this is equivalent to
\begin{align}\label{eq:treatment}
\E \sbr[3]{\frac{D_{1,n}\bm{Y}_{1,n}(1)}{\pi}-\frac{(1-D_{1,n})\bm{Y}_{1,n}(0)}{1-\pi} - \bm{\beta}^*_n}
=
\E\sbr[3]{\frac{D_{1,n}\bm{Y}_{1,n}}{\pi}-\frac{(1-D_{1,n})\bm{Y}_{1,n}}{1-\pi}- \bm{\beta}^*_n}
=
\bm{0}_{d},
\end{align} 
where~$\bm{Y}_{1,n}=D_{1,n}\bm{Y}_{1,n}(1)+(1-D_{1,n})\bm{Y}_{1,n}(0)$ is the observed outcome. This is a testing problem of the form in~\eqref{eq:test_problem} with~$h_{n}(\bm{X}_{1,n},\bm{\beta}^*_n)=\frac{D_{1,n}\bm{Y}_{1,n}}{\pi}-\frac{(1-D_{1,n})\bm{Y}_{1,n}}{1-\pi} - \bm{\beta}^*_n$. 

\subsubsection{Randomized controlled trials with many groups}

Similar to the previous example of Section~\ref{sec:exmanyoutcomes}, consider a randomized controlled trial. However, let there be only one outcome variable for each individual. Suppose, however, that each subject falls into one of $d$ groups and one wants to test whether a treatment has the hypothesized effect (e.g., no effect) in every one of these groups. Assuming for simplicity that there are~$n$ observations available in each group, one can thus again write the problem in the form of two~$d$-dimensional outcome vectors~$\bm{Y}_{i,n}(0)$ and~$\bm{Y}_{i,n}(1)$ as in the previous example ($\bm{Y}_{i,n}(0)$ is the vector of the stacked outcomes of the $i$-th subjects in each group under the no treatment condition, whereas $\bm{Y}_{i,n}(1)$ is the vector of the stacked outcomes of the $i$-th subjects in each group under the treatment condition). The null hypotheses of the treatment having effect $\beta_{n,j}^*$ in group $j$ for $j = 1, \hdots, d$ then again amounts to~\eqref{eq:treatment}.

As pointed out in, e.g.,~\cite{belloni2018high}, the project STAR investigating the effect of class size on student learning falls within the framework of the present example as many background characteristics were collected on pupils and teachers. For example gender or race. For pupils, their month of birth was also recorded. Identifying each group with a particular combination of covariates, this results in a large number of groups.

\section{Characterization of consistency properties of~$p$-norm based tests}

Recall from the introduction that given~$\bm{\beta}_n^* \in \mathbf{B}_n$, we write
\begin{align*}
\bm{H}_{n}(\bm{\beta}^*_n)&:=\frac{1}{\sqrt{n}}\sum_{i=1}^nh_{n}(\bm{X}_{i,n},\bm{\beta}^*_n),
\end{align*}
and let~$\hat{\bm{\Sigma}}_n(\bm{\beta}^*_n)$ be a positive semidefinite and symmetric estimator of the covariance matrix~$\bm{\Sigma}_n(\bm{\beta}^*_n)$ of~$\bm{H}_n(\bm{\beta}^*_n)$, which we assume to exist.
\begin{remark}
Already under independent sampling, the choice of~$\hat{\bm{\Sigma}}_n(\bm{\beta}^*_n)$ may depend on instance-specific additional structural information on~$\bm{\Sigma}_n(\bm{\beta}^*_n)$ (e.g., bandedness, sparsity or factor structure), the incorporation of which can improve the estimator. We are particularly interested in a setting where no structural information on~$\bm{\Sigma}_n(\bm{\beta}^*_n)$ is available and leave the specific choice of estimator unspecified in the following. In Section~\ref{sec:primcond} we highlight how the precision of the chosen~$\hat{\bm{\Sigma}}_n(\bm{\beta}^*_n)$ influences the allowed growth rate of~$d$ and discuss precision guarantees available in the literature. 
\end{remark} 
A common way of measuring the empirical evidence against the null of~$\E h_{n}(\bm{X}_{1,n},\bm{\beta}^*_n)=\bm{0}_{d}$ is based on the the Euclidean norm of~$\hat{\bm{\Sigma}}^{-1/2}_n(\bm{\beta}^*_n)\bm{H}_{n}(\bm{\beta}^*_n)$, that is, one rejects~$H_{0,n}$ null whenever
\begin{align}
S_{n,2}(\bm{\beta}^*_n):=\sqrt{\bm{H}'_{n}(\bm{\beta}^*_n) \hat{\bm{\Sigma}}^{-1}_n(\bm{\beta}^*_n)\bm{H}_{n}(\bm{\beta}^*_n)}
=
\enVert[2]{\hat{\bm{\Sigma}}^{-1/2}_n(\bm{\beta}^*_n)\bm{H}_{n}(\bm{\beta}^*_n)}_2\label{eq:2normstat}
\end{align}
exceeds a critical value~$\kappa_{n,2}$ chosen to ensure that the resulting test has a desired (asymptotic) size~$\alpha\in(0,1)$, and where~$\enVert[0]{\bm{x}}_2=\sqrt{\sum_{i=1}^dx_i^2}$ for~$\bm{x}\in\R^d$.\footnote{In order to avoid taking a stance on~$\hat{\bm{\Sigma}}_n(\bm{\beta}^*_n)$ potentially not being invertible, we denote by~$\bm{A}^{-1}$ the Moore-Penrose pseudoinverse for any matrix~$\bm{A}$ and define~$\bm{A}^{-1/2}:=\del[1]{\bm{A}^{1/2}}^{-1}$ in case~$\bm{A}$ is symmetric and positive semidefinite, and where~$\bm{A}^{1/2}$ denotes the unique symmetric positive semidefinite square root of~$\bm{A}$. Recall that the Moore-Penrose inverse is identical to the regular matrix inverse whenever the latter exists.\label{fn:MPinv}} For~$\bm{x} = (x_1, \hdots, x_d)' \in \R^d$ and~$p \in [2, \infty]$, define the~$p$-norm
\begin{equation*}
\|\bm{x}\|_p := 
\begin{cases}
\left(\sum_{i = 1}^d |x_i|^p\right)^{\frac{1}{p}} & \text{if } p < \infty, \\
\max_{i = 1, \hdots, d} |x_i| & \text{else}
\end{cases}
\end{equation*}
and introduce, analogously to~$S_{n,2}(\bm{\beta}^*_n)$ in~\eqref{eq:2normstat} but based on the $p$-norm, the family of test statistics
\begin{align*}
S_{n,p}(\bm{\beta}^*_n):=\enVert[2]{\hat{\bm{\Sigma}}^{-1/2}_n(\bm{\beta}^*_n)\bm{H}_{n}(\bm{\beta}^*_n)}_p,\qquad p\in[2,\infty].
\end{align*}
For sequences of critical values~$(\kappa_{n,p})_{n\in\N}$ that guarantee a desired asymptotic size, we shall now study consistency properties of the~\emph{$p$-norm based} tests
\begin{equation*}
\mathds{1}(S_{n,p}(\bm{\beta}^*_n)\geq \kappa_{n,p})
\end{equation*}
and coverage properties of the associated confidence sets
\begin{equation*} \cbr[1]{\bm{\beta}_n \in\mathbf{B}_n:S_{n,p}(\bm{\beta}_n)\leq \kappa_{n,p}},
\end{equation*}
when~$d\to\infty$ as $n \to \infty$. Define $$\mathbf{B}^* := \bigtimes_{n = 1}^{\infty} \mathbf{B}_n,$$ i.e., the set of possible \emph{sequences} of parameters $\bm{\beta}_n$ for $n \in \N$, along which power or coverage properties of tests for the sequence of testing problems in~\eqref{eq:test_problem} can be studied asymptotically in the current setup. In particular, we write~$\bm{\beta}^*=(\bm{\beta}_1^*,\bm{\beta}_2^*,\hdots)$. Let $$\mathbf{B}^{(0)}=\cbr[1]{\bm{\beta} = (\bm{\beta}_1, \bm{\beta}_2, \hdots) \in\mathbf{B}^*:\E h_{n}(\bm{X}_{1,n},\bm{\beta}_n)=\bm{0}_{d}\text{ for every }n\in\N} \subseteq \mathbf{B}^*$$ be the set of sequences of parameters satisfying~$H_{0,n}$ in~\eqref{eq:test_problem} for every $n \in \N$. 

As we shall see, the scaled deviations from the null hypothesis
\begin{align}\label{eq:thetadef}
\bm{\theta}_n(\bm{\beta}^*_n)=\sqrt{n}\bm{\Sigma}^{-1/2}_{n}(\bm{\beta}^*_n)\E h_{n}(\bm{X}_{1,n},\bm{\beta}^*_n) \in \R^d,
\end{align}
play the role of a noncentrality parameter in characterizing the consistency properties of tests based on~$S_{n,p}(\bm{\beta}^*_n)$. We also let~$\sigma_p^2 := \mathbb{V}ar(|Z|^p)$  with~$Z\sim\mathsf{N}_1(0,1)$ and define the functions
\begin{align}\label{eqn:rhodef}
\lambda_p(x)=\E|Z+x|^p\qquad\text{and}\qquad g_p(x) =x^2\vee |x|^p,\qquad x\in\R,\ p\in[2,\infty).
\end{align}
\subsection{Power properties of~$p$-norm based tests:~$p\in[2,\infty)$}
We first consider the case of~$p\in[2,\infty)$ as the characterization of their power has a common structure. The case of~$p=\infty$ is then covered in Section~\ref{sec:suptest}.

Our most general characterization of the power properties of~$p$-norm based tests is carried out under the following high-level conditions for which primitive conditions are provided in Section~\ref{sec:primcond}.
\begin{assumption}[Size]\label{ass:pnormsize}
$\bm{\beta}^* \in \mathbf{B}^{(0)}$ 
and
\begin{align}\label{eq:pnormsize}
\frac{S_{n,p}^p(\bm{\beta}_n^*)-d \lambda_p(0)}{\sqrt{d} \sigma_p}
\rightsquigarrow \mathsf{N}_1(0, 1).
\end{align}
\end{assumption}
Note that~$d \lambda_p(0)$ appearing in Assumption~\ref{ass:pnormsize} is not the expected value of~$S_{n,p}^p(\bm{\beta}_n^*)$, which, in fact, may not even exist. Section~\ref{sec:primcond} provides conditions under which Assumption~\ref{ass:pnormsize} is satisfied for \emph{all} powers~$p\in[2,\infty)$ even when the~$h_{j,n}(\bm{X}_{1,n},\bm{\beta}_n^*)$ have only four moments.

\begin{assumption}[Power]\label{ass:pnormconsistency} 
For $\bm{\beta}^*  \in\mathbf{B}^*$ we summarize the following two properties:
\begin{enumerate}
\item If~$\frac{1}{\sqrt{d'}} \sum_{i = 1}^{d'} [\lambda_p(\theta_{i,n'}(\bm{\beta}_n^*))-\lambda_p(0)]$ is bounded for a subsequence~$n'$ of~$n$ and where~$d'=d(n')$, then 
\begin{align}\label{eq:pnormlocalhighlevel}
\frac{S_{n',p}^p(\bm{\beta}_{n'}^*)-\sum_{i = 1}^{d'} \lambda_p(\theta_{i,n'}(\bm{\beta}_{n'}^*))}{\sqrt{d'} \sigma_p}
\rightsquigarrow \mathsf{N}_1(0, 1).
\end{align}
\item If~$\frac{1}{\sqrt{d}} \sum_{i = 1}^{d} [\lambda_p(\theta_{i,n}(\bm{\beta}_n^*))-\lambda_p(0)]\to \infty$, then
\begin{align}\label{eq:pnormconshighlevel}
\frac{S_{n,p}^p(\bm{\beta}_n^*)-\sum_{i=1}^{d}\lambda_p(\theta_{i,n}(\bm{\beta}_n^*))}{\sum_{i = 1}^{d} [\lambda_p(\theta_{i,n}(\bm{\beta}_n^*))-\lambda_p(0)]}
=
o_{\P}(1).
\end{align}
\end{enumerate}
\end{assumption}

The following theorem characterizes asymptotic size and power properties of $p$-norm based tests for $p \in [2, \infty)$. We denote the cdf of a standard normal distribution by~$\Phi$.
\begin{theorem}\label{thm:pnorm-char_GMM}
Let~$p\in[2,\infty)$ and $\alpha\in(0,1)$.  
\begin{enumerate}
\item Size control: Under Assumption~\ref{ass:pnormsize}, a sequence of real numbers~$(\kappa_{n,p})_{n\in\N}$ satisfies
\begin{align}\label{eq:sizepnorm}
\P\del[1]{S_{n,p}(\bm{\beta}_n^*)\geq \kappa_{n,p}}\to\alpha, 
\end{align}
if and only if~$\kappa_{n,p}=\kappa_{n,p}(\alpha)=\sbr[1]{(\Phi^{-1}(1-\alpha)+o(1))d^{1/2}\sigma_p+d\lambda_p(0)}^{1/p}$.
\item Local power: If Part 1 of Assumption \ref{ass:pnormconsistency} holds and if~$$\frac{1}{\sqrt{d}} \sum_{i = 1}^{d} [\lambda_p(\theta_{i,n}(\bm{\beta}_n^*))-\lambda_p(0)]\to c\in[0,\infty),$$ then, for~$(\kappa_{n,p})_{n\in\N}$ satisfying~\eqref{eq:sizepnorm},
\begin{align*}
\P\del[1]{S_{n,p}(\bm{\beta}_n^*)\geq \kappa_{n,p}}\to 1-\Phi(\Phi^{-1}(1-\alpha)-c/\sigma_p).	
\end{align*}
\item Consistency: Under Assumption~\ref{ass:pnormconsistency} and for~$(\kappa_{n,p})_{n\in\N}$ satisfying~\eqref{eq:sizepnorm}, it holds that
\begin{align}
\P\del[1]{S_{n,p}(\bm{\beta}_n^*)\geq \kappa_{n,p}}\to 1 
\quad 
&\Longleftrightarrow \quad \frac{\sum_{i=1}^{d}[\lambda_p(\theta_{i,n}(\bm{\beta}_n^*))-\lambda_p(0)]}{\sqrt{d}}\to\infty \label{eq:conscharlambda}\\
&\Longleftrightarrow	 \quad\frac{\sum_{i=1}^dg_p(\theta_{i,n}(\bm{\beta}_n^*))}{\sqrt{d}}\to\infty.\label{eq:conschargp}
\end{align}
\end{enumerate}
\end{theorem}

Part~1 of Theorem~\ref{thm:pnorm-char_GMM} characterizes sequences of critical values that yield asymptotic size control. For any~$\alpha\in(0,1)$, a canonical choice is~$$\kappa_{n,p}=\sbr[0]{\Phi^{-1}(1-\alpha)d^{1/2}\sigma_p+d\lambda_p(0)}^{1/p}.$$
Part 2 provides results on local asymptotic power and Part 3 provides a complete characterization of the alternatives that a~$p$-norm based test is consistent against. Apart from when~$p=2$, neither of these can be expressed in terms of only the~$p$-norm of the (scaled) deviation from the null~$\|\bm{\theta}_n(\bm{\beta}_n^*)\|_p$; instead they depend on the asymptotic behaviour of~$$d^{-1/2}\sum_{i = 1}^{d} [\lambda_p(\theta_{i,n}(\bm{\beta}_n^*))-\lambda_p(0)].$$ Concerning consistency, the divergence of the latter is equivalent to divergence of $$\frac{\sum_{i=1}^dg_p(\theta_{i,n}(\bm{\beta}_n^*))}{\sqrt{d}}=\frac{\sum_{i=1}^d\sbr[1]{\theta_{i,n}^2(\bm{\beta}_n^*)\vee |\theta_{i,n}(\bm{\beta}_n^*)|^p}}{\sqrt{d}},$$ which is somewhat easier to interpret.

Concerning the consistency of a~$p$-norm based test (with asymptotic size in~$(0, 1)$),~\eqref{eq:conschargp} implies that a sufficient condition is that~$d^{-1/2}\|\bm{\theta}_n(\bm{\beta}_n^*)\|_2^2\to\infty$ \emph{or}~$d^{-1/2}\|\bm{\theta}_n(\bm{\beta}_n^*)\|_p^p\to\infty$, cf.~also the previous display. What is more, observing that
\begin{align*}
g_p(x)\leq g_q(x)\qquad\text{for}\qquad 2\leq p<q<\infty\text{ and all }x\in\R,
\end{align*}
Part 3~of Theorem~\ref{thm:pnorm-char_GMM} shows that
\begin{align}\label{eq:dom}
\P\del[1]{S_{n,p}(\bm{\beta}_n^*)\geq \kappa_{n,p}}\to 1\qquad \text{implies} \qquad \P\del[1]{S_{n,q}(\bm{\beta}_n^*)\geq \kappa_{n,q}}\to 1
\end{align}
when~$(\kappa_{n,p})_{n\in\N}$ and~$(\kappa_{n,q})_{n\in\N}$ are chosen such that the corresponding tests have asymptotic sizes in~$(0,1)$. In words,~\emph{any violation of the moment conditions~$\E h_{n}(\bm{X}_{1,n},\bm{\beta}_n^*)=\bm{0}_{d}$ that a test based on the~$p$-norm is consistent against, a test based on the~$q$-norm will also be consistent against},~$2\leq p<q<\infty$. Thus, the~$q$-norm based test always weakly dominates the~$p$-norm based test in terms of consistency. This domination is \emph{strict} if and only if there exists a~$\bm{\beta}^*\in\mathbf{B}^*$ such that $\liminf_{n\to\infty}\P\del[1]{S_{n,p}(\bm{\beta}_n^*)\geq \kappa_{n,p}}<1$ but~$\P\del[1]{S_{n,q}(\bm{\beta}_n^*)\geq \kappa_{n,q}}\to 1$, i.e., if and only if (cf.~\eqref{eq:conschargp})
\begin{align}\label{eq:highlevchar}
\liminf_{n\to\infty}\frac{\sum_{i=1}^{d}g_p(\theta_{i,n}(\bm{\beta}_n^*))}{\sqrt{d}}<\infty\qquad \text{and}\qquad
\frac{\sum_{i=1}^{d}g_q(\theta_{i,n}(\bm{\beta}_n^*))}{\sqrt{d}}\to\infty.
\end{align}

Summarizing, given the conditions imposed in Theorem \ref{thm:pnorm-char_GMM}, the characterization and monotonicity properties of the consistency of~$p$-norm based tests established in \cite{kp2021consistency} in the Gaussian sequence model continue to hold in the current setting.

By Theorem 3.4 in~\cite{kp2021consistency} the~$\bm{\theta}_n(\bm{\beta}_n^*)$ that satisfy~\eqref{eq:highlevchar} are necessarily approximately sparse and unbalanced. Furthermore, if the limit inferior to the left in~\eqref{eq:highlevchar} is zero, then there exists a subsequence along which the~$p$-norm based test has asymptotic power equal to its size, while the~$q$-norm based test is consistent --- a substantial difference in power.

\subsection{Primitive conditions for Assumptions \ref{ass:pnormsize} and \ref{ass:pnormconsistency}}\label{sec:primcond}
We next provide conditions that will be shown to be sufficient for Assumptions~\ref{ass:pnormsize} and~\ref{ass:pnormconsistency} to be satisfied and hence for Theorem~\ref{thm:pnorm-char_GMM} to apply. To this end, let~$\bm{Z}_d\sim\mathsf{N}_d(\bm{0}_d,\mathbf{I}_d)$,~$$\mc{C}_n=\cbr[0]{C\subseteq \R^{d(n)}: C\text{ is convex and Borel measurable}},$$ and~$\|\bm{A}\|_2$ be the spectral norm of the matrix~$\bm{A}$. 
\begin{assumption}\label{as:suffc}
Let~$\bm{\beta}^* \in \mathbf{B}^*$, assume that~$d\to\infty$, that the eigenvalues of~$\bm{\Sigma}_n(\bm{\beta}_n^*)$ are (uniformly) bounded away from zero and from above, and that
\begin{enumerate}
	\item $\|\hat{\bm{\Sigma}}_n(\bm{\beta}_n^*)-\bm{\Sigma}_n(\bm{\beta}_n^*)\|_2=O_\P(a_n)$ with~$d^{3/4}a_n\to 0$, and
	\item $\sup_{C\in\mc{C}_n}\envert[1]{\P\del[1]{\bm{\Sigma}_n^{-1/2}(\bm{\beta}_n^*)\bm{H}_n(\bm{\beta}_n^*)\in  C}-\P\del[1]{\bm{Z}_d+\bm{\theta}_n(\bm{\beta}_n^*) \in  C}}\to 0$.
\end{enumerate}
\end{assumption}
Note that these conditions do not impose independence. 

\begin{theorem}\label{thm:prim}
Suppose Assumption~\ref{as:suffc} is satisfied. Then, the following holds:
\begin{itemize}
	\item If $\bm{\beta}^* \in \mathbf{B}^{(0)}$, then Assumption~\ref{ass:pnormsize} holds for all~$p\in[2,\infty)$.
	\item Assumption \ref{ass:pnormconsistency} holds for all~$p\in[2,\infty)$.
\end{itemize}
\end{theorem}
\begin{remark}
The requirement~$d^{3/4}a_n\to 0$ in Part 1 of Assumption~\ref{as:suffc} can be relaxed to~$d^{3/4\wedge (1-1/p)}a_n\to 0$, which is milder for~$p\in[2,4)$, but we omit a formal statement.
\end{remark}
Theorem~\ref{thm:prim} shows that for Theorem~\ref{thm:pnorm-char_GMM} to apply it suffices to upper bound the estimation error of the covariance matrix, $\|\hat{\bm{\Sigma}}_n(\bm{\beta}_n^*)-\bm{\Sigma}_n(\bm{\beta}_n^*)\|_2$, and a Gaussian approximation to hold.\footnote{We impose the Gaussian approximation to hold over convex sets as we study all~$p$-norm based tests simultaneously. If one is only interested in a single~$p$, it may be enough for the  approximation to hold over~$p$-norm balls with arbitrary centres and radii for this given $p$ --- potentially allowing for~$d$ to grow faster.} As discussed below, for both of these requirements catalogues of sufficient conditions exist in the literature.

To prove Theorem~\ref{thm:prim} we show in Lemma~\ref{lem:convgauss} that when~$d^{-1}\sum_{i=1}^dg_2(\theta_{i,n}(\bm{\beta}_n^*))$ is bounded, a Gaussian approximation property holds even when~$\bm{\Sigma}_n(\bm{\beta}_n^*)$ is replaced by~$\hat{\bm{\Sigma}}_n(\bm{\beta}_n^*)$:
\begin{align}\label{eq:newGaussApprox}
\sup_{C\in\mc{C}_n}\envert[1]{\P\del[1]{\hat{\bm{\Sigma}}_n^{-1/2}(\bm{\beta}_n^*)\bm{H}_n(\bm{\beta}_n^*)\in  C}-\P\del[1]{\bm{Z}_d+\bm{\theta}_n(\bm{\beta}_n^*) \in  C}}\to 0.
\end{align}
This essentially allows us to deduce the desired properties of~$S_{n,p}(\bm{\beta}_n^*)=\|\hat{\bm{\Sigma}}_n^{-1/2}(\bm{\beta}_n^*)\bm{H}_n(\bm{\beta}_n^*)\|_p$ in Assumptions \ref{ass:pnormsize} and \ref{ass:pnormconsistency} from the ones of~$p$-norms of the high-dimensional Gaussian mean shift~$\bm{Z}_d+\bm{\theta}_n(\bm{\beta}_n^*)$ in the ``limit experiment'', the latter of which was analyzed in~\cite{kp2021consistency}.\footnote{This is done when~$d^{-1}\sum_{i=1}^dg_2(\theta_{i,n}(\bm{\beta}_n^*))=d^{-1}\|\bm{\theta}_n(\bm{\beta}_n^*)\|_2^2$ is bounded. In regimes wherein this is not the case, arguments tailored to those settings are used to verify~\eqref{eq:pnormconshighlevel} of Assumption~\ref{ass:pnormconsistency}.} More broadly,~\eqref{eq:newGaussApprox} may be of independent interest beyond the scope of this paper as it allows one to deduce probabilistic properties of the scaled ``central statistic''~$\hat{\bm{\Sigma}}_n^{-1/2}(\bm{\beta}_n^*)\bm{H}_n(\bm{\beta}_n^*)$ from those of~$\bm{Z}_d+\bm{\theta}_n(\bm{\beta}_n^*)$ even when~$d\to\infty$. We next discuss Parts 1~and 2~of  Assumption \ref{as:suffc} in further detail. 

Concerning 1~of Assumption~\ref{as:suffc}, it is well-known that if the vectors~$h_n(\bm{X}_{i,n},\bm{\beta}_n^*)$ are i.i.d.~and appropriately sub-Gaussian, then the empirical covariance matrix~$\hat{\bm{\Sigma}}_{n,\text{emp}}(\bm{\beta}_n^*)$ satisfies~$\|\hat{\bm{\Sigma}}_{n,\text{emp}}(\bm{\beta}_n^*)-\bm{\Sigma}_n(\bm{\beta}_n^*)\|_2=O_\P(\sqrt{d/n})$, cf., e.g.,~\cite{vershynin2012close} or \cite{koltchinskiilounici2017}. If, in addition,~$\bm{\Sigma}_n(\bm{\beta}_n^*)$ possesses further structure, this can be utilised to construct  estimators with even better performance guarantees. For example,~\cite{bickel2008covariance} showed that if~$\bm{\Sigma}_n(\bm{\beta}_n^*)$ is sparse,~$\hat{\bm{\Sigma}}_{n,\text{emp}}(\bm{\beta}_n^*)$ can be thresholded to yield an estimator satisfying~$\|\hat{\bm{\Sigma}}_n(\bm{\beta}_n^*)-\bm{\Sigma}_n(\bm{\beta}_n^*)\|_2=O_\P(\sqrt{\log(d)/n})$. The overviews in~\cite{fan2016overview} and \cite{cai2016estimating} list further estimators utilising structural properties of~$\bm{\Sigma}_n(\bm{\beta}_n^*)$ and their performance guarantees.

One is often not willing to assume that the vectors~$h_n(\bm{X}_{i,n},\bm{\beta}_n^*)$ are sub-Gaussian. In particular, much effort is currently being devoted to the construction of covariance matrix estimators with the sub-Gaussian performance guarantee~$\|\hat{\bm{\Sigma}}_n(\bm{\beta}_n^*)-\bm{\Sigma}_n(\bm{\beta}_n^*)\|_2=O_\P(\sqrt{d/n})$ even when only \emph{four} moments exist. Although the sample covariance matrix does not work well under such heavy tails,\footnote{This is in analogy to the sub-optimality of the sample average under heavy tails in the one-dimensional mean estimation problem, cf.~the overview in~\cite{lugosi2019mean}.} the existence of estimators satisfying~$\|\hat{\bm{\Sigma}}_{n}(\bm{\beta}_n^*)-\bm{\Sigma}_n(\bm{\beta}_n^*)\|_2=O_\P(\sqrt{d/n})$ --- without imposing structure on~$\bm{\Sigma}_n(\bm{\beta}_n^*)$ --- has been established in \cite{abdalla2022covariance} and~\cite{oliveira2022improved}, cf.~also~\cite{mendelson2020robust} for  bounds containing an extra factor~$\sqrt{\log(d)}$. These results are proven under the assumption that the kurtosis of all one-dimensional marginals of~$h_n(\bm{X}_{1,n}, \bm{\beta}_n^*)$ is bounded (in particular only four moments need to exist):
\begin{assumption}\label{ass:boundedkurt}
Denote~$\bm{\mu}_n(\bm{\beta}_n^*)=\E h_n(\bm{X}_{1,n},\bm{\beta}_n^*)$ and let~$\langle \cdot,\cdot\rangle$ be the standard inner product in~$\R^d$. There exists a real number~$L(\bm{\beta}^*)\geq 1$, such that for all~$t\in\R^d$ and every $n \in \N$
\begin{align*}
\del[1]{\E \langle [h_n(\bm{X}_{1,n},\bm{\beta}_n^*)-\bm{\mu}_n(\bm{\beta}_n^*)],t\rangle^4}^{\frac{1}{4}}
\leq  
L(\bm{\beta}^*)\del[1]{\E \langle [h_n(\bm{X}_{1,n},\bm{\beta}_n^*)-\bm{\mu}_n(\bm{\beta}_n^*)],t\rangle^2}^{\frac{1}{2}}.
\end{align*}
\end{assumption}
As pointed out in~\cite{mendelson2020robust}, Assumption~\ref{ass:boundedkurt} is satisfied if, e.g.,~$h_n(\bm{X}_{1,n},\bm{\beta}_n^*)-\bm{\mu}_n(\bm{\beta}_n^*)$ follows a multivariate $t$-distribution with~$\nu>4$ degrees of freedom. This example is one of rather heavy tails as only moments strictly lower than~$\nu$ exist. 

Concerning Part 2~of Assumption \ref{as:suffc}, note that by translation invariance of convex sets this is equivalent to
\begin{align*}
\sup_{C\in\mc{C}_n}\envert[3]{\P\del[2]{n^{-1/2}\sum_{i=1}^n\bm{\Sigma}_n^{-1/2}(\bm{\beta}_n^*)[h_n(\bm{X}_{i,n},\bm{\beta}_n^*)-\bm{\mu}_n(\bm{\beta}_n^*)] \in  C}-\P\del[1]{\bm{Z}_d \in  C}}\to 0;
\end{align*}
that is (scaled) partial sums in~$\R^d$ with mean~$\bm{0}_d$ and identity covariance matrix~$\mathbf{I}_d$ should obey a Gaussian approximation. We show that by~\cite{fang2020large}  this holds under i.i.d.~sampling and Assumption \ref{ass:boundedkurt} (or, slightly weaker, bounded fourth moments) if, up to logarithmic factors in~$n$,~$d/n^{2/5}\to 0$. 

Summarizing the above discussion results in the following statement.
\begin{corollary}\label{cor:fourmoments}
Let~$\bm{\beta}^* \in \mathbf{B}^*$, let~$d\to\infty$, and suppose that the eigenvalues of~$\bm{\Sigma}_n(\bm{\beta}_n^*)$ are (uniformly) bounded away from zero and from above. Furthermore, let~$\bm{X}_{1,n},\hdots,\bm{X}_{n,n}$ be i.i.d.~for each~$n\in\N$. Suppose~$\frac{d}{n^{2/5}}[\log(n)]^{2/5}\to 0$ and one uses the estimator~$\hat{\bm{\Sigma}}_n(\bm{\beta}_n^*)$ from \cite{abdalla2022covariance} (with their $\eta = 0$, $\delta = 1/d$ and $p = 4$) or~\cite{oliveira2022improved} (with their $\eta = 0$, $\alpha = 1/d$ and $p = 4$) applied to the auxiliary sample 
\begin{equation*}
\del[1]{h_{n}(\bm{X}_{2,n},\bm{\beta}_n^*)-h_{n}(\bm{X}_{1,n},\bm{\beta}_n^*)}/\sqrt{2},\hdots,\del[1]{(h_{n}(\bm{X}_{2\lfloor n/2 \rfloor,n},\bm{\beta}_n^*)-h_{n}(\bm{X}_{2 \lfloor n/2 \rfloor-1,n},\bm{\beta}_n^*)}/\sqrt{2}.
\end{equation*} 
If Assumption~\ref{ass:boundedkurt} holds, then Assumption~\ref{as:suffc} is satisfied. 
\end{corollary}
Under the assumptions of	 Corollary~\ref{cor:fourmoments}, Theorem~\ref{thm:pnorm-char_GMM} provides a characterization of the consistency properties of~$p$-norm based tests for all~$p\in[2,\infty)$. In particular, even though the~$h_{j,n}(\cdot,\bm{\beta}_n^*)$ only possess four moments, the characterization also holds for~$p$-norm based tests with~$p>4$; that is the consistency properties of~$p$-norm based tests are characterized even when the~$p$th moments of the entries of~$h_{n}(\bm{X}_{1,n},\bm{\beta}_n^*)$ need not exist.
\begin{remark}
Many recent tests in high-dimensional testing problems are based on the supremum-norm such that Gaussian approximations over the class of hyperrectangles (contained in~$\mc{C}_n$) suffice. This allows~$d$ to increase faster than~$n^{2/5}$. However, even for this smaller class of sets, there exist distributions with bounded fourth moments such that the Gaussian approximation breaks down if~$d$ grows faster than~$n$, cf.~\cite{zhangwu2017} and~\cite{kock2023moment} for precise formulations. In particular, there exist distributions satisfying the null such that for all~$\alpha\in(0,1)$ the rejection frequency of the supremum-norm based test using Gaussian size~$\alpha$ critical values tends to \emph{one}. Thus, not even supremum-norm based tests control size uniformly over distributions with bounded fourth moments beyond~$d=n$. In Section~\ref{sec:suptest} we characterize the consistency properties of supremum-norm based tests and in Section~\ref{sec:rapidgrowth} we show how a sample split can be used to accommodate any growth rate of~$d$.
\end{remark}
\begin{remark}
As acknowledged by~\cite{abdalla2022covariance} and~\cite{oliveira2022improved}, their estimators of the population covariance matrix are targeting the best possible statistical performance guarantees at the expense of not being practical to implement. ``User-friendly'' estimators, losing logarithmic factors in~$d$ in the performance guarantees, have been proposed in~\cite{ke2019user}.
\end{remark}

\subsection{The size of confidence sets based on~$p$-norms for~$p\in[2,\infty)$}\label{sec:confsets}
Theorem~\ref{thm:pnorm-char_GMM} also has implications for the coverage properties and size of confidence sets on the form
\begin{align*}
\mathsf{CS}_{n,p,1-\alpha_p}:=\cbr[1]{\bm{\beta}\in\mathbf{B}_n:S_{n,p}(\bm{\beta})\leq \kappa_{n,p}(\alpha_p)}\qquad\text{for }p\in[2,\infty)\text{ and }\alpha_p\in(0,1).
\end{align*}
Under the assumptions of Part 1~of Theorem~\ref{thm:pnorm-char_GMM}, these confidence sets have the desired (uniform) asymptotic coverage guarantee. That is, for any~$\bm{\beta}^*\in\mathbf{B}^{(0)}$, it holds that 
\begin{align*}
\lim_{n\to\infty}\P\del[1]{\bm{\beta}_n^*\in \mathsf{CS}_{n,p,1-\alpha_p}}= 1-\alpha_p.
\end{align*}
Furthermore, under the assumptions of Part 3~of the same theorem and~$2\leq p<q<\infty$, one has for any~$\bm{\beta}^*$ (not satisfying the moment conditions) that
\begin{align}\label{eq:confweak}
\lim_{n\to\infty}\P\del[1]{\bm{\beta}_n^*\in \mathsf{CS}_{n,p,1-\alpha_p}}= 0\qquad\text{implies}\qquad \lim_{n\to\infty}\P\del[1]{\bm{\beta}_n^*\in \mathsf{CS}_{n,q,1-\alpha_q}}= 0;
\end{align}
that is any~$\bm{\beta}^*$ that is excluded from~$\mathsf{CS}_{n,p,1-\alpha_p}$ with probability tending to one will also be excluded from~$\mathsf{CS}_{n,q,1-\alpha_q}$ with probability tending to one. Also, if~$\bm{\beta}^*$ satisfies the two conditions in~\eqref{eq:highlevchar}, with the limit inferior being zero, then
\begin{align}\label{eq:confstrong}
\limsup_{n\to\infty}\P\del[1]{\bm{\beta}_n^*\in \mathsf{CS}_{n,p,1-\alpha_p}}\to 1-\alpha_p\qquad\text{yet}\qquad \lim_{n\to\infty}\P\del[1]{\bm{\beta}_n^*\in \mathsf{CS}_{n,q,1-\alpha_q}}= 0.
\end{align}	
In words,~$\bm{\beta}^*$ then does not satisfy the moment conditions, yet it is included in the confidence set~$\mathsf{CS}_{n,p,1-\alpha_p}$ with the same asymptotic probability as those parameters that do, whereas~$\mathsf{CS}_{n,q,1-\alpha_q}$ correctly excludes this~$\bm{\beta}^*$ asymptotically.  Thus,~\eqref{eq:confweak} and~\eqref{eq:confstrong} give a sense in which basing confidence sets on larger~$p\in[2,\infty)$ results in ``smaller'' confidence sets. However, using results in~\cite{pinelis2010asymptotic}, it is not difficult to show that already in the high-dimensional Gaussian location model one can have that~$\E h_{n}(\bm{X}_{1,n},\bm{\beta}_n^*)\neq\bm{0}_{d}$ and
\begin{align*}
0<\lim_{n\to\infty}\P\del[1]{\bm{\beta}_n^*\in \mathsf{CS}_{n,p,1-\alpha_p}}< \lim_{n\to\infty}\P\del[1]{\bm{\beta}_n^*\in\mathsf{CS}_{n,q,1-\alpha_q}}.
\end{align*}
Thus, it is not generally the case that eventually~$\mathsf{CS}_{n,q,1-\alpha_q}\subseteq \mathsf{CS}_{n,p,1-\alpha_p}$.

\subsection{Tests based on the supremum-norm:~$p=\infty$}\label{sec:suptest}
Tests based on the supremum-norm have received considerable attention due to the work on high-dimensional Gaussian approximations over hyperrectangles
\begin{align*}
\mathcal{H}_n=\cbr[3]{\prod_{j=1}^{d}[a_j,b_j] \cap \R :\ -\infty\leq a_j\leq b_j\leq \infty,\ j=1,\hdots,d}.
\end{align*}	
The following theorem exactly characterizes which alternatives tests based on the supremum-norm are consistent against under the same assumptions as those of Theorem~\ref{thm:prim}.
\begin{theorem}\label{thm:supnorm}
Let~$\alpha\in(0,1)$ and suppose Assumption~\ref{as:suffc} is satisfied. Then, the following holds.
\begin{enumerate}
\item[a)] Size control: Suppose~$\bm{\beta}^*\in \mathbf{B}^{(0)}$. A sequence of real numbers~$(\kappa_{n,\infty})_{n\in\N}$ satisfies 
\begin{align}\label{eq:sizesup}
\P\del[1]{S_{n,\infty}(\bm{\beta}_n^*)\geq \kappa_{n,\infty}}\to \alpha
\end{align}
if and only if~$\kappa_{n,\infty}= \kappa_{n,\infty}(\alpha)= \sqrt{2\log(d)}-\frac{\log\log(d)+\log(4\pi)}{2\sqrt{2\log(d)}}-\frac{\log\del[0]{-\log(1-\alpha)/2}+o(1)}{\sqrt{2\log(d)}}$.

\item[b)] Consistency: For~$(\kappa_{n,\infty})_{n\in\N}$ as in~\eqref{eq:sizesup}
\begin{align*}
\P\del[1]{S_{n,\infty}(\bm{\beta}_n^*)\geq \kappa_{n,\infty}}\to 1\quad
\Longleftrightarrow 
\quad 
\sum_{i = 1}^d \frac{\overline{\Phi}\left(\mathfrak{c}_d - |\theta_{i,d}(\bm{\beta}_n^*)|\right)}{\Phi\left(\mathfrak{c}_d - |\theta_{i,d}(\bm{\beta}_n^*)|\right)} \to \infty,
\end{align*}
where~$\mathfrak{c}_d := \sqrt{2\log(d)} - \frac{\log \log(d)}{2 \sqrt{2\log(d)}}$ for~$d\geq 2$ and~$\overline{\Phi} = 1-\Phi$ (and one may set~$\mathfrak{c}_1 := 0$).
\end{enumerate}
\end{theorem}

Well-known \emph{sufficient}, yet not necessary, conditions for consistency of the supremum-norm based test, such as~$\|\bm{\theta}_n(\bm{\beta}_n^*)\|_\infty-\sqrt{2\log(d)}\to\infty$, are trivial special cases of Theorem~\ref{thm:supnorm}. In particular, the $\infty$-norm based test is consistent against~$(\sqrt{3\log(d)},0,\hdots,0)$ whereas by~\eqref{eq:conschargp} of Theorem~\ref{thm:pnorm-char_GMM} no~$p$-norm based test with~$p\in [2,\infty)$ is consistent.

Concerning dense alternatives of the form~$\bm{\theta}_n(\bm{\beta}_n^*)=(c_n,\hdots,c_n)$, it follows that the supremum-norm based test is consistent if and only if~$\sqrt{\log(d)}|c_n|\to\infty$, cf.~Appendix A.4 of~\cite{kp2021consistency}. Thus, the supremum-norm based test is \emph{not} consistent if~$c_n=1/\sqrt{\log(d)}$ whereas by~\eqref{eq:conschargp} of Theorem~\ref{thm:pnorm-char_GMM} every~$p$-norm based test with~$p\in [2,\infty)$ \emph{is} consistent. Hence, i) the monotonicity in~\eqref{eq:dom} does not extend to~$q=\infty$ and ii) no single~$p$-norm based test dominates all others in terms of its consistency properties.

The exact characterization of consistency in Theorem~\ref{thm:supnorm} sheds new light on the frequently used supremum-norm based test for testing moment equalities. It is more informative than the currently available power analysis of the supremum-norm based test through a maximin-lens. The latter provides the sufficient condition for consistency $\|\bm{\theta}_n(\bm{\beta}_n^*)\|_\infty-\sqrt{2\log(d)}\to\infty$, but is silent about the alternatives not satisfying this condition that the supremum-norm based test is nevertheless consistent against. For example, if~$\bm{\theta}_n(\bm{\beta}_n^*)=(1,\hdots,1)'$, then by the observation on dense alternatives in the previous paragraph the supremum-norm based test \emph{is} consistent, yet~$\|\bm{\theta}_n(\bm{\beta}_n^*)\|_\infty-\sqrt{2\log(d)}\not\to\infty$. 

Concerning confidence sets, Theorem \ref{thm:supnorm} shows that for a given~$\alpha\in(0,1)$,~$\kappa_{n,\infty}(\alpha)$ as in~\eqref{eq:sizesup}, and~$\bm{\beta}^* \in\mathbf{B}^{(0)}$, the  sets $\mathsf{CS}_{n,\infty,1-\alpha}:=\cbr[1]{\bm{\beta}\in\mathbf{B}_n:S_{n,\infty}(\bm{\beta})\leq \kappa_{n,\infty}(\alpha)}$ satisfy $\lim_{n\to\infty}\P\del[1]{\bm{\beta}_n^*\in \mathsf{CS}_{n,\infty,1-\alpha}}= 1-\alpha$.

\begin{remark}
To prove Theorem~\ref{thm:supnorm}, Lemma \ref{lem:hypergauss} establishes the following variant of~\eqref{eq:newGaussApprox} which, crucially for our characterization of consistency, holds for~$\bm{\theta}_n(\bm{\beta}_n^*)$ further away from the origin than required for~\eqref{eq:newGaussApprox}:
\begin{align*}
\sup_{H\in\mc{H}_n}\envert[1]{\P\del[1]{\hat{\bm{\Sigma}}_n^{-1/2}(\bm{\beta}_n^*)\bm{H}_n(\bm{\beta}_n^*)\in  H}-\P\del[1]{\bm{Z}_d+\bm{\theta}_n(\bm{\beta}_n^*) \in  H}}\to 0.
\end{align*}	
\end{remark}

\section{A test dominating all~$p$-norm based test for~$p\in[2,\infty]$}\label{sec:domtest}
The discussion following Theorem~\ref{thm:supnorm} revealed that no single~$p$-norm based test is ``best'' in terms of the amount of alternatives it is consistent against. This makes choosing a~$p$ to base a test on very difficult. Even if one \emph{knows} that the deviation from the null~$\E h_n(\bm{X}_{1,n},\bm{\beta}_n^*)$ is sparse, which \emph{could} suggest using a test based on the supremum-norm, sparsity need not be inherited by the noncentrality parameter 	
\begin{align*}
\bm{\theta}_n(\bm{\beta}_n^*)=\sqrt{n}\bm{\Sigma}^{-1/2}_{n}(\bm{\beta}_n^*)\E h_{n}(\bm{X}_{1,n},\bm{\beta}_n^*)
\end{align*}
due to the presence of~$\bm{\Sigma}^{-1/2}_{n}(\bm{\beta}_n^*)$. Therefore,  one often does not know anything about the structure of~$\bm{\theta}_n(\bm{\beta}_n^*)$, which is the actual quantity entering the characterization of the power of~$p$-norm based tests in~Theorems~\ref{thm:pnorm-char_GMM} and~\ref{thm:supnorm}. This underscores the importance of having a test that is powerful irrespective of the unknown structure of~$\bm{\theta}_n(\bm{\beta}_n^*)$.

One could hope that combining tests based on the two ``endpoints''~$p=2$ and~$p=\infty$, e.g.,~by the power enhancement principle of \cite{fan2015power}, of the interval~$[2,\infty]$ results in a test that is consistent whenever \emph{some}~$p$-norm based test for~$p\in[2,\infty]$ is consistent. However, Corollary 4.2 in \cite{kp2021consistency} shows that already in the Gaussian sequence model this is not the case: There exist (semi-sparse) alternatives against which tests based on~$p=2$ and~$p=\infty$ are inconsistent, but against which a test based on any single~$p\in(2,\infty)$ \emph{is} consistent. To construct a test that is consistent whenever a test based on some~$p\in[2,\infty]$ is consistent, we use a construction similar in spirit to the one employed in the Gaussian sequence model in~\cite{kp2021consistency}. A crucial difference is that the present construction explicitly includes the supremum-norm based test rather than using a~$p_n$-norm based test with sufficiently quickly increasing~$p_n$ to dominate the former. The present construction is more convenient as it does not require the analysis of tests for moving~$p_n$, which allows us to reduce the number of approximation steps in the analysis.

\subsection{The dominant test}\label{sec:domtestconstruct}
Let~$\alpha\in(0,1)$ and~$\alpha_2,\alpha_I$ and~$\alpha_\infty$ be non-negative with~$\alpha_2+\alpha_I+\alpha_\infty=\alpha$ and let~$\kappa_{n,2}$ and~$\kappa_{n,\infty}$ satisfy
\begin{align*}
\mathbb{P}\left( \|\bm{Z}_{d} \|_{2} \geq \kappa_{n,2} \right)=\alpha_2\quad\text{and}\quad \mathbb{P}\left( \|\bm{Z}_{d} \|_{\infty} \geq \kappa_{n,\infty} \right)=\alpha_\infty.
\end{align*}
Furthermore, let~$p_n$ be a strictly increasing and unbounded sequence in~$(2, \infty)$ and let~$m_n$ be a non-decreasing and unbounded sequence in~$\N$. Fix an array
\begin{align*}
\mathcal{A}  = \left\{\alpha_{n,p_j} \in (0, 1): n \in \N,~ j = 1, \hdots, m_n\right\}	
\end{align*}
such that
\begin{equation*}
\sum_{j = 1}^{m_n} \alpha_{n,p_j} = \alpha_I \text{ for every } n\in \N \quad\text{and}\quad \lim_{n \to \infty} \alpha_{n,p_j}  > 0 \text{ for every } j \in \mathbb{N},
\end{equation*}
where the conditions implicitly impose the existence of the respective limits. For every~$n \in \N$ and every~$j = 1, \hdots, m_n$, choose~$\kappa_{n,p_j}>0$ and~$c_n\in(0,1]$ such that 
\begin{equation}\label{eq:kappachoice}
\mathbb{P}\left( \|\bm{Z}_{d} \|_{p_j} \geq \kappa_{n,p_j} \right) = \alpha_{n,p_j} \qquad \text{and}\qquad 
\P\del[3]{\max_{p \in \mathfrak{P}_n} \kappa_{n,p}^{-1}
\| \bm{Z}_{d} \|_{p} \geq c_n}
=
\alpha,
\end{equation}
where~$\mathfrak{P}_n=\cbr[0]{2,p_1,\hdots,p_{m_n},\infty}$.
Define the test~$\psi_n(\bm{\beta}_n^*)$ as
\begin{align}\label{eq:psidef}
\psi_n(\bm{\beta}_n^*):=\mathds{1}\cbr[2]{\max_{p\in\mathfrak{P}_n} \kappa_{n,p}^{-1}
S_{n,p}(\bm{\beta}_n^*)\geq c_n}.
\end{align}
Thus,~$\psi_n(\bm{\beta}_n^*)$ rejects~$H_{0,n}:\E h_{n}(\bm{X}_{1,n},\bm{\beta}_n^*)=\bm{0}_{d}$ if for any~$p$-norm based tests with~$p\in\mathfrak{P}_n$ it is the case that~$S_{n,p}(\bm{\beta}_n^*)$ exceeds~$c_n\kappa_{n,p}$. The tests based on~$p=2$ and~$p=\infty$ are included in the construction of~$\psi_n(\bm{\beta}_n^*)$ to make it powerful against dense and sparse alternatives, respectively. The tests based on~$p\in\cbr[0]{p_1,\hdots,p_{m_n}}$ are included to cover the (semi-sparse) alternatives that neither tests based on the 2- or supremum-norm are consistent against. The quantities~$\alpha_2,~\alpha_I,$ and~$\alpha_\infty$ control how much size is distributed to the components of~$\psi_n(\bm{\beta}_n^*)$ targeting dense, semi-sparse and sparse alternatives, respectively while the array~$\mc{A}$ distributes the size~$\alpha_I$ to the~$p$-norm based tests with~$p\in\cbr[0]{p_1,\hdots,p_{m_n}}$. The role of~$c_n$ is to ensure that~$\psi_n(\bm{\beta}_n^*)$ has asymptotic size~$\alpha$, cf. Part~1~of Theorem~\ref{thm:domtest} below. Choosing~$c_n=1$ will generally result in~$\psi_n(\bm{\beta}_n^*)$ having asymptotic size not larger but potentially smaller than~$\alpha$, i.e.,~a conservative asymptotic level~$\alpha$ test.
\begin{theorem}\label{thm:domtest}
Let~$\alpha\in(0,1)$ and suppose Assumption~\ref{as:suffc} is satisfied. Then, the following holds.
\begin{enumerate}
\item Size control: If~$\bm{\beta}^*\in \mathbf{B}^{(0)}$, then~$\lim_{n\to\infty}\E \psi_n(\bm{\beta}_n^*)= \alpha$.
\item Dominance: For any~$p\in[2,\infty]$,~$\bar{\alpha}_p\in(0,1)$, and~$(\bar{\kappa}_{n,p}(\bar{\alpha}_p))_{n\in\N}$ given by~\eqref{eq:sizepnorm} in case~$p\in[2,\infty)$ and by~\eqref{eq:sizesup} in case~$p=\infty$ it holds that
\begin{align*}
\P\del[1]{S_{n,p}(\bm{\beta}_n^*)\geq \bar{\kappa}_{n,p}(\bar{\alpha}_p)}\to 1\qquad \text{implies} \qquad \E\psi_n(\bm{\beta}_n^*)\to 1.
\end{align*}
\end{enumerate}

\end{theorem}

Part 1~of Theorem~\ref{thm:domtest} shows that the critical values~$\kappa_{n,p}$,~$p\in\mathfrak{P}_n$, and the~$c_n$ that guarantee exact size~$\alpha\in(0,1)$ under Gaussianity in~\eqref{eq:kappachoice} also guarantee that~$\psi_n(\bm{\beta}_n^*)$ has asymptotic size~$\alpha$. This may not be obvious ex-ante as the construction of~$\psi_n(\bm{\beta}_n^*)$ is based on combining~$|\mathfrak{P}_n|=(m_n+2)\to\infty$ tests. However, having established~\eqref{eq:newGaussApprox}, this becomes a trivial consequence of writing non-rejection of~$\psi_n(\bm{\beta}_n^*)$ as the event
\begin{align*}
\cbr[3]{\hat{\bm{\Sigma}}_n^{-1/2}(\bm{\beta}_n^*)\bm{H}_n(\bm{\beta}_n^*)\in\bigcap_{p\in\mathfrak{P}_n}\mathbb{B}_p(c_n\kappa_{n,p})}
\end{align*}
where for~$r\in(0,\infty)$ we set~$\mathbb{B}_p(r)=\cbr[0]{\bm{x}\in\R^d:\|\bm{x}\|_p<r}$ and observe that $\bigcap_{p\in\mathfrak{P}_n}\mathbb{B}_p(c_n\kappa_{n,p})$ is convex.

Part 2~shows that~$\psi_n(\bm{\beta}_n^*)$ is consistent against \emph{every} deviation from the null hypothesis that \emph{some}~$p$-norm based test with~$p\in[2,\infty]$ is consistent against (irrespective of the asymptotic size of the latter, as long as it is a number in $(0, 1)$). In particular,~$\psi_n(\bm{\beta}_n^*)$ is consistent whenever tests based on the~$2$- or~$\infty$-norm with asymptotic size in $(0, 1)$ are consistent. Recall that the (dense and sparse) alternatives that tests based on these norms are consistent against are the ones targeted by current applications of the power enhancement principle of~\cite{fan2015power}, cf.~also~\cite{kp}. In addition,~$\psi_n(\bm{\beta}_n^*)$ is consistent against further alternatives, as it is also consistent as soon as there exists a~$p\in(2,\infty)$ such that a test based on this~$p$ is consistent (a property that tests based on the power enhancement principle that combine $2$- and $\infty$-norm based tests do not share, cf.~\cite{kp2021consistency}).

Furthermore, as a consequence of Theorems~\ref{thm:pnorm-char_GMM} and~\ref{thm:supnorm}, it is not possible to choose a single best~$p$-norm to base a test on: Our results show that no matter which \emph{single}~$p\in[2,\infty]$ one chooses, there exists another~$p$-norm based test that is consistent against alternatives that the test based on the chosen~$p$ is inconsistent against. Theorem~\ref{thm:domtest} suggests a way to overcome this impossibility of choosing a single ``best'' norm to base a tests on as~$\psi_n(\bm{\beta}_n^*)$ simultaneously dominates all~$p$-norm based tests for~$p\in[2,\infty]$.

Theorem~\ref{thm:domtest} is silent about the power properties of~$\psi_n(\bm{\beta}_n^*)$ when no~$p$-norm based test is consistent. Thus, there could, in principle, exist alternatives against which a~$p$-norm based test has substantially higher asymptotic power than~$\psi_n(\bm{\beta}_n^*)$ without the former converging to one.  The following theorem shows that this possibility can be ruled out by allocating sufficient asymptotic size~$\alpha_p$ to a given~$p$-norm in the construction of~$\psi_n(\bm{\beta}_n^*)$.
\begin{theorem}\label{thm:boundedloss}
Let~$\alpha\in(0,1)$ and suppose Assumption~\ref{as:suffc} is satisfied. In the context of the construction of~$\psi_n(\bm{\beta}_n^*)$ fix~$p\in\cup_{n\in\N}\mathfrak{P}_n$ and define~$\alpha_{p}:=\lim_{n \to \infty}\alpha_{n,p}\in(0,1)$ if~$p\in(2,\infty)$.\footnote{Note that~$\alpha_2$ and~$\alpha_\infty$ have been defined already in the construction of~$\psi_n(\bm{\beta}_n^*)$.}
\begin{enumerate}
\item If~$p\in[2,\infty)$, then for any sequence~$(\bar{\kappa}_{n,p}(\alpha))_{n\in\N}$ as in~\eqref{eq:sizepnorm}  
\begin{align*}
\limsup_{n\to\infty}\sbr[2]{\P\del[1]{S_{n,p}(\bm{\beta}_n^*)\geq \bar{\kappa}_{n,p}(\alpha)}-\E \psi_n(\bm{\beta}_n^*)}
\leq
\frac{\Phi^{-1}(1-\alpha_{p}) - \Phi^{-1}(1-\alpha)}{\sqrt{2\pi}}.
\end{align*}
\item If~$p=\infty$, then for any sequence~$(\bar{\kappa}_{n,\infty}(\alpha))_{n\in\N}$ as in~\eqref{eq:sizesup} %
\begin{align*}
\limsup_{n\to\infty}\sbr[2]{\P\del[1]{S_{n,\infty}(\bm{\beta}_n^*)\geq \bar{\kappa}_{n,\infty}(\alpha)}-\E \psi_n(\bm{\beta}_n^*)}
\leq
f(\alpha)-f(\alpha_\infty),
\end{align*}
where~$f:(0,1)\to \R$ is defined via~$x\mapsto \log\del[0]{-\log(1-x)/2}$.	
\end{enumerate}

\end{theorem}
Theorem~\ref{thm:boundedloss} shows that compared to the~$p$-norm based test of asymptotic size~$\alpha$ and with exponent $p$, there are no alternatives for which much is lost in terms of asymptotic power by using a test~$\psi_n(\bm{\beta}_n^*)$ with~$\alpha_{p}$ close to~$\alpha$. On the other hand: i) $\psi_n(\bm{\beta}_n^*)$ is consistent against any alternative that some~$p$-norm based test is consistent against and ii) alternatives can exist against which~$\psi_n(\bm{\beta}_n^*)$ is consistent yet the~$p$-norm based test has asymptotic power equal to its asymptotic size. Thus, there can be much to gain from using~$\psi_n(\bm{\beta}_n^*)$. Note, however, that the above construction of a test necessitates the choice of $p$ to target. That is,~$\psi_n(\bm{\beta}_n^*)$ with asymptotic size~$\alpha$ cannot have power arbitrarily close to, e.g., that of the $2$-norm and $\infty$-norm based tests (each with asymptotic size~$\alpha$) simultaneously.

\section{Many IVs and dominating the Anderson-Rubin test}\label{ex:IVdom}
Recall that in the context of the linear IV model in Section~\ref{sec:IVexample} one has
\begin{align*}
h_{n}(\bm{X}_{1,n},\bm{\beta}_n^*)
=
(y_{1,n}-\bm{Y}_{1,n}'\bm{\beta}_n^*)\bm{z}_{1,n}.
\end{align*}
The theory developed in the previous sections has the following consequences for testing whether a given~$\bm{\beta}_n^*$ satisfies~$\E (y_{1,n}-\bm{Y}_{1,n}'\bm{\beta}_n^*)\bm{z}_{1,n}=\bm{0}_d$:
\begin{enumerate}
\item Choosing~$p=2$ results in a classic weak identification robust Anderson-Rubin (AR) test. Thus, this AR test is dominated in terms of consistency by any~$p$-norm based test with~$p\in(2,\infty)$, cf.~the relationship in~\eqref{eq:dom}. There is no ranking between tests based on~$p\in[2,\infty)$ and the sup-score type test corresponding to~$p=\infty$; see~\cite{belloni2012sparse} for the definition of the original sup-score statistic.
\item For simplicity, let~$\mathbf{B}_n=\R$ and assume that~$\bm{\beta}_n$ satisfies~$\E (y_{1,n}-\bm{\beta}_n \bm{Y}_{1,n})\bm{z}_{1,n}=\bm{0}_d$. Even if only one instrument (say) is relevant in the sense of~$\E \bm{z}_{1,n}\bm{Y}_{1,n}=(a_n,0,\hdots,0)'$ for some~$a_n\neq 0$ such that the moments
\begin{align*}
\E(y_{1,n}-\bm{\beta}_n^*\bm{Y}_{1,n})\bm{z}_{1,n}
=
\E \bm{z}_{1,n}\bm{Y}_{1,n}(\bm{\beta}_n-\bm{\beta}_n^*)
=
\del[1]{a_n(\bm{\beta}_n-\bm{\beta}_n^*),0,\hdots,0}'
\end{align*}
only differ from zero in one entry, this sparsity need not be inherited by
\begin{align*}
\bm{\theta}_n(\bm{\beta}_n^*)=\sqrt{n}[\bm{\Sigma}_{n}(\bm{\beta}_n^*)]^{-1/2}\del[1]{a_n(\bm{\beta}_n-\bm{\beta}_n^*),0,\hdots,0}'	
\end{align*}
unless one is willing to impose structure on the covariance matrix~$\bm{\Sigma}_{n}(\bm{\beta}_n^*)$ of~$(y_{1,n}-\bm{\beta}_n^* \bm{Y}_{1,n})\bm{z}_{1,n}$. Thus, as it is~$\bm{\theta}_n(\bm{\beta}_n^*)$ that determines the power properties of~$p$-norm based tests, one cannot advise on which~$p$ to use solely based on the number of instruments that one suspects to be relevant. Similarly, one should be cautious basing advice on the sparsity structure of~$\bm{\pi}_n$ in the ``first-stage''~$\bm{Y}_{1,n}=\bm{\pi}_n'\bm{z}_{1,n}+\nu_{1,n}$: Since~$$\bm{\pi}_n=[\E (\bm{z}_{1,n}\bm{z}_{1,n}')]^{-1}\E\bm{z}_{1,n}\bm{Y}_{1,n}$$ (assuming that~$\E (\bm{z}_{1,n}\bm{z}_{1,n}')$ is invertible) one can write~$\E\bm{z}_{1,n}\bm{Y}_{1,n}=\E (\bm{z}_{1,n}\bm{z}_{1,n}')\bm{\pi}_n$ and hence
\begin{align*}
\bm{\theta}_n(\bm{\beta}_n^*)=\sqrt{n}[\bm{\Sigma}_{n}(\bm{\beta}_n^*)]^{-1/2}\E (\bm{z}_{1,n}\bm{z}_{1,n}')\bm{\pi}_n\cdot(\bm{\beta}_n-\bm{\beta}_n^*),
\end{align*}
cf.~the first equality in the penultimate display. Again there is no link between the sparsity pattern of the first-stage~$\bm{\pi}_n$ and that of~$\bm{\theta}_n(\bm{\beta}_n^*)$; the latter determining the power properties of~$p$-norm based tests.
\item In light of the previous point, it is difficult to give advice on which single~$p$-norm to base a test on even if one is willing to impose assumptions on the number of relevant instruments. Thus, it is useful that the test~$\psi_n(\bm{\beta}_n^*)$ in~\eqref{eq:psidef} is consistent whenever some~$p$-norm based test (including the AR- or sup-score type tests) is consistent. Therefore,~$\psi_n(\bm{\beta}_n^*)$ is consistent against strictly more alternatives than any test based on a single~$p$.  In this sense,~$\psi_n(\bm{\beta}_n^*)$ does not rely on knowing whether~$\bm{\theta}_n(\bm{\beta}_n^*)$ is sparse or not. 
\end{enumerate}

\section{Rapidly increasing dimension via sample splitting}\label{sec:rapidgrowth}
The results presented so far generally restrict the growth rate of~$d$. Suppose, however, that one wishes to test whether~$\bm{\beta}_n^*$ satisfies~$D=D(n)$ moment conditions
\begin{align}\label{eq:GMMlarge}
\E h_{j,n}(\bm{X}_{1,n},\bm{\beta}_n^*)=0,\qquad j=1,\hdots,D, 
\end{align}
with~$D$ \emph{unrestricted}. 
Under independent sampling one can use a two-step procedure to test the validity of~\eqref{eq:GMMlarge} by means of the tests developed in the previous sections: Based on a partition of the sample into two disjoint subsamples~$N_1=N_1(n)\subset\cbr[0]{1,\hdots, n}$ and~$N_2=N_2(n)\subset\cbr[0]{1,\hdots, n}$ of sizes~$n_1=n_1(n)$ and~$n_2=n_2(n)$, respectively, one first uses the subsample~$N_1$ to select a subset~$S=S(n_1,n_2,\cbr[0]{\bm{X}_{i,n}}_{i\in N_1})\subset \cbr[0]{1,\hdots,D}$ of moment functions. In a second step, using the subsample~$N_2$, one tests whether~$\bm{\beta}_n^*$ satisfies
\begin{align*}
\E h_{j,n}(\bm{X}_{1,n},\bm{\beta}_n^*)=0,\qquad j\in S 
\end{align*}
using only the data~$\cbr[0]{\bm{X}_{i,n}}_{i\in N_2}$. In Theorem~\ref{thm:dimred} in Section~\ref{app:samplesplit} in the appendix we show that if a predetermined fixed number of \emph{selected} moments~$|S|=d(n_2)$ satisfies the growth conditions of the theory developed in the previous sections relative to the second step ``sample size''~$n_2$, then all asymptotic size guarantees established so far remain valid. 
The concrete moment selection mechanism~$S$ does not affect the asymptotic size of the subsequent tests as independent sampling guarantees that no selection bias is introduced. However, it is clear that the power of a test depends on~$S$ including moments that are violated. Some ways to select the subset of moments~$S$ to be tested in order to obtain powerful tests are discussed in Appendix~\ref{app:samplesplit}, where the size-control result informally sketched above is also formally established in Theorem~\ref{thm:dimred}.

\section{Numerical results}\label{sec:sims}
In Section~\ref{sec:simIV} we investigate numerically the properties of our test~$\psi_n(\bm{\beta}_n^*)$ for hypothesis testing on a structural parameter in the presence of many instruments, cf.~Sections~\ref{sec:IVexample} and~\ref{ex:IVdom}. Then, to investigate the effect of even larger~$d$ without unduly increasing the computational burden, Section~\ref{sec:simGauss} considers directly the limiting Gaussian testing problem which by~\eqref{eq:newGaussApprox}  approximates the rejection frequencies of~$\psi_n(\bm{\beta}_n^*)$ and~$p$-norm based tests in many moment testing problems.

\subsection{Testing in the presence of many instruments}\label{sec:simIV}
We generate data from the following standard two-equation model with one endogenous regressor ($\mathbf{B}_n=\R$), but many instruments~$d$:
\begin{align*}
y_{i,n}
&=
\bm{Y}_{i,n}\bm{\beta}_n+u_{i,n},\hspace{3.5cm}\text{``Second stage''}\\
	\bm{Y}_{i,n}
&=
	\bm{\pi}_n'\bm{z}_{i,n}+\nu_{i,n},\qquad i=1,\hdots,n.\qquad\text{``First stage''}
\end{align*}
The variables are named and interpreted as in Section~\ref{sec:IVexample}. All variables are i.i.d.~across~$i=1,\hdots,n$ and we consider a setting where the instruments~$\bm{z}_{i,n}$ and error terms have heavy tails in the sense that the~$d$ independent entries of~$\bm{z}_{1,n}$ follow a~$t(5)$-distribution and
\begin{align*}
	\begin{pmatrix}
		u_{1,n}\\
		\nu_{1,n}
	\end{pmatrix}
	\sim t_{5}\del[1]{\bm{0}_2,\bm{\Omega}},\qquad \text{where }\bm{\Omega}=
	\begin{pmatrix}
		1 & 0.9\\
		0.9 & 1
	\end{pmatrix}.
\end{align*}
In addition,~$\bm{z}_{1,n}\indep(u_{1,n},\nu_{1,n})$. Thus, the ``true'' parameter~$\bm{\beta}_n$ satisfies~$\E (y_{1,n}-\bm{Y}_{1,n}'\bm{\beta}_n)\bm{z}_{1,n}=\E u_{1,n}\bm{z}_{1,n}=\bm{0}_d$ and we test whether the candidate structural parameter~$\bm{\beta}_n^*$ satisfies 
\begin{align*}
	H_{0,n}:\E (y_{1,n}-\bm{Y}_{1,n}'\bm{\beta}_n^*)\bm{z}_{1,n}=\bm{0}_d\qquad \text{vs.} \qquad H_{1,n}: \E (y_{1,n}-\bm{Y}_{1,n}'\bm{\beta}_n^*)\bm{z}_{1,n}\neq\bm{0}_d,
\end{align*}
which falls within our general testing framework. Throughout we use~$\bm{\beta}_n^*=0$ and gauge the size and power of the tests considered by generating data for a range of~$\bm{\beta}_n$. 

We consider~$(n,d)\in\cbr[1]{(5{,}000,100),(25{,}000,500),(100{,}000,1{,}000)}$. Although one cannot always obtain~$n$ of this order, these settings reflect many practically relevant situations, cf.~Footnote~\ref{fn:motivation}. 

 The number of relevant instruments is governed by the number of non-zero entries of~$\bm{\pi}_n\in\R^d$, i.e.,~the sparsity/denseness of the ``first stage''. For each pair~$(n,d)$ we consider the sparsest and densest possible first stage, namely~$\bm{\pi}_n=(1,0,\hdots,0)$ and~$\bm{\pi}_n=(1,\hdots,1)$. As the latter first stage is clearly more informative, we adjust the magnitude of the deviation of~$\bm{\beta}_n$ from~$\bm{\beta}_n^*=0$ according to the ``strength'' of~$\bm{\pi}_n$ to get non-trivial power curves.\footnote{Alternatively, one can adjust the magnitude of the entries of~$\bm{\pi}_n$, i.e.~change the strength of the individual instruments, with the number of non-zero entries of~$\bm{\pi}_n$.} Finally, we consider a setting of a semi-sparse first stage where a moderate number of instruments is relevant. Here~$\bm{\pi}_n=(\bm{\iota}_{d_{R,n}},0,\hdots,0)$ where~$\iota_{d_{R,n}}$ is a vector of~$d_{R,n}$ ones and~$d_{R,n}=4,7$,~$12$, for~$n=5{,}000,25{,}000$ and~$100{,}000$, respectively. We also note that a simple calculation reveals that in the present setting the number of non-zero entries of~$\bm{\theta}_n(\bm{\beta}^*_n)=\sqrt{n}\bm{\Sigma}^{-1/2}_{n}(\bm{\beta}^*_n)\E h_{n}(\bm{X}_{1,n},\bm{\beta}^*_n)$ equals that of~$\bm{\pi}_n$ as~$\bm{\Sigma}^{-1/2}_{n}(\bm{\beta}^*_n)$ is block-diagonal. Thus, when we vary the sparsity of~$\bm{\pi}_n$, we identically vary that of~$\bm{\theta}_n(\bm{\beta}^*_n)$.

Throughout, we estimate~$\bm{\Sigma}_n(\bm{\beta}_n^*)$ by the sample covariance matrix of the vectors~$(y_{i,n}-\bm{Y}_{i,n}'\bm{\beta}_n^*)\bm{z}_{i,n},\ i=1,\hdots,n$. We also experimented with the median-of-means estimator in Section 3.4 of~\cite{ke2019user}, but (although being more precise for heavy-tailed distributions and allowing for larger~$d$) this did not improve the performance of the resulting tests in the settings considered.

We implement our test~$\psi_n(\bm{\beta}_n^*)$ from Section~\ref{sec:domtestconstruct} with~$\alpha=0.05$,~$\mathfrak{P}_n=\cbr[0]{2,3,5,10,\infty}$, $\alpha_2=\alpha_\infty=\alpha_{n,p}=\alpha/5$ for~$p\in\cbr[0]{3,5,10}$ and all~$n$ considered. The power enhancement (PE) principle of~\cite{fan2015power}, which combines the~$2$- and supremum-norm based tests, is implemented by setting~$\alpha_2=\alpha_\infty=\alpha/2$ and~$\alpha_I=0$. For~$n=5{,}000$ we also implement the jacknifed AR test of~\cite{mikusheva2020inference}, which we refer to as MS.\footnote{For larger sample sizes the implementation time became prohibitive for the jacknifed AR test. However, unreported simulations revealed that for sufficiently small~$n$ (relative to~$d$) it controls size better than the~$p$-norm based tests that we study. This is in line with the theoretical results on size control of the jackknifed AR test in~\cite{mikusheva2020inference}.} The number of replications is~$1{,}000$ for~$n=5{,}000$ and~$n={25{,}000}$, but~$500$ for~$n=100{,}000$. Figure~\ref{fig:IVDiff} plots by how much the power of~$\psi_n(\bm{\beta}_n^*)$ \emph{exceeds} that of each of the other tests considered. Thus, when comparing to the~$2$-norm based test we plot the power difference~$\E \psi_n(\bm{\beta}_n^*)-\P\del[1]{S_{n,2}(\bm{\beta}_n^*)\geq \kappa_{n,2}}$ and similarly for the other tests considered. Positive numbers mean that~$\psi_n(\bm{\beta}_n^*)$ is superior and vice versa for negative numbers. The raw power function can be found in Figure~\ref{fig:IV} in Section~\ref{sec:furtherplots} of the appendix. Figure~\ref{fig:IVDiff} reveals: 
\begin{itemize}
	\item For each~pair~$(n,d)$ (i.e.,~in each row) our test~$\psi_n(\bm{\beta}_n^*)$  has power similar to that based on the supremum-norm and the PE based test when there is one relevant instrument (sparse first stage). The power can be much higher than that of the~$2$-norm based test. When all instruments are relevant (dense first stage),~$\psi_n(\bm{\beta}_n^*)$ has power similar to that of the~$2$-norm based test and the PE based one. The power of~$\psi_n(\bm{\beta}_n^*)$ can be much higher than that of the supremum-norm based test. $\psi_n(\bm{\beta}_n^*)$ is more powerful than~$2$-norm,~$\infty$-norm and PE based tests when the first-stage is semi-sparse. The power of the $2$-norm (AR-test) and the jackknifed AR-test are comparable.
\item As~$n$ and~$d$ increase (i.e.,~in each column) the power advantage of~$\psi_n(\bm{\beta}_n^*)$, the supremum-norm, and the PE based test over the~$2$-norm based test increases when there is one relevant instrument (cf.~the first column). Similarly, the power advantage of~$\psi_n(\bm{\beta}_n^*)$, the~$2$~norm, and the PE based test over the supremum-norm based test increases when all instruments are relevant (cf.~the third column). For~$(n,d)=(100{,}000,1{,}000)$ the gains of~$\psi_n(\bm{\beta}_n)$ are up to~$0.68$ and~$0.85$ for the sparse and dense first stage, respectively. 
	For the semi-sparse first stage the power of~$\psi_n(\bm{\beta}_n^*)$ is generally higher than that of all other tests considered. The power advantage increases for larger values of~$n$ and~$d$; for~$(n,d)=(100{,}000,1{,}000)$ it is seen that~$\psi_n(\bm{\beta}_n^*)$ can have power that is~$0.23$ and~$0.29$ higher than that of the~$2$- and supremum-norm based tests, respectively. The power advantage over the PE based is up to~$0.15$.
\end{itemize}

In the next section we shall see that the power gains from using~$\psi_n(\bm{\beta}_n^*)$ instead of a 2- or supremum-norm based test (or a power enhancement based combination of these) become even larger in higher dimensions.

\begin{figure}

\begin{center}
	\footnotesize $n=5{,}000$ and~$d=100$
\end{center}

\vspace{-0.5cm}
\includegraphics[width=5.2cm]{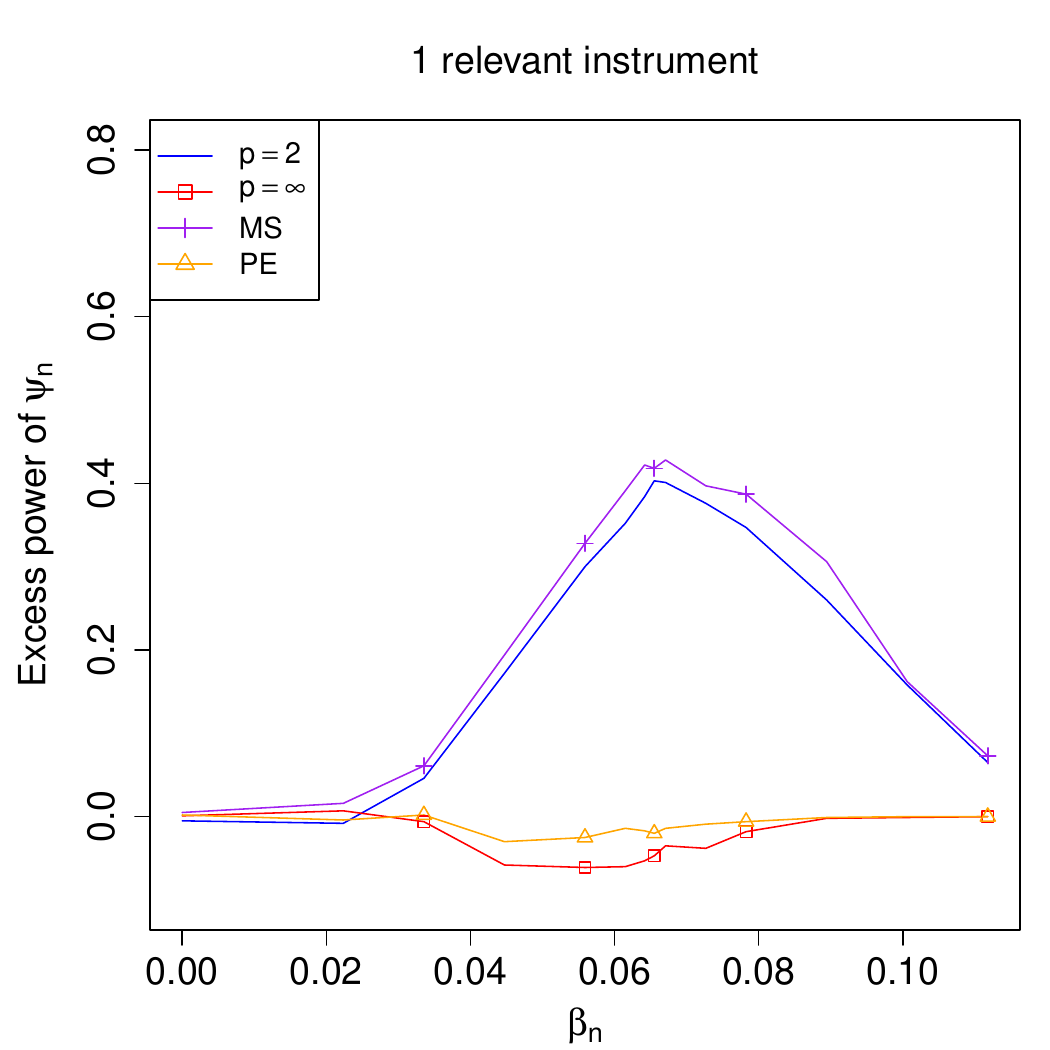}
\hspace{-0.5cm}
\includegraphics[width=5.2cm]{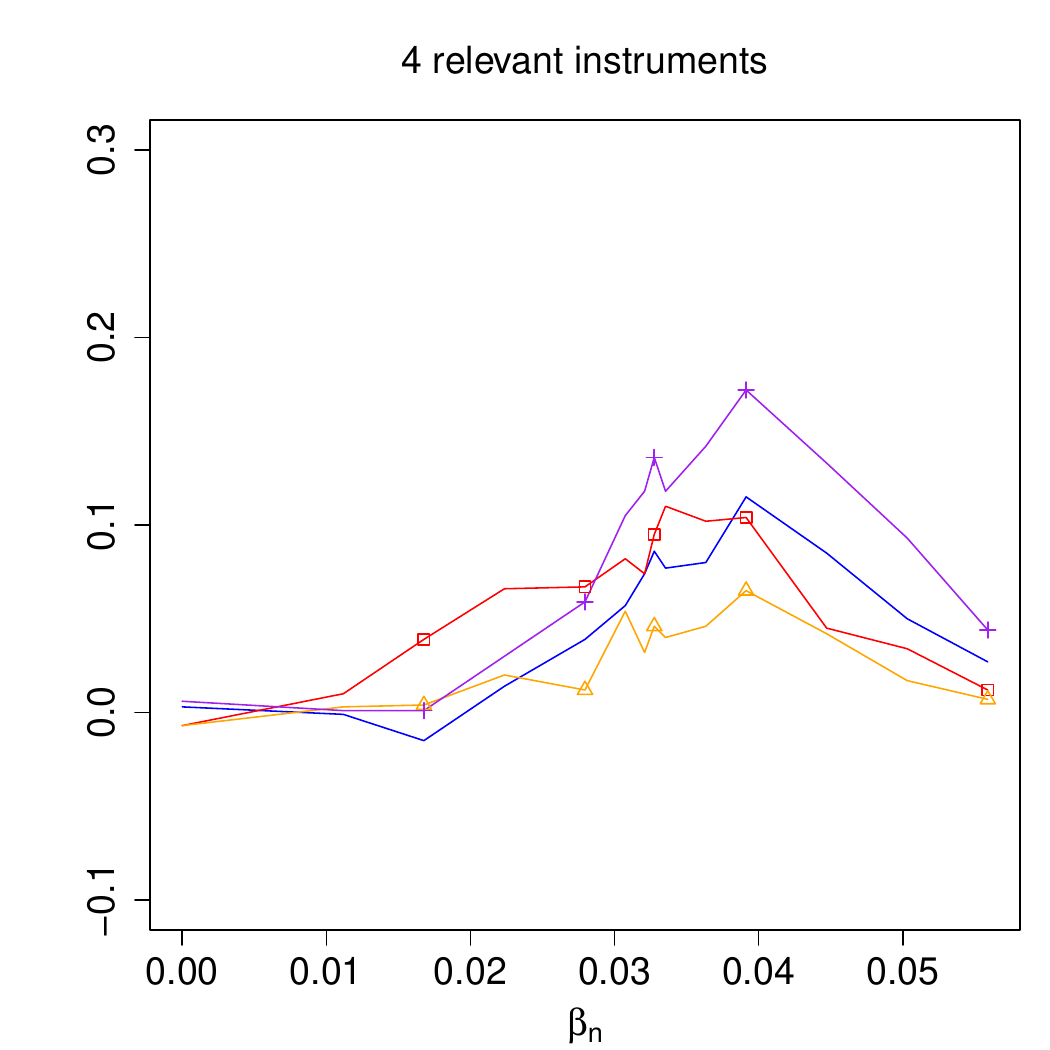}
\hspace{-0.5cm}
\includegraphics[width=5.2cm]{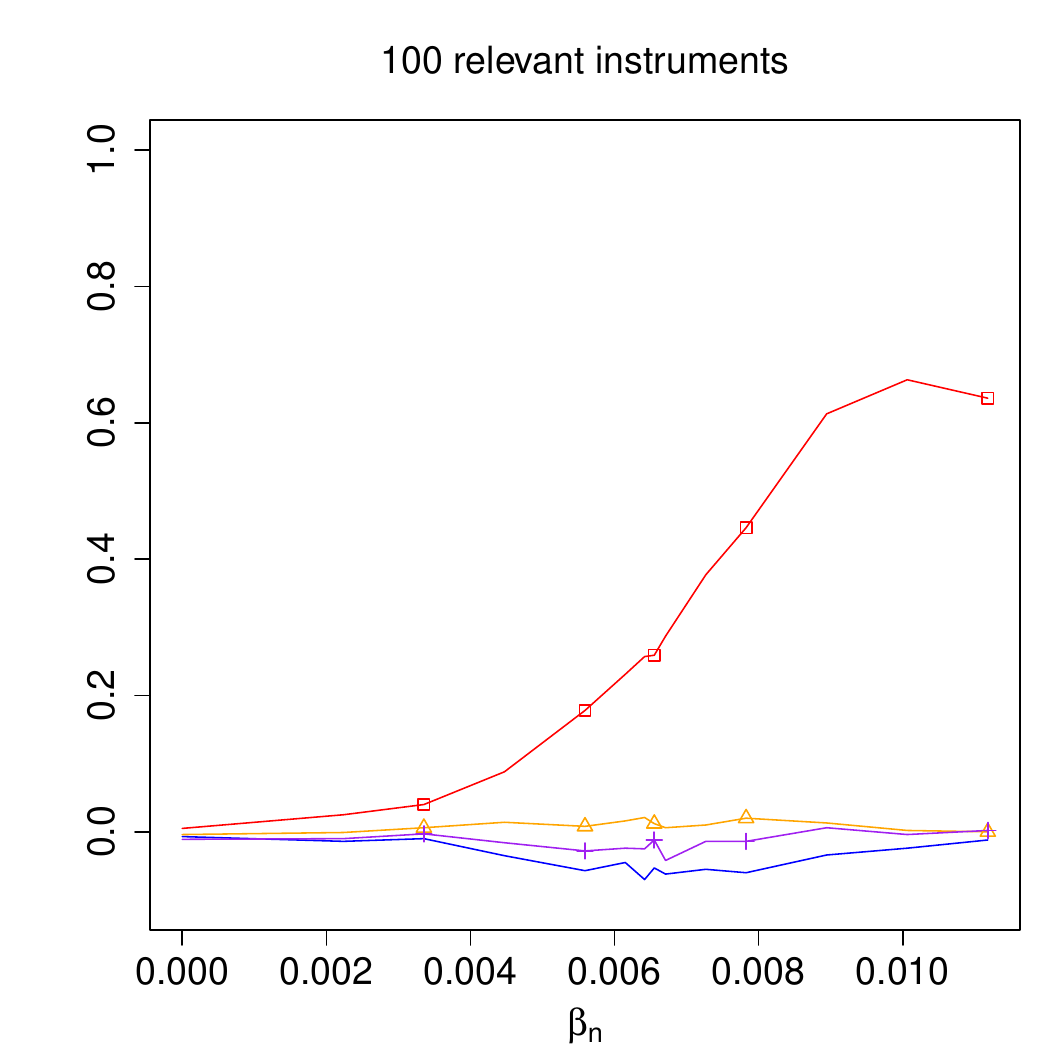}
\vspace{-0.2cm}

\begin{center}
	\footnotesize $n=25{,}000$ and~$d=500$
\end{center}

\vspace{-0.5cm}
\includegraphics[width=5.2cm]{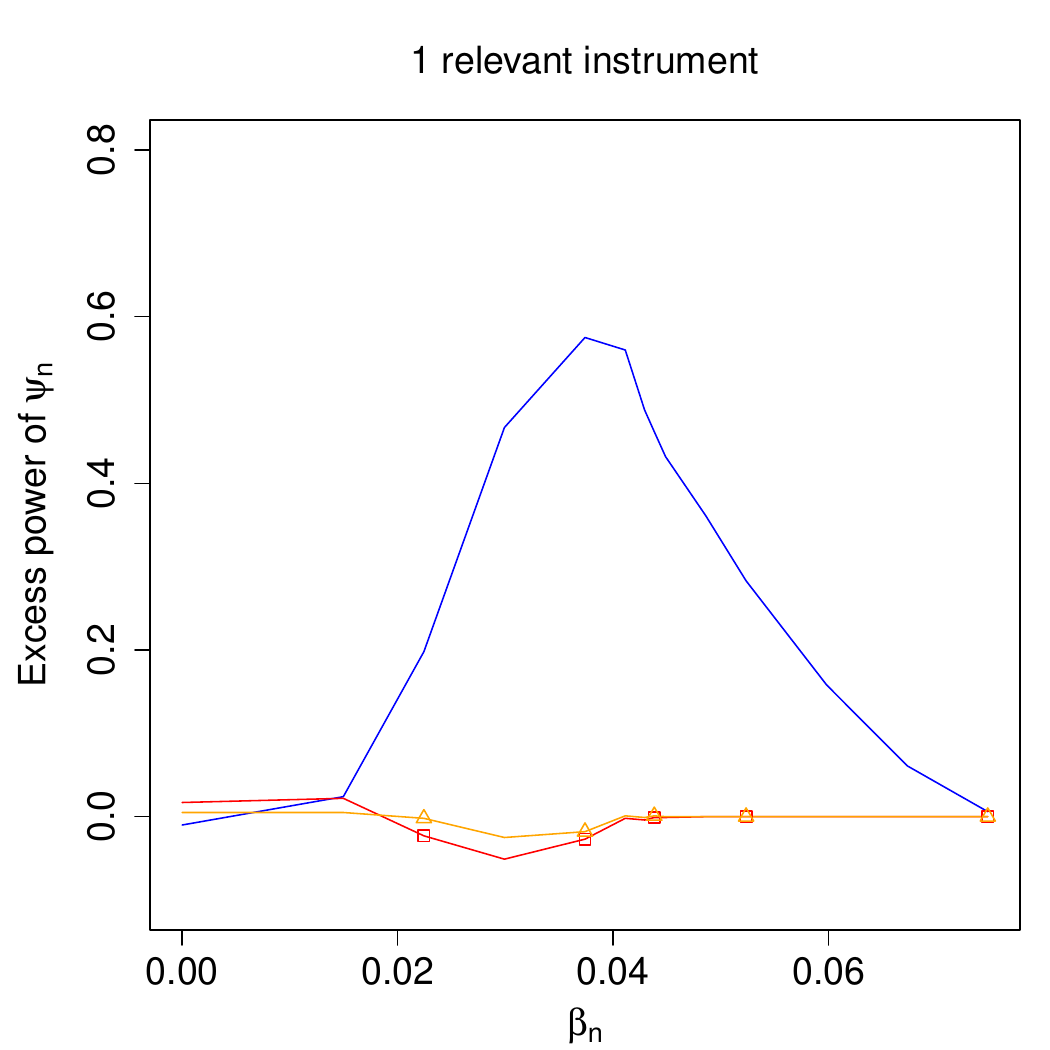}
\hspace{-0.5cm}
\includegraphics[width=5.2cm]{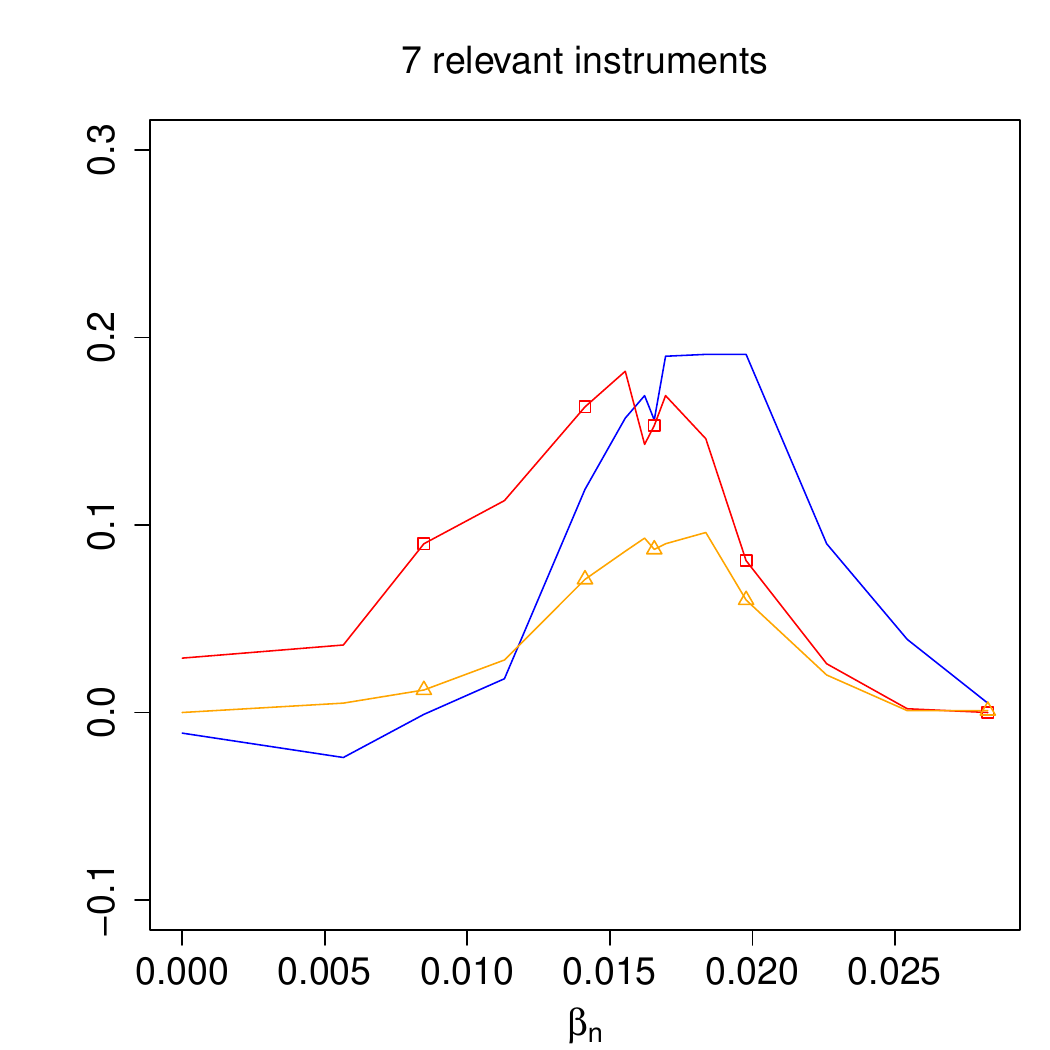}
\hspace{-0.5cm}
\includegraphics[width=5.2cm]{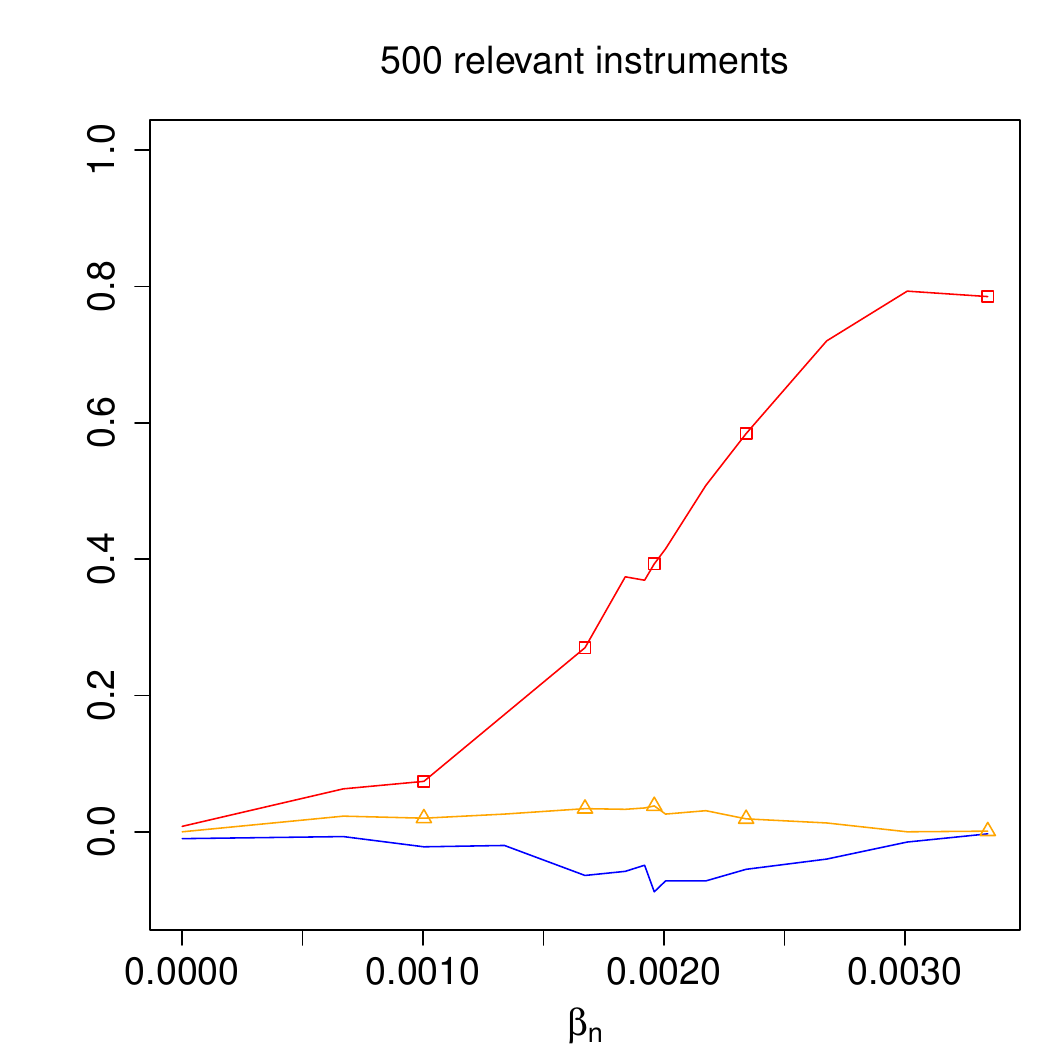}
\vspace{-0.2cm}

\begin{center}
	\footnotesize $n=100{,}000$ and~$d=1{,}000$
\end{center}

\vspace{-0.5cm}
\includegraphics[width=5.2cm]{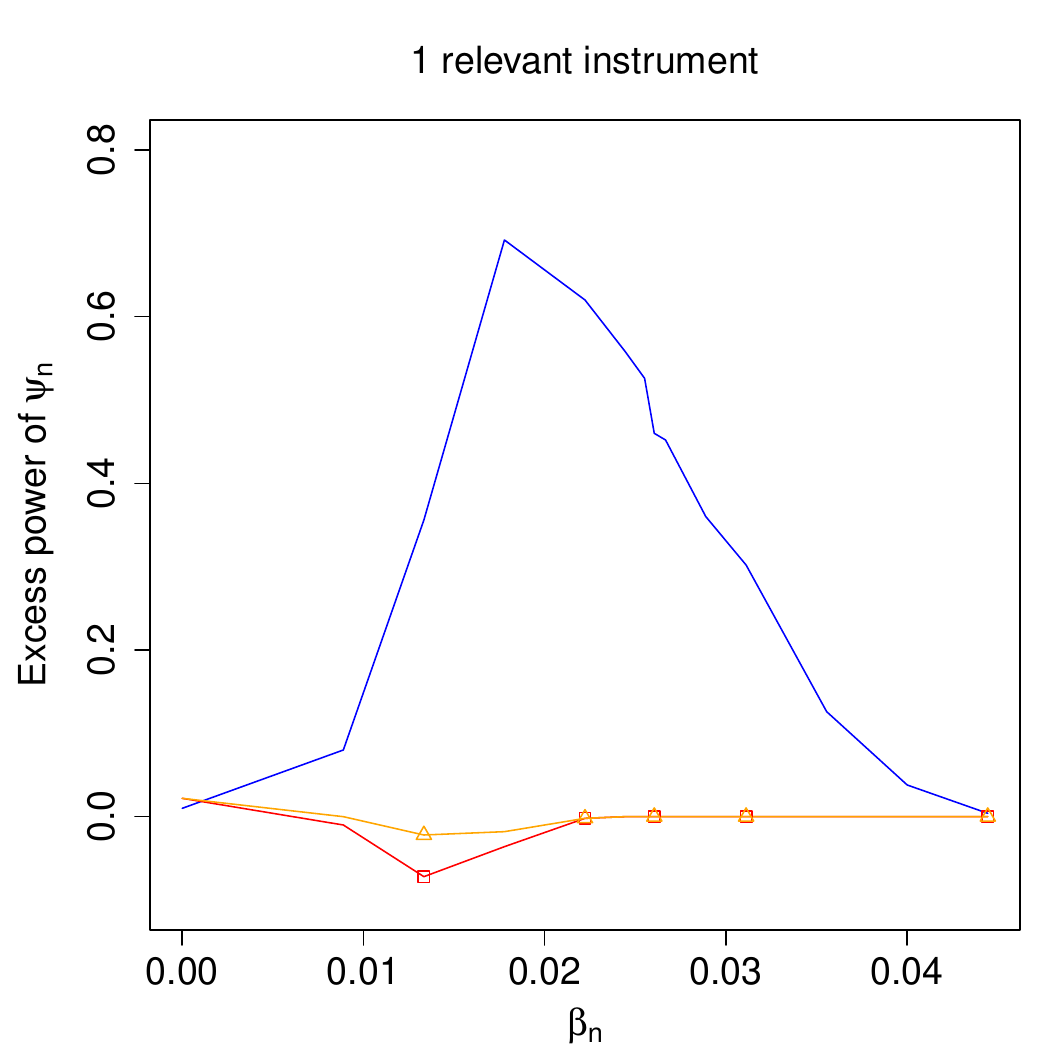}
\hspace{-0.5cm}
\includegraphics[width=5.2cm]{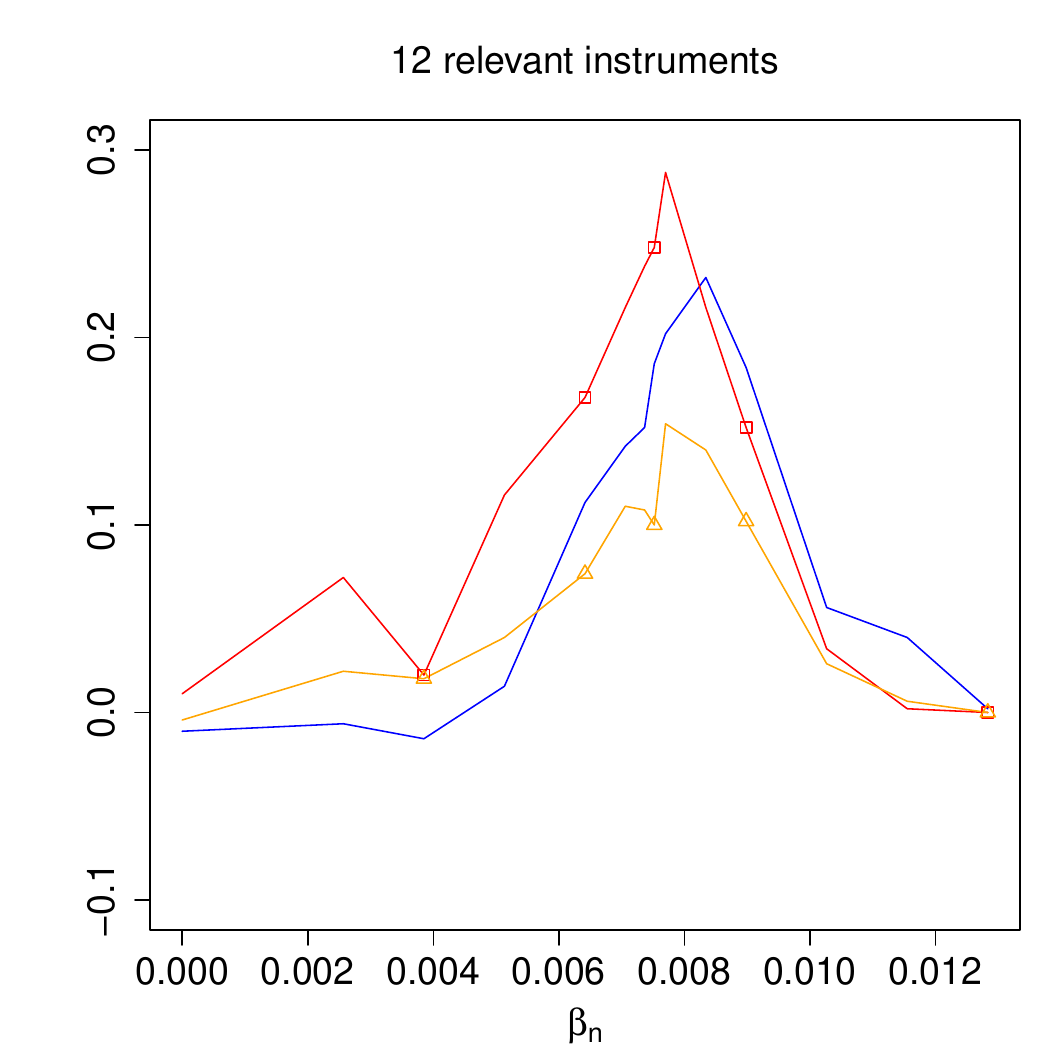}
\hspace{-0.5cm}
\includegraphics[width=5.2cm]{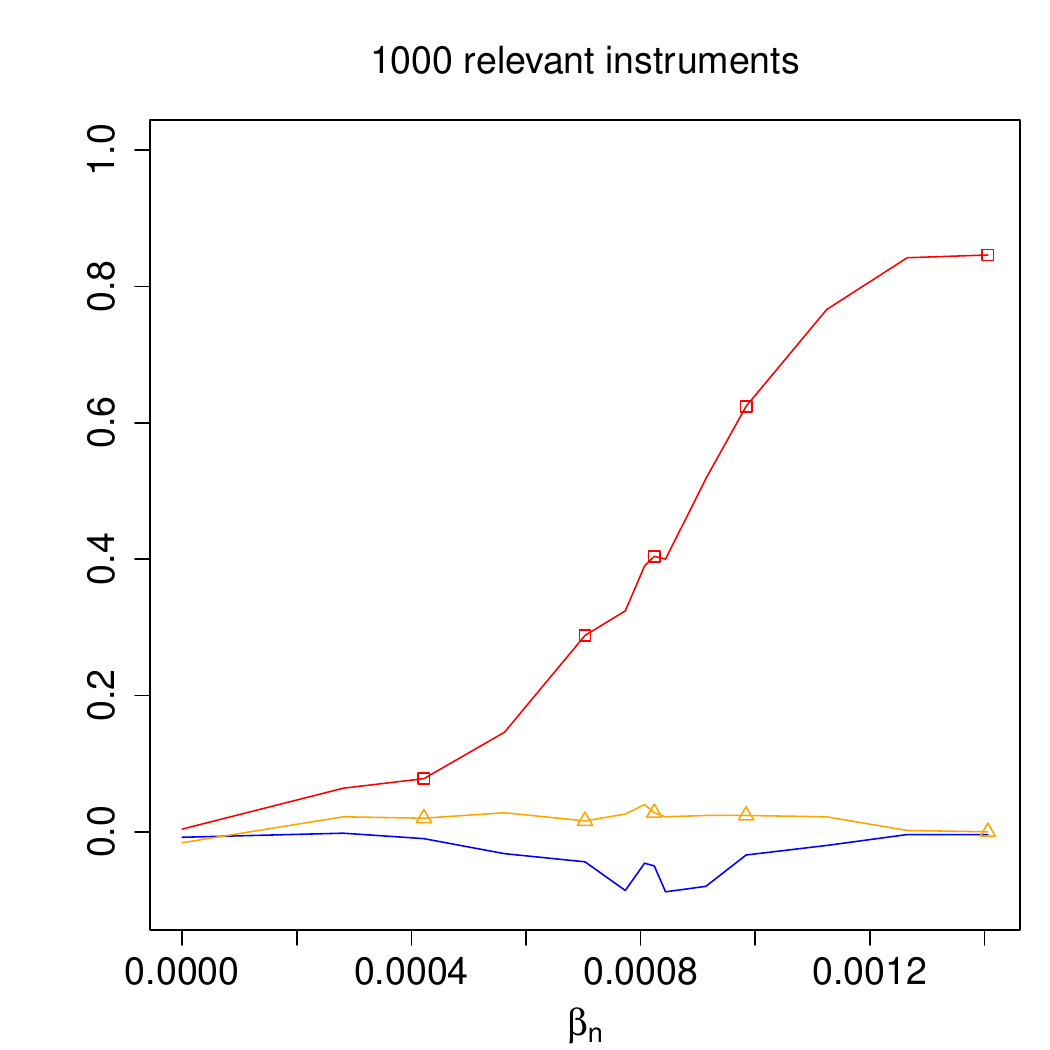}
\caption{\footnotesize Excess power of~$\psi_n(\bm{\beta}_n^*)$ over each of the other tests studied for~$(n,d)=(5{,}000,100)$ [first row], $(25{,}000, 500)$ [second row], $(100{,}000, 1{,}000)$ [third row] for sparse [first column], semi-sparse [second column] and dense [third column] alternatives. MS is the jackknifed Anderson-Rubin test of~\cite{mikusheva2020inference} and PE is the power enhancement principle. Full implementation details are given in the body text.}
\label{fig:IVDiff}	
\end{figure}

\subsection{Gaussian sequence model}\label{sec:simGauss}
To illustrate the power gains that one can expect when~$d$ becomes even larger than in the IV-based example in the previous subsection, we study testing
\begin{align*}
	H_{0,d}:\bm{\theta}_d=\bm{0}_d\qquad \text{vs.}\qquad H_{1,d}: \bm{\theta}_d\neq\bm{0}_d,
\end{align*}
when one observes one realization of~$\bm{Z}_d\sim\mathsf{N}_d(\bm{\theta}_d,\mathbf{I}_d)$ and~$p$-norm based tests are on the form~$\mathds{1}\del[0]{\enVert[0]{\bm{Z}_d}_p\geq \kappa_{d,p}}$,~$p\in[2,\infty]$.\footnote{Here we write~$\kappa_{d,p}$ rather than~$\kappa_{n,p}$ as~$n=1$ in the ``limit experiment'' (note that also in the previous sections the critical values~$\kappa_{n,p}$ only depended on~$n$ via~$d(n)$. A similar remark applies in the sequel.} By~\eqref{eq:newGaussApprox} the rejection frequencies $\P\del[0]{\enVert[0]{\bm{Z}_d}_p\geq \kappa_{d,p}}$ in this ``limit experiment'' will be close to~$\P\del[1]{S_{n,p}(\bm{\beta}_n^*)\geq \kappa_{n,p}}$ in a large set of moment testing problems as~$n\to\infty$ and with~$\bm{\theta}_d=\bm{\theta}_n(\bm{\beta}^*_n)=\sqrt{n}\bm{\Sigma}^{-1/2}_{n}(\bm{\beta}^*_n)\E h_{n}(\bm{X}_{1,n},\bm{\beta}^*_n)$. We study~$d=1{,}000,5{,}000, 50{,}000, 250{,}000$ as well as sparse, semi-sparse and dense deviations from~$H_{0,d}$. The sparse and dense~$\bm{\theta}_d$ have one and~$d$ (equal) non-zero entries, respectively. The semi-sparse alternatives have~$15$,~$30$,~$80$, and~$170$ (equal) non-zero entries, respectively (for increasing values of~$d$). All tests are implemented as in the IV-based simulations in the previous subsection. Analogously to Figure~\ref{fig:IVDiff}, Figure~\ref{fig:GaussDiff} plots by how much the power of~$\psi_d$ \emph{exceeds} that of the other tests studied. It is based on~$10{,}000$ replications.\footnote{The figures for~$d=50{,}000$ are identical to those in Figure~\ref{fig:intro}.}  The raw power function can be found in Figure~\ref{fig:Gauss} in Section~\ref{sec:furtherplots} of the appendix. Figure~\ref{fig:GaussDiff} reveals:
\begin{itemize}
	\item For each~$d$ (i.e., in each row) our test~$\psi_d$ has similar power to those based on the supremum-norm and the PE principle against sparse alternatives. The power is much higher than that of the~$2$-norm based test. For dense alternatives~$\psi_d$ is about as powerful as the tests based on the~$2$-norm and the PE principle. The power vastly exceeds that of the supremum-norm based test. $\psi_d$ is more powerful than~$2$-norm,~$\infty$-norm and PE based tests against the semi-sparse alternatives studied.
	\item As~$d$ increases (i.e., in each column) the power advantage of~$\psi_d$ and the PE based test over the~$2$-norm based test against sparse alternatives increases (cf.~the first column). Similarly, the power advantage of~$\psi_d$ and the PE based test over the supremum-norm based test against dense alternatives increases (cf.~the third column). These power gains are close to  the maximal possible value $0.95$($=1-\alpha$) for~$d\geq 5{,}000$. For the semi-sparse alternatives it is seen that at the point where the power functions of the~$2$- and supremum-norm based test intersect,~$\psi_d$ has a power that is~$0.2,0.3,0.4$ and~$0.5$ higher than that of these tests for~$d=1{,}000,5{,}000, 50{,}000, 250{,}000$, respectively. The power advantage over the PE based is up to~$0.15,0.2,0.3$ and~$0.4$, respectively, for these~$d$.
	\end{itemize} 
\emph{Thus, significant power gains can be obtained against semi-sparse alternatives from using~$\psi_d$ instead of a~$2$- or supremum-norm based test (or a combination of these) without sacrificing power against sparse or dense alternatives.} These power gains are obtained by harnessing the strength of norms beyond~$2$ and~$\infty$; in our case the~$3$-, $5$- and~$10$-norm.

\begin{figure}
\begin{center}
\vspace{-1.3cm}
\includegraphics[width=5.2cm]{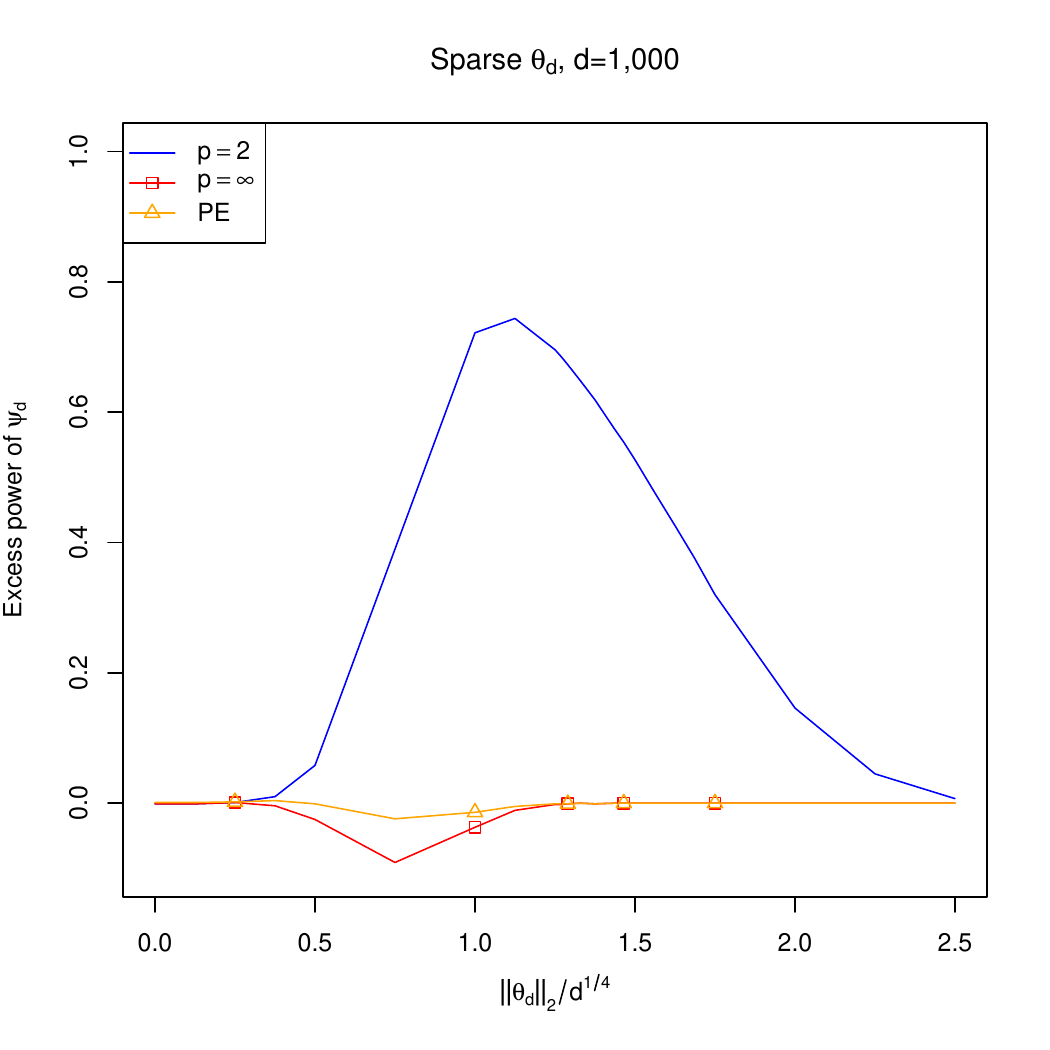}
\hspace{-0.6cm}
\includegraphics[width=5.2cm]{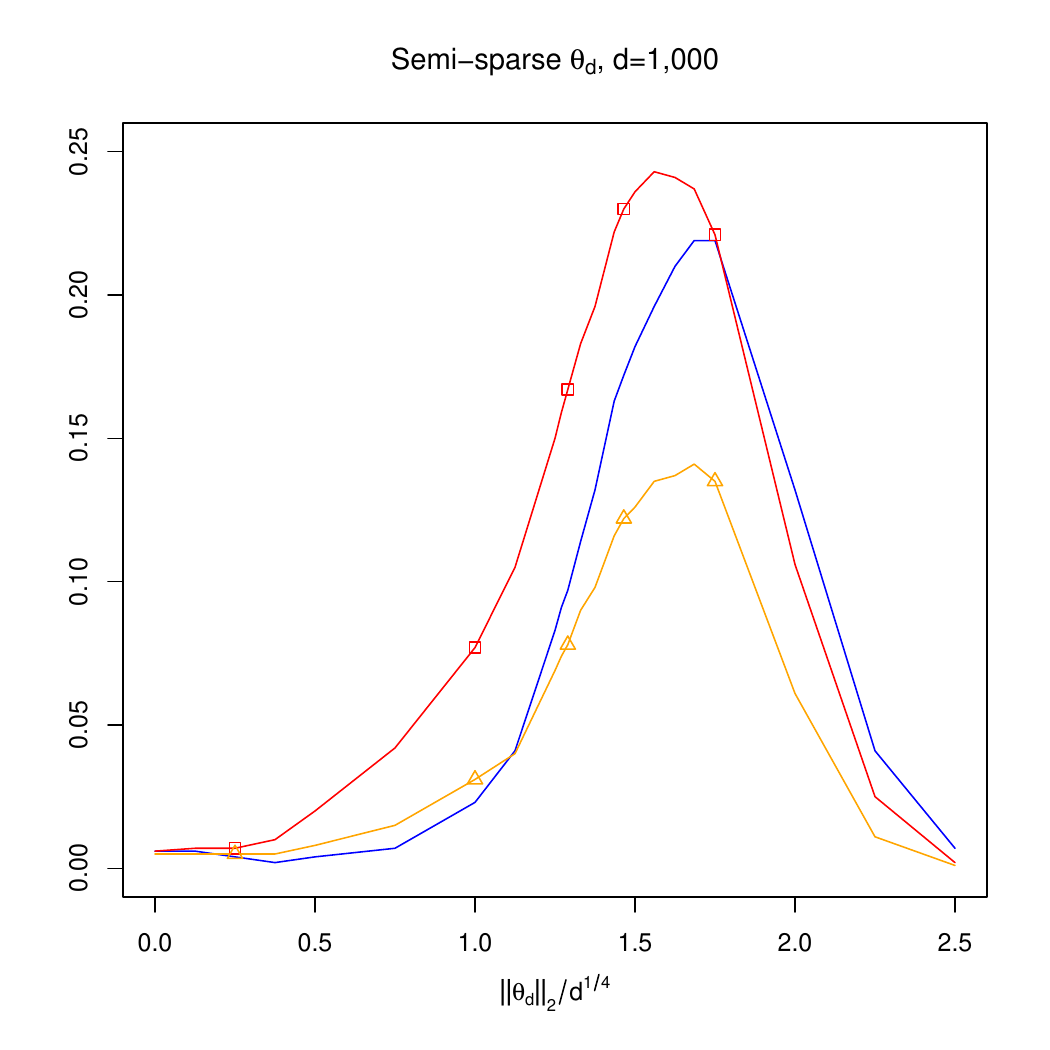}
\hspace{-0.6cm}
\includegraphics[width=5.2cm]{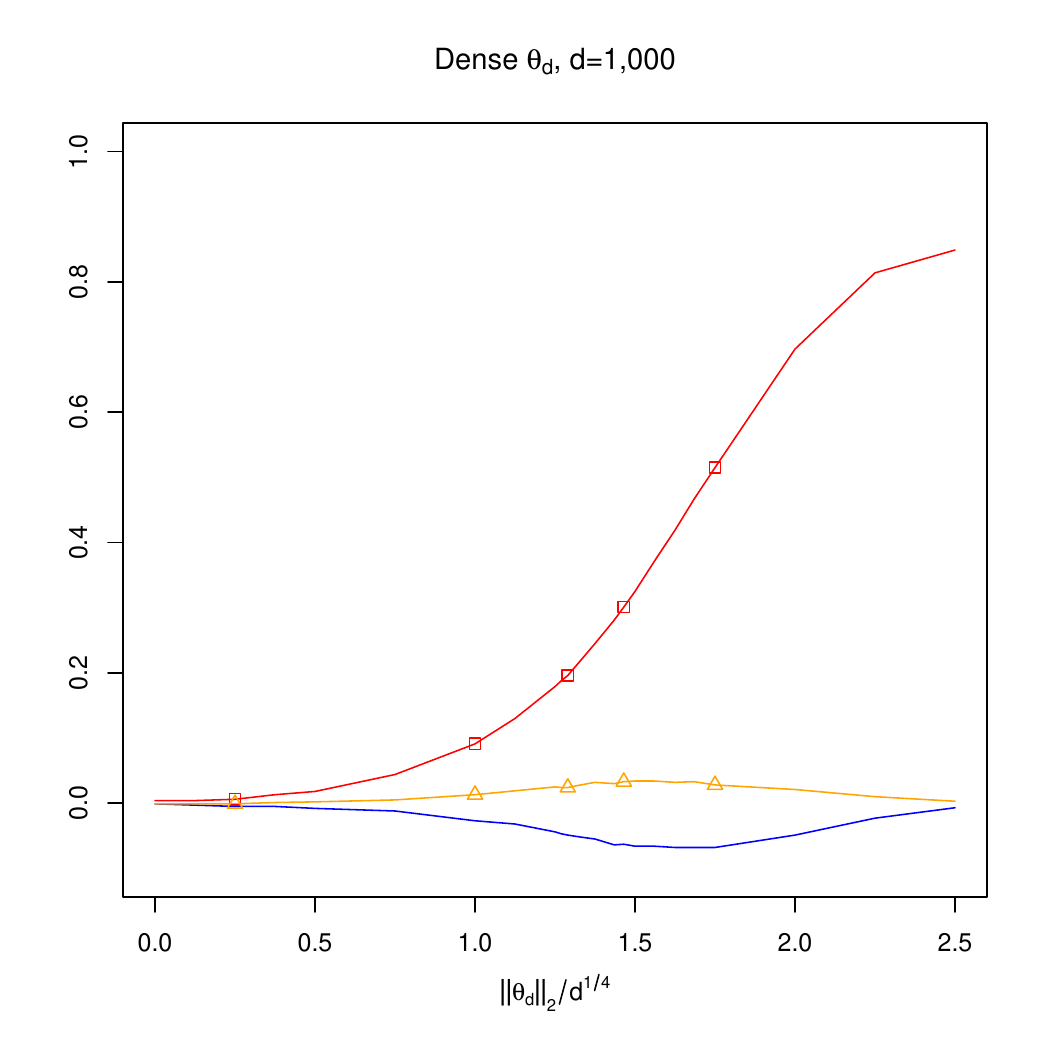}\null 
\vspace{-0.3cm}
\includegraphics[width=5.2cm]{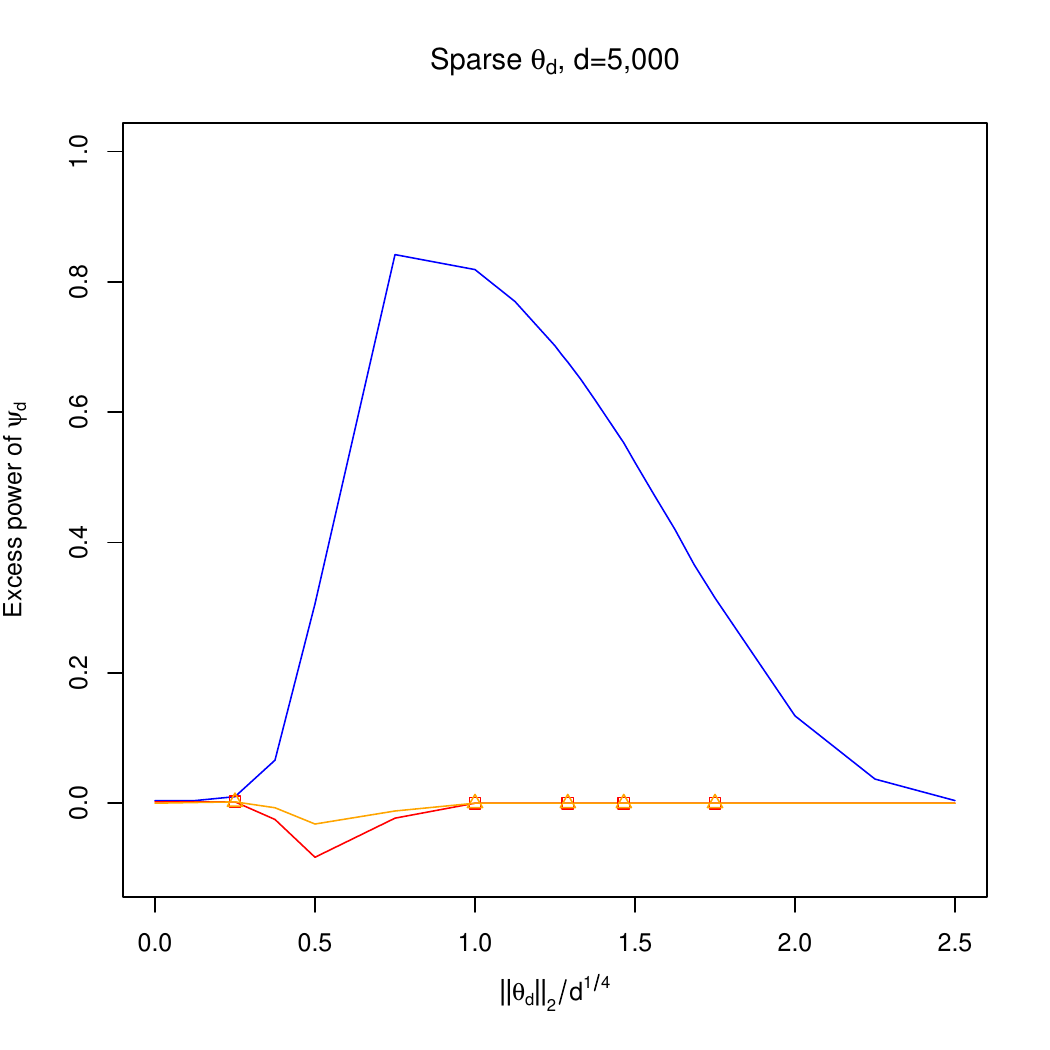}
\hspace{-0.6cm}
\includegraphics[width=5.2cm]{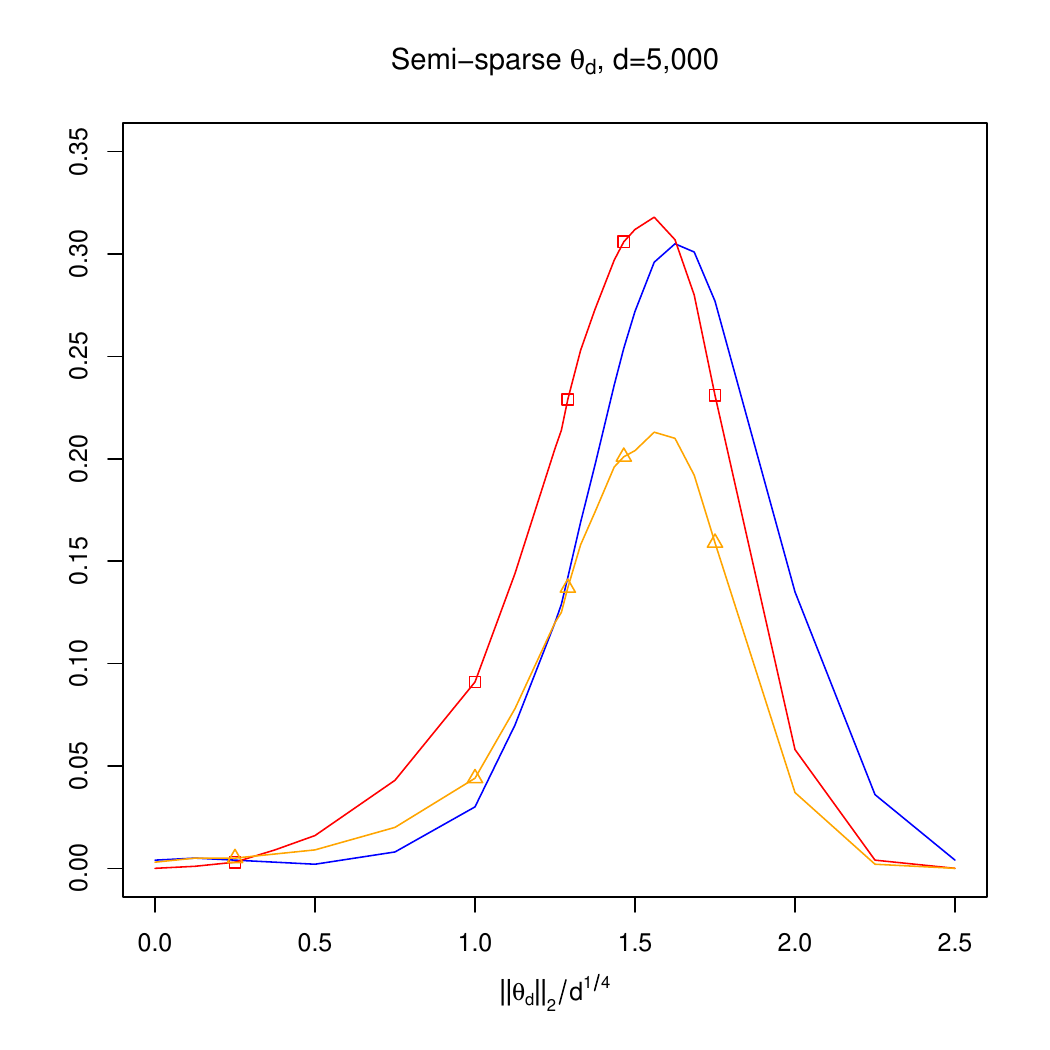}
\hspace{-0.6cm}
\includegraphics[width=5.2cm]{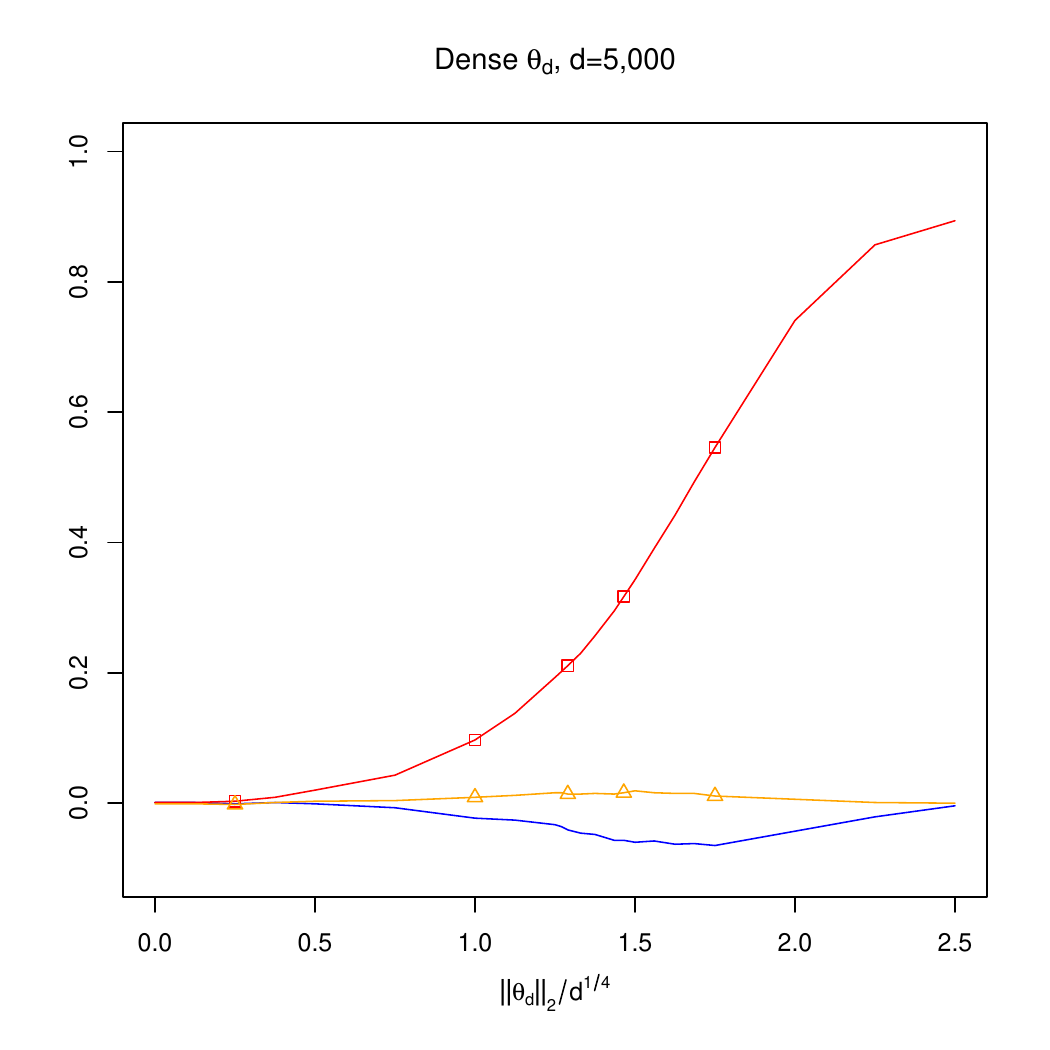}\null
\vspace{-0.3cm}
\includegraphics[width=5.2cm]{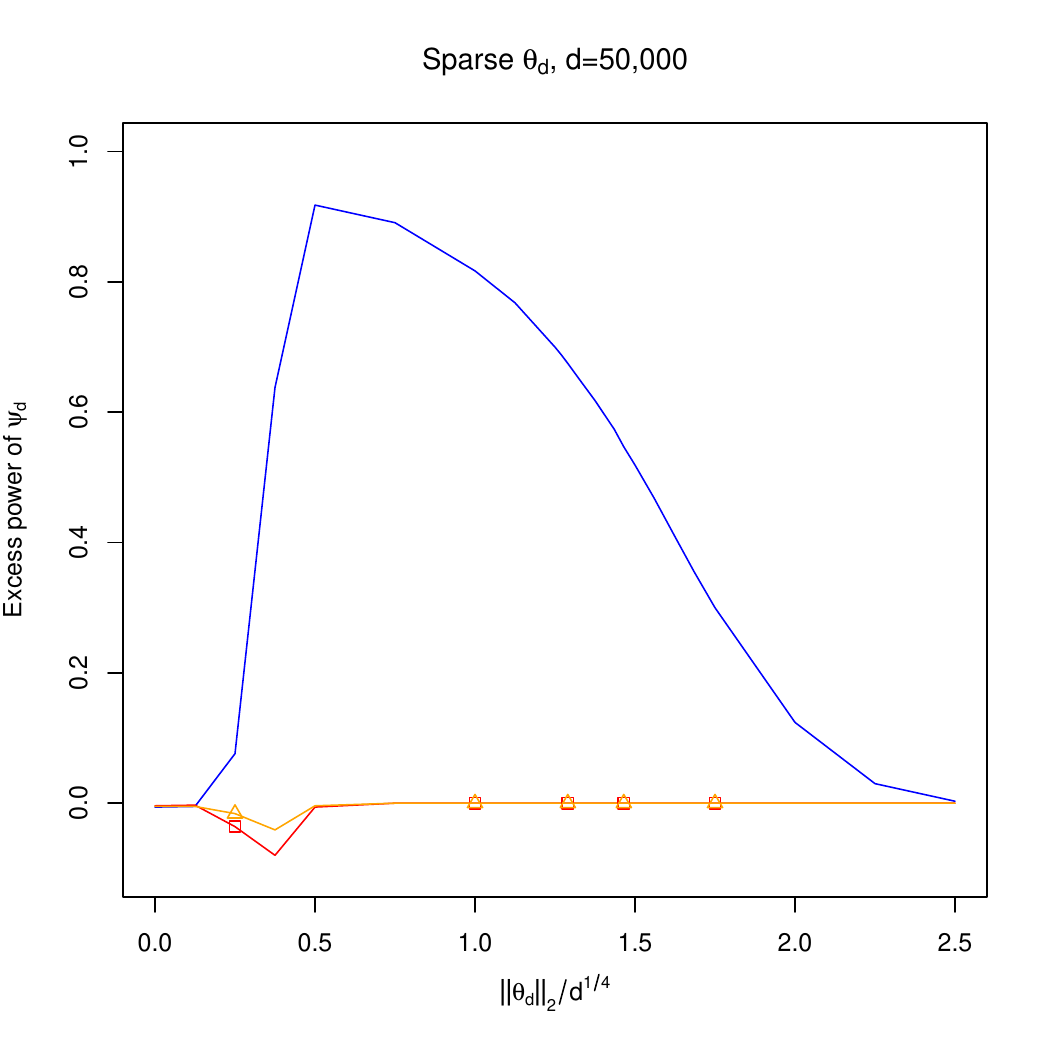}
\hspace{-0.6cm}
\includegraphics[width=5.2cm]{d=50000d_R=80MC=10000Diffsimsec.pdf}
\hspace{-0.6cm}
\includegraphics[width=5.2cm]{d=50000d_R=50000MC=10000Diffsimsec.pdf}
\null 
\vspace{-0.3cm}
\includegraphics[width=5.2cm]{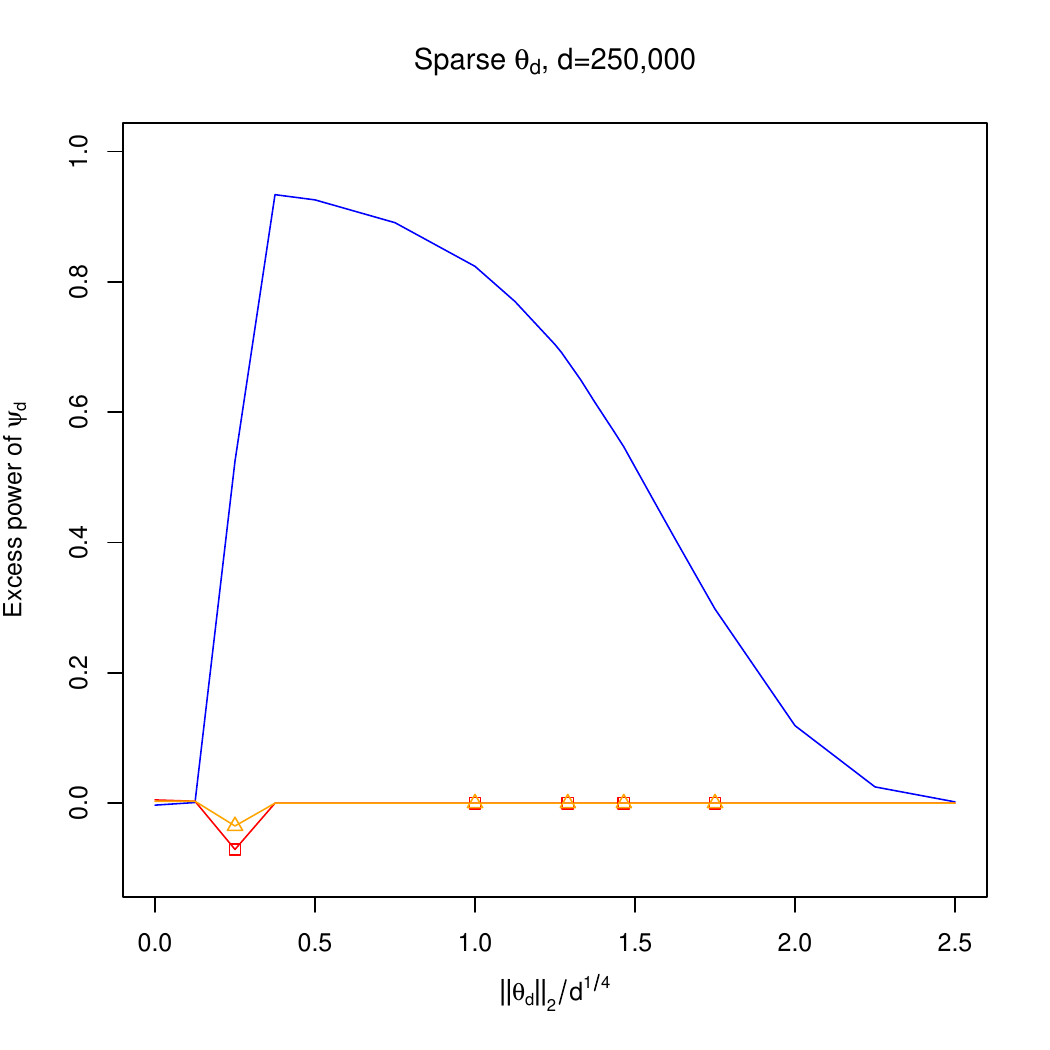}
\hspace{-0.6cm}
\includegraphics[width=5.2cm]{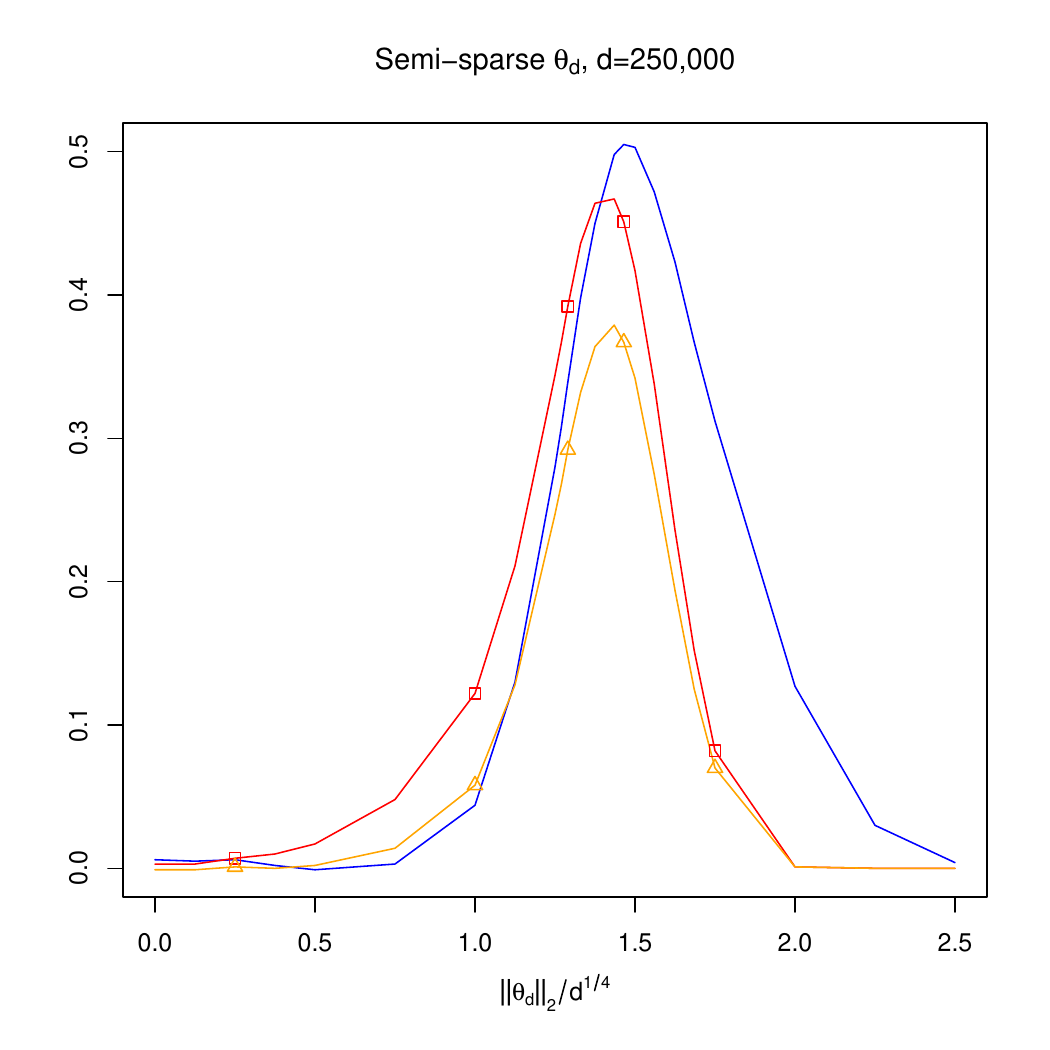}
\hspace{-0.6cm}
\includegraphics[width=5.2cm]{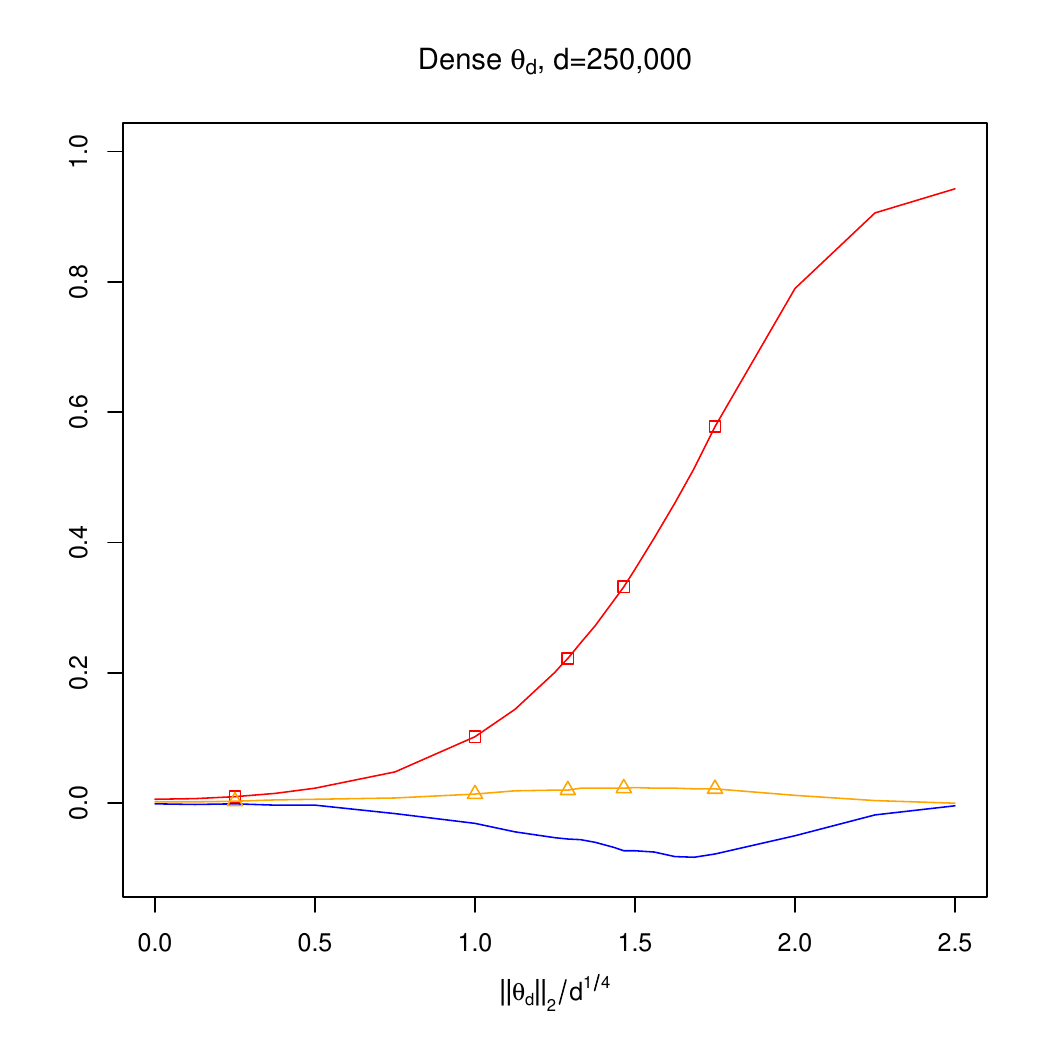}
\vspace{-0.9cm}
\end{center}
\caption{\scriptsize Excess power of~$\psi_d$ over each of the other tests studied as function of~$\enVert[0]{\bm{\theta}_d}_2/d^{1/4}$ for~$d=1{,}000$ (first row), $5{,}000$ (second row), $50{,}000$ (third row) and $250{,}000$ (fourth row) for sparse (first column), semi-sparse (second column) and dense (third column) alternatives. The tests compared to are~$\mathds{1}\del[0]{\enVert[0]{\bm{Z}_d}_p\geq \kappa_{d,p}}$ for~$p\in\cbr[0]{2,\infty}$, and the PE based test. Full implementation details are given in the body text.}
\label{fig:GaussDiff}
\end{figure}

\newpage

\begin{appendix}
	
\numberwithin{equation}{section}

\section{Notation}\label{app:Notation}
For brevity, we define~$\bm{Y}_{i,n}=h_{n}(\bm{X}_{i,n},\bm{\beta}_n^*)$ and write~$\bm{\mu}_{n}=\E\bm{Y}_{i,n}$ as well as~$\bm{\Sigma}_{n}=\bm{\Sigma}_{n}(\bm{\beta}_n^*)$. Write~$\bm{\theta}_{n}:=\sqrt{n}\bm{\Sigma}_{n}^{-1/2}\bm{\mu}_{n}$ and let~$\rightsquigarrow$ denote convergence in distribution. We shall also use that by Lemma C.1 in~\cite{kp2021consistency} for every~$p\in[2,\infty)$ there exists a constant~$c_p\in(0,\infty)$ such that
\begin{align}\label{eq:lambdaprop}
c_p^{-1}g_p(x)\leq \lambda_p(x)-\lambda_p(0)\leq c_pg_p(x)\qquad\text{for all }x\in\R,
\end{align}
where~$g_p$ and~$\lambda_p$ are defined in~\eqref{eqn:rhodef}. Furthermore, for every~$p \in [1, \infty)$ and for a matrix~$\bm{A}$ we denote by $\| \bm{A} \|_p$ the matrix norm induced by the vector~$p$-norm~$\|\cdot\|_p$.

\section{Proof of Theorem \ref{thm:pnorm-char_GMM}}

We start with Part 1~and note that $\bm{\beta}^* \in \mathbf{B}^{(0)}$ in this case due to Assumption~\ref{ass:pnormsize}, which also implies $\bm{\theta}_n = \bm{0}_d$. Fix~$\alpha \in (0, 1)$,~$p\in[2,\infty)$ and a sequence of real numbers~$\kappa_{n}$.	
First, observe that $\enVert[1]{{n}^{-1/2}\sum_{i=1}^{n}\hat{\bm{\Sigma}}_{n}^{-1/2}\bm{Y}_{i,n}}_p \geq \kappa_n$ is equivalent to~
\begin{align}\label{eqn:split}
\frac{\enVert[1]{{n}^{-1/2}\sum_{i=1}^{n}\hat{\bm{\Sigma}}_{n}^{-1/2}\bm{Y}_{i,n}}_p^p - \sum_{i = 1}^{d} \lambda_p(\theta_{i,n})}{\sqrt{d}\sigma_p} &+ \frac{\sum_{i = 1}^{d}\left( \lambda_p(\theta_{i,n}) - \lambda_p(0)\right)}{\sqrt{d}\sigma_p}\notag\\ 
&\geq 
\frac{\kappa_n^p-d\lambda_p(0)}{\sqrt{d}\sigma_p}=:\bar{\kappa}_n.
\end{align}
Under Assumption~\ref{ass:pnormsize}, the second term in the first line of~\eqref{eqn:split} vanishes and the first converges in distribution to a standard normal distribution, so that Polya's theorem shows that
\begin{align*}
\P\del[4]{\enVert[3]{n^{-1/2}\sum_{i=1}^n\hat{\bm{\Sigma}}^{-1/2}_{n}\bm{Y}_{i,n}}_p\geq \kappa_{n}}\to\alpha 
\end{align*}
if and only if~$\bar{\kappa}_n\to \Phi^{-1}(1-\alpha)=:z_{1-\alpha}$, which establishes Part 1~of Theorem~\ref{thm:pnorm-char_GMM}.

To establish Part 2~and~``$\Longrightarrow$'' in~\eqref{eq:conscharlambda} of Part~3, let $\bm{\beta}^* \in \mathbf{B}^*$ and consider (more generally than needed for Part 2) a subsequence~$n'$ of $n$ along which
\begin{align}\label{eq:locsubapp}
\frac{1}{\sqrt{d'}} \sum_{i = 1}^{d'} [\lambda_p(\theta_{i,n'})-\lambda_p(0)]\to c\in\R.
\end{align}
Then, by~\eqref{eq:pnormlocalhighlevel},~\eqref{eqn:split}, and~$\bar{\kappa}_n\to z_{1-\alpha}$, along this subsequence
\begin{align*}
\lim_{n' \to \infty} \P\del[4]{\enVert[3]{{n'}^{-1/2}\sum_{i=1}^{n'}\hat{\bm{\Sigma}}_{n'}^{-1/2}\bm{Y}_{i,n'}}_p \geq \kappa_{n'}} 
=1-\Phi(z_{1-\alpha}-c/\sigma_p)
< 1.
\end{align*}
Thus, Part~2~follows by the above with~``$n'=n$'' while ``$\Longrightarrow$'' in~\eqref{eq:conscharlambda} of Part~3~follows from observing that $\frac{1}{\sqrt{d}} \sum_{i = 1}^{d} [\lambda_p(\theta_{i,n})-\lambda_p(0)]\not\to \infty$ implies the existence of a subsequence along which \eqref{eq:locsubapp} holds for some~$c\in\R$.

We next establish~``$\Longleftarrow$'' of \eqref{eq:conscharlambda} of Part 3. Let $\bm{\beta}^* \in \mathbf{B}^*$. The rejection event in~\eqref{eqn:split}	 can be rewritten as
\begin{equation*}
\frac{\enVert[1]{{n}^{-1/2}\textstyle\sum_{i=1}^{n}\hat{\bm{\Sigma}}_{n}^{-1/2}\bm{Y}_{i,n}}_p^p - \sum_{i = 1}^{d} \lambda_p(\theta_{i,n})}{\sum_{i=1}^{d}[\lambda_p(\theta_{i,n})-\lambda_p(0)]} + 1 - \bar{\kappa}_n\frac{\sqrt{d}\sigma_p}{\sum_{i=1}^{d}[\lambda_p(\theta_{i,n})-\lambda_p(0)]}  \geq 0.
\end{equation*}
Hence, from~$d^{-1/2}\sum_{i=1}^{d}[\lambda_p(\theta_{i,n})-\lambda_p(0)]\to\infty$ and \eqref{eq:pnormconshighlevel} together with~$\bar{\kappa}_n\to \Phi^{-1}(1-\alpha)$, it follows that
\begin{align*}
\P\del[4]{\enVert[3]{n^{-1/2}\sum_{i=1}^n\hat{\bm{\Sigma}}^{-1/2}_{n}\bm{Y}_{i,n}}_p\geq \kappa_{n}}\to 1.
\end{align*}
Finally,~\eqref{eq:conschargp} of Part 3~follows from~\eqref{eq:lambdaprop}.

\section{Gaussian approximation over convex sets with estimated covariance matrix and~$d\to\infty$}
\subsection{Estimating $\mathbf{\Sigma}_n^{-1/2}$}
Recall that for any~$d\times d$ matrix~$\bm{A}$ we denote by~$\bm{A}^{-1}$ its Moore-Penrose generalized inverse and define~$\bm{A}^{-1/2}:=(\bm{A}^{1/2})^{-1}$, in case~$\bm{A}$ is symmetric and positive semidefinite, cf.~Footnote~\ref{fn:MPinv}. The following result is well known and immediately follows from, e.g., Equations (7.2.13) and  (5.8.4) in~\cite{horn2012matrix}. We provide the argument for completeness. Recall also that $\hat{\bm{\Sigma}}$ is positive semidefinite and symmetric by assumption.
\begin{lemma}\label{lem:covestim}
Let~$n'$ be a subsequence of~$n$, assume that the eigenvalues of~$\bm{\Sigma}_{n'}$ are (uniformly) bounded away from zero and from above, and that $\enVert[0]{\hat{\bm{\Sigma}}_{n'}-\bm{\Sigma}_{n'}}_2=O_\P(a_{n'})$ for some~$a_{n'}\to 0$. Then, $\enVert[1]{\hat{\bm{\Sigma}}^{-1/2}_{n'}-\bm{\Sigma}^{-1/2}_{n'}}_2
=
O_{\P}(a_{n'})$.
\end{lemma}
\begin{proof}
Write~$n$ instead of~$n'$. Since~$\enVert[0]{\hat{\bm{\Sigma}}_n-\bm{\Sigma}_{n}}_2\to 0$ in probability and the eigenvalues of $\bm{\Sigma}_n$ are uniformly bounded away from zero and infinity, $\hat{\bm{\Sigma}}_{n}$ is invertible with probability tending to one. Denote by~$G_n$ the event on which~$\hat{\bm{\Sigma}}_{n}$, and hence~$\hat{\bm{\Sigma}}_n^{1/2}$, is invertible. On~$G_n$ one has that~$\hat{\bm{\Sigma}}_n^{-1/2}$ is the regular matrix inverse of~$\hat{\bm{\Sigma}}_n^{1/2}$ and so
\begin{align*}
\enVert[1]{\hat{\bm{\Sigma}}^{-1/2}_{n}-\bm{\Sigma}^{-1/2}_{n}}_2
=
\enVert[1]{(\hat{\bm{\Sigma}}_n^{1/2})^{-1}-(\bm{\Sigma}_{n}^{1/2})^{-1}}_2\mathds{1}_{G_n}+\enVert[1]{\hat{\bm{\Sigma}}^{-1/2}_{n}-\bm{\Sigma}^{-1/2}_{n}}_2\mathds{1}_{G_n^c},
\end{align*}
where, with a slight overload of notation, the matrix inverse is now a regular inverse on~$G_n$. Equations (7.2.13) and  (5.8.4) in~\cite{horn2012matrix} guarantee that
\begin{align*}
\enVert[1]{\hat{\bm{\Sigma}}^{-1/2}_{n}-\bm{\Sigma}^{-1/2}_{n}}_2
\leq 
O_{\P}(a_n) +\enVert[1]{\hat{\bm{\Sigma}}^{-1/2}_{n}-\bm{\Sigma}^{-1/2}_{n}}_2\mathds{1}_{G_n^c}
,
\end{align*}
and the conclusion follows from~$\P(G_n^c)\to 0$.
\end{proof}

\subsection{An approximation result}
\begin{lemma}\label{eq:pnormerror}
Let~$p\in[2,\infty)$,~$n'$ be a subsequence of~$n$ and~$d'\to\infty$. Assume that the eigenvalues of~$\bm{\Sigma}_{n'}$ are (uniformly) bounded away from zero and from above and that $\enVert[0]{\hat{\bm{\Sigma}}_{n'}-\bm{\Sigma}_{n'}}_2=O_\P(a_{n'})$ for some~$a_{n'}\to 0$, as well as
\begin{align}\label{eq:GAlemmaass}
\sup_{t\in\R}\envert[4]{\P\del[3]{\enVert[3]{{n'}^{-1/2}\sum_{i=1}^{n'}\bm{\Sigma}_{n'}^{-1/2}\bm{Y}_{i,n'}}_p\leq t}-\P\del[1]{\|\bm{Z}_d+\bm{\theta}_{n}\|_p\leq t}}\to 0.
\end{align}
\begin{enumerate}
\item If $\frac{1}{d'} \sum_{i = 1}^{d'} g_p(\theta_{i,n'})$ is bounded, then
\begin{align*}
\enVert[3]{{n'}^{-1/2}\sum_{i=1}^{n'}\hat{\bm{\Sigma}}_{n'}^{-1/2}\bm{Y}_{i,n'}-{n'}^{-1/2}\sum_{i=1}^{n'}\bm{\Sigma}_{n'}^{-1/2}\bm{Y}_{i,n'}}_p
=
O_\P\del[1]{{d'}^{1/2}a_{n'}}.
\end{align*}
\item If $\frac{1}{d'} \sum_{i = 1}^{d'} g_p(\theta_{i,n'})\to\infty$, then
\begin{align*}
\enVert[3]{{n'}^{-1/2}\sum_{i=1}^{n'}\hat{\bm{\Sigma}}_{n'}^{-1/2}\bm{Y}_{i,n'}-{n'}^{-1/2}\sum_{i=1}^{n'}\bm{\Sigma}_{n'}^{-1/2}\bm{Y}_{i,n'}}_p
=
O_\P\del[4]{{d'}^{\frac{1}{2}-\frac{1}{p}}\sbr[3]{\sum_{i = 1}^{d'} g_p(\theta_{i,n'})}^{\frac{1}{p}}a_{n'}}.
\end{align*}
\end{enumerate}
\end{lemma}

\begin{proof}
Let~$p\in[2,\infty)$ and write~$n$ instead of~$n'$. Note that
\begin{equation}\label{eq:pnormdiff}
\begin{aligned}
& ~ \enVert[4]{\frac{1}{\sqrt{n}}\sum_{i=1}^{n}\hat{\bm{\Sigma}}_{n}^{-1/2}\bm{Y}_{i,n}
-
\frac{1}{\sqrt{n}}\sum_{i=1}^{n}\bm{\Sigma}_{n}^{-1/2}\bm{Y}_{i,n}}_p\\
\leq 
&~\enVert[2]{(\hat{\bm{\Sigma}}_{n}^{-1/2}-\bm{\Sigma}_{n}^{-1/2})\bm{\Sigma}_{n}^{1/2}}_p\enVert[4]{\frac{1}{\sqrt{n}}\sum_{i=1}^{n}\bm{\Sigma}_{n}^{-1/2}\bm{Y}_{i,n}}_p\\
\leq 
&~ d^{\frac{1}{2}-\frac{1}{p}}\enVert[2]{\bm{\Sigma}_{n}^{1/2}}_2\enVert[2]{\hat{\bm{\Sigma}}_{n}^{-1/2}-\bm{\Sigma}_{n}^{-1/2}}_2\enVert[4]{\frac{1}{\sqrt{n}}\sum_{i=1}^{n}\bm{\Sigma}_{n}^{-1/2}\bm{Y}_{i,n}}_p \\
=
&~ d^{\frac{1}{2}-\frac{1}{p}} O_\P(a_n)\enVert[4]{\frac{1}{\sqrt{n}}\sum_{i=1}^{n}\bm{\Sigma}_{n}^{-1/2}\bm{Y}_{i,n}}_p
\end{aligned}	
\end{equation}
where the second inequality follows from (1.2) of Theorem 1 in~\cite{goldberg1987equivalence} and the equality follows from
Lemma~\ref{lem:covestim} together with the uniform boundedness assumption on the eigenvalues of~$\bm{\Sigma}_n$. To proceed, observe that by~\eqref{eq:lambdaprop}
\begin{align}\label{eq:lambdatheta}
\sum_{i = 1}^{d} \lambda_p(\theta_{i,n})
=
\sum_{i = 1}^{d} [\lambda_p(\theta_{i,n})-\lambda_p(0)]+d \lambda_p(0)
\leq
c_p\sum_{i = 1}^{d}g_p(\theta_{i,n})+d\lambda_p(0).
\end{align} 

Consider first the case where~$d^{-1/2}\sum_{i = 1}^{d} g_p(\theta_{i,n})$ is bounded, which is only possible in Part 1: Then, Lemma~C.3 in~\cite{kp2021consistency} gives
\begin{equation*}
\frac{\|\bm{Z}_{d}+\bm{\theta}_{n}\|_p^p - \sum_{i = 1}^{d} \lambda_p(\theta_{i,n})}{\sqrt{d \mathbb{V}ar|Z|^p}} \rightsquigarrow \mathsf{N}_1(0, 1),
\end{equation*}
which together with~\eqref{eq:GAlemmaass} yields [use ``$t=\del[0]{x\cdot\sqrt{d \mathbb{V}ar|Z|^p}+\sum_{i = 1}^{d} \lambda_p(\theta_{i,n})}^{1/p}$'' and~$x\in\R$] 
\begin{align*}
\frac{\enVert[1]{{n}^{-1/2}\sum_{i=1}^{n}\bm{\Sigma}_{n}^{-1/2}\bm{Y}_{i,n}}_p^p-\sum_{i = 1}^{d} \lambda_p(\theta_{i,n})}{\sqrt{d \mathbb{V}ar|Z|^p}}
\rightsquigarrow \mathsf{N}_1(0, 1).
\end{align*}
Therefore, by~\eqref{eq:lambdatheta}, and exploiting boundedness of~$d^{-1/2} \sum_{i = 1}^{d} g_p(\theta_{i,n})$,
\begin{align*}
\enVert[4]{{n}^{-1/2}\sum_{i=1}^{n}\bm{\Sigma}_{n}^{-1/2}\bm{Y}_{i,n}}_p
=
O_{\P}\del[4]{\sbr[3]{\sum_{i = 1}^{d} \lambda_p(\theta_{i,n})}^{1/p}+\sbr[2]{{d \mathbb{V}ar|Z|^p}}^{1/2p}}
=
O_\P(d^{1/p}),
\end{align*}
which together with~\eqref{eq:pnormdiff} implies Part 1 when~$d^{-1/2}\sum_{i = 1}^{d} g_p(\theta_{i,n})$ is bounded.

Consider next the case where~$d^{-1/2} \sum_{i = 1}^{d} g_p(\theta_{i,n})\to\infty$ (potentially along a subsequence $n''$ of $n'$, which we again denote by~$n$ for simplicity of notation), a condition that is always satisfied in Part 2: By (C.23) in~\cite{kp2021consistency},
\begin{align*}
\frac{\enVert[1]{\bm{Z}_{d}+\bm{\theta}_n}_p^p-\sum_{i=1}^{d}\lambda_p(\theta_{i,n})}{\sum_{i=1}^{d}g_p(\theta_{i,n})}
=
o_{\P}(1),
\end{align*}
which together with~\eqref{eq:GAlemmaass} yields
\begin{align*}
\frac{\enVert[1]{{n}^{-1/2}\sum_{i=1}^{n}\bm{\Sigma}_{n}^{-1/2}\bm{Y}_{i,n}}_p^p	-\sum_{i=1}^{d}\lambda_p(\theta_{i,n})}{\sum_{i=1}^{d}g_p(\theta_{i,n})}
=
o_{\P}(1).
\end{align*}
Hence,
\begin{equation}\label{eqn:bd2}
\enVert[4]{{n}^{-1/2}\sum_{i=1}^{n}\bm{\Sigma}_{n}^{-1/2}\bm{Y}_{i,n}}_p
=
O_{\P}\del[4]{\sbr[3]{\sum_{i = 1}^{d} \lambda_p(\theta_{i,n})}^{1/p}+\sbr[3]{\sum_{i=1}^{d}g_p(\theta_{i,n})}^{1/p}}.
\end{equation}
Together with~\eqref{eq:lambdatheta} this yields 
\begin{equation}\label{eqn:aux3}
\enVert[4]{{n}^{-1/2}\sum_{i=1}^{n}\bm{\Sigma}_{n}^{-1/2}\bm{Y}_{i,n}}_p
=
\begin{cases}
O_\P\del[0]{d^{1/p}}\quad &\text{if}\quad \frac{1}{d} \sum_{i = 1}^{d} g_p(\theta_{i,n})\text{ is bounded}\\
O_\P\del[1]{[\sum_{i = 1}^{d} g_p(\theta_{i,n})]^{1/p}}\quad &\text{if}\quad\frac{1}{d} \sum_{i = 1}^{d} g_p(\theta_{i,n})\to\infty.
\end{cases}
\end{equation}
Together with~\eqref{eq:pnormdiff}, the first case in \eqref{eqn:aux3} now also establishes Part 1 in the remaining case where~$d^{-1/2} \sum_{i = 1}^{d} g_p(\theta_{i,n})\to\infty$  (along subsequences). Furthermore, the second case in \eqref{eqn:aux3} establishes Part 2.
\end{proof}

\subsection{Gaussian approximation}
For any non-empty $A\subseteq \R^d$,~$\eps>0$, and $p \in [1, \infty]$, denote $$A^{\eps}=\cbr[0]{\bm{x}\in\R^d:\inf_{\bm{y}\in A}\|\bm{x}-\bm{y}\|_2\leq \eps} \quad \text{ and } \quad \mc{B}_p(\bm{x},\eps)=\cbr[1]{\bm{y}\in\R^d:\|\bm{y}-\bm{x}\|_p\leq \eps}.$$ We abbreviate~$\mc{B}_p(\eps)=\mc{B}_p(\bm{0}_d,\eps)$,  and set $$A^{-\eps}=\cbr[0]{\bm{x}\in\R^d:\mc{B}_2(\bm{x},\eps) \subseteq A}.$$ Recall that~$\mc{C}_n=\cbr[0]{C\subseteq \R^{d(n)}: C\text{ is convex and Borel measurable}}$.

\begin{lemma}\label{lem:convgauss}
Let~$n'$ be a subsequence of~$n$ and~$d'\to\infty$. Assume that the eigenvalues of~$\bm{\Sigma}_{n'}$ are (uniformly) bounded away from zero and from above and that $\enVert[0]{\hat{\bm{\Sigma}}_{n'}-\bm{\Sigma}_{n'}}_2=O_\P(a_{n'})$ with~${d'}^{3/4}a_{n'}\to 0$ as well as
\begin{align}\label{eq:convass}
\sup_{C\in\mc{C}_{n'}}\envert[3]{\P\del[3]{{n'}^{-1/2}\sum_{i=1}^{n'}\bm{\Sigma}_{n'}^{-1/2}\bm{Y}_{i,n'}\in C}-\P\del[1]{\bm{Z}_{d'}+\bm{\theta}_{n'}\in C}}
\to
0.
\end{align}
Then, if~$\frac{1}{d'} \sum_{i = 1}^{d'} g_2(\theta_{i,n'})=\frac{1}{d'}\|\bm{\theta}_{n'}\|_2^2$ is bounded,  it holds that
\begin{align}\label{eq:GAconv}
\sup_{C\in\mc{C}_{n'}}\envert[3]{\P\del[3]{{n'}^{-1/2}\sum_{i=1}^{n'}\hat{\bm{\Sigma}}_{n'}^{-1/2}\bm{Y}_{i,n'}\in C}-\P\del[1]{\bm{Z}_{d'}+\bm{\theta}_{n'}\in C}}
\to
0.
\end{align}
\end{lemma}
\begin{proof}
The quantity inside the supremum in~\eqref{eq:GAconv} vanishes for~$C=\emptyset$. Thus, consider a sequence~$C_n$ with~$C_n \in \mc{C}_n$ and $C_n \neq \emptyset$ for every $n$. Write~$n$ instead of~$n'$ and let
\begin{align*}
A_n:=\enVert[3]{n^{-1/2}\sum_{i=1}^n\hat{\bm{\Sigma}}_n^{-1/2}\bm{Y}_{i,n}-n^{-1/2}\sum_{i=1}^n\bm{\Sigma}_{n}^{-1/2}\bm{Y}_{i,n}}_{2}
\end{align*}
Since~$d^{-1}\sum_{i = 1}^{d} g_2(\theta_{i,n})$ is bounded, Part 1~of Lemma~\ref{eq:pnormerror} implies that for every~$\delta\in(0,1)$ there exists an~$M=M_\delta>0$ such that for~$\eps_n=Md^{1/2}a_n$ one has~$\P(A_n> \eps_n)\leq \delta$ for all~$n$. Therefore, we have
\begin{align*}
\cbr[4]{{n^{-1/2}\sum_{i=1}^n\hat{\bm{\Sigma}}_n^{-1/2}\bm{Y}_{i,n}\in C_n}}
\subseteq
\cbr[4]{{n^{-1/2}\sum_{i=1}^n\bm{\Sigma}_n^{-1/2}\bm{Y}_{i,n}\in C_n^{\eps_n}}}
\cup \cbr[1]{A_n> \eps_n}.
\end{align*}	
Denoting the supremum in~\eqref{eq:convass} by~$r_n$ it thus follows from convexity of~$C_n^{\eps_n}$ that
\begin{align*}
\P\del[2]{{n^{-1/2}\sum_{i=1}^n\hat{\bm{\Sigma}}_n^{-1/2}\bm{Y}_{i,n}\in C_n}}
\leq
\P\del[1]{\bm{Z}_d+\bm{\theta}_n \in C_n^{\eps_n}}+r_n +\delta.	
\end{align*}
Next, by~$C_n^{\eps_n}-\bm{\theta}_n=(C_n-\bm{\theta}_n)^{\eps_n}$, convexity of~$C_n-\bm{\theta}_n$,  and Lemma 2.6 in \cite{bentkus2003dependence} (cf.~also \cite{nazarov2003maximal}), it eventually holds that
\begin{align*}
\P\del[1]{\bm{Z}_d+\bm{\theta}_n \in C_n^{\eps_n}}
&=
\P\del[1]{\bm{Z}_d \in C_n^{\eps_n}-\bm{\theta}_n}\\
&=
\P\del[1]{\bm{Z}_d \in (C_n-\bm{\theta}_n)^{\eps_n}}\\
&\leq
\P\del[1]{\bm{Z}_d \in C_n-\bm{\theta}_n}+d^{1/4}\eps_n\\
&=
\P\del[1]{\bm{Z}_d+\bm{\theta}_n \in C_n}
+
d^{1/4}\eps_n.
\end{align*}
Therefore,
\begin{align*}
\limsup_{n\to\infty}\sbr[3]{\P\del[3]{{n^{-1/2}\sum_{i=1}^n\hat{\bm{\Sigma}}_n^{-1/2}\bm{Y}_{i,n}\in C_n}}-\P\del[1]{\bm{Z}_d+\bm{\theta}_n \in C_n}}
\leq
\delta.
\end{align*}
On the other hand, it also holds that
\begin{align*}
\cbr[4]{{n^{-1/2}\sum_{i=1}^n\bm{\Sigma}_n^{-1/2}\bm{Y}_{i,n}\in C_n^{-\eps_n}}}	
\subseteq
\cbr[4]{{n^{-1/2}\sum_{i=1}^n\hat{\bm{\Sigma}}_n^{-1/2}\bm{Y}_{i,n}\in C_n}}	
\cup
\cbr[1]{A_n> \eps_n}.
\end{align*}
Hence, by convexity of~$C_n^{-\eps_n}$ and recalling that~$r_n$ denotes the supremum in~\eqref{eq:convass} (with~$r_n=0$ for~$C_n^{-\eps_n}=\emptyset$)
\begin{align*}
\P\del[1]{\bm{Z}_d+\bm{\theta}_n \in C_n^{-\eps_n}}-r_n \leq \P\del[2]{{n^{-1/2}\sum_{i=1}^n\hat{\bm{\Sigma}}_n^{-1/2}\bm{Y}_{i,n}\in C_n}}+\delta	
\end{align*}
Next, by~$C_n^{-\eps_n}-\bm{\theta}_n=(C_n-\bm{\theta}_n)^{-\eps_n}$ (which remains true for~$C_n^{-\eps_n}=\emptyset$), convexity of~$C_n-\bm{\theta}_n$ and Lemma 2.6 \cite{bentkus2003dependence} it eventually holds that
\begin{align*}
\P\del[1]{\bm{Z}_d+\bm{\theta}_n \in C_n^{-\eps_n}}
&=
\P\del[1]{\bm{Z}_d \in C_n^{-\eps_n}-\bm{\theta}_n}\\
&=
\P\del[1]{\bm{Z}_d \in (C_n-\bm{\theta}_n)^{-\eps_n}}\\
&\geq
\P\del[1]{\bm{Z}_d \in C_n-\bm{\theta}_n}-d^{1/4}\eps_n\\
&=
\P\del[1]{\bm{Z}_d+\bm{\theta}_n \in C_n}-d^{1/4}\eps_n.
\end{align*}
Therefore,
\begin{align*}
\liminf_{n\to\infty}\sbr[3]{\P\del[3]{{n^{-1/2}\sum_{i=1}^n\hat{\bm{\Sigma}}_n^{-1/2}\bm{Y}_{i,n}\in C_n}}-
\P\del[1]{\bm{Z}_d+\bm{\theta}_n \in C_n}}
\geq 
-\delta.
\end{align*}
Thus,~\eqref{eq:GAconv} follows because the sequence~$C_n \neq \emptyset$ and~$\delta\in(0,1)$ were arbitrary.
\end{proof}

\section{Proof of Theorem \ref{thm:prim}}
It will be repeatedly used that by~\eqref{eq:lambdaprop}, for every~$p\in[2,\infty)$, there exists a~$c_p\in(0,\infty)$ such that
\begin{equation}\label{eqn:uplow}
c_p^{-1}\sum_{i=1}^dg_p(\theta_{i,n})\leq 
\sum_{i=1}^d[\lambda_p(\theta_{i,n})-\lambda_p(0)]\leq c_p\sum_{i=1}^dg_p(\theta_{i,n}).
\end{equation}
Fix~$p\in[2,\infty)$.

Assumption~\ref{ass:pnormsize}, in case we additionally have $\bm{\beta}^* \in \mathbf{B}^{(0)}$, and, more generally, \eqref{eq:pnormlocalhighlevel} of Assumption \ref{ass:pnormconsistency} are verified by combining Lemma~\ref{lem:convgauss} and 
\begin{equation*}
\frac{\|\bm{Z}_{d'}+\bm{\theta}_{n'}\|_p^p - \sum_{i = 1}^{d'} \lambda_p(\theta_{i,n'})}{\sqrt{d' \mathbb{V}ar|Z|^p}} \rightsquigarrow \mathsf{N}_1(0, 1),
\end{equation*}
for every subsequence $n'$ of $n$ along which $(1/\sqrt{d}) \sum_{i=1}^{d}[\lambda_p(\theta_{i,n})-\lambda_p(0)]$ (and hence also ${d}^{-1}\sum_{i=1}^{d}g_2(\theta_{i,n})$) is bounded --- the previous display following from Lemma~C.3 in~\cite{kp2021consistency}.

To verify~\eqref{eq:pnormconshighlevel} of Assumption~\ref{ass:pnormconsistency}, let~${d}^{-1/2} \sum_{i = 1}^{d} [\lambda_p(\theta_{i,n})-\lambda_p(0)]\to\infty$, 
so that by (C.23) in \cite{kp2021consistency} together with \eqref{eqn:uplow}
\begin{align}\label{eq:op1}
\frac{\enVert[0]{\bm{Z}_{d}+\bm{\theta}_n}_p^p-\sum_{i=1}^{d}\lambda_p(\theta_{i,n})}{\sum_{i=1}^{d}[\lambda_p(\theta_{i,n})-\lambda_p(0)]}
=
o_{\P}(1).
\end{align}
We show that for any subsequence $n'$ of $n$ there exists a subsequence $n''$ of $n'$ along which~\eqref{eq:pnormconshighlevel} holds:
Pick a subsequence~$n'$ of~$n$. 

\underline{Case 1:} Suppose~$n'$ possesses a subsequence $n''$, say, along which $d^{-1} \sum_{i = 1}^{d} [\lambda_p(\theta_{i,n})-\lambda_p(0)]$ (and hence also~${d}^{-1}\sum_{i=1}^dg_2(\theta_{i,n})$) is bounded,~then~\eqref{eq:pnormconshighlevel} follows along~$n''$ by combining Lemma \ref{lem:convgauss} with~\eqref{eq:op1}. 

\underline{Case 2:} Assume~${d'}^{-1} \sum_{i = 1}^{d'} [\lambda_p(\theta_{i,n'})-\lambda_p(0)]\to\infty$. In this case, set $n'' = n'$ and write~$n$ instead of~$n'$. We first observe that for any~$\delta\in(0,1)$,
\begin{align*}
\cbr[4]{\frac{\enVert[1]{n^{-1/2}\sum_{i=1}^n\hat{\bm{\Sigma}}_n^{-1/2}\bm{Y}_{i,n}}_p^p-\sum_{i=1}^{d}\lambda_p(\theta_{i,n})}{\sum_{i = 1}^{d} [\lambda_p(\theta_{i,n})-\lambda_p(0)]}
\leq 
\delta}
=
\cbr[2]{n^{-1/2}\sum_{i=1}^n\hat{\bm{\Sigma}}_n^{-1/2}\bm{Y}_{i,n}\in \mc{B}_p(r_{n,\delta})},
\end{align*}
where~$r_{n,\delta}=\left(\delta\sum_{i = 1}^{d} [\lambda_p(\theta_{i,n})-\lambda_p(0)]+\sum_{i=1}^{d}\lambda_p(\theta_{i,n})\right)^{1/p}$. Letting
\begin{align*}
B_n=\enVert[3]{{n}^{-1/2}\sum_{i=1}^{n}\hat{\bm{\Sigma}}_{n}^{-1/2}\bm{Y}_{i,n}-{n}^{-1/2}\sum_{i=1}^{n}\bm{\Sigma}_{n}^{-1/2}\bm{Y}_{i,n}}_p,
\end{align*}
it follows from Lemma \ref{eq:pnormerror} that for every~$\rho\in(0,1)$ there exists an~$M=M_\rho$ such that for~$\eps_n=M{d}^{\frac{1}{2}-\frac{1}{p}}a_{n}\sbr[1]{\sum_{i = 1}^{d} g_p(\theta_{i,n})}^{\frac{1}{p}}$ one has~$\P\del[0]{B_n> \eps_n}\leq \rho$. Furthermore, note that since~$r_{n,\delta}^p\geq \delta c_p^{-1}\sum_{i=1}^dg_p(\theta_{i,n})$ it holds that
\begin{align}\label{eq:epsdivr}
\frac{\eps_n}{r_{n,\delta}}
=
O({d}^{\frac{1}{2}-\frac{1}{p}}a_{n})
=
o(1),
\end{align}
such that, in particular,~$r_{n,\delta}-\eps_n$ is eventually strictly positive. Next, observe that (eventually)
\begin{align*}
\cbr[3]{\frac{1}{\sqrt{n}}\sum_{i=1}^n\bm{\Sigma}_n^{-1/2}\bm{Y}_{i,n}\in \mc{B}_p(r_{n,\delta}-\eps_n)}
\subseteq
\cbr[3]{\frac{1}{\sqrt{n}}\sum_{i=1}^n\hat{\bm{\Sigma}}_n^{-1/2}\bm{Y}_{i,n}\in \mc{B}_p(r_{n,\delta})}\cup\cbr[0]{B_n > \eps_n},
\end{align*}
and note that by assumption
\begin{align*}
\liminf_{n\to\infty}\P\del[3]{\frac{1}{\sqrt{n}}\sum_{i=1}^n\bm{\Sigma}_n^{-1/2}\bm{Y}_{i,n}\in \mc{B}_p(r_{n,\delta}-\eps_n)}	
=
\liminf_{n\to\infty}\P\del[1]{\bm{Z}_d+\bm{\theta}_n\in \mc{B}_p(r_{n,\delta}-\eps_n)}.
\end{align*}
Furthermore, 
by Bernoulli's inequality (eventually)
\begin{align*}
(r_{n,\delta}-\eps_n)^p
=
r_{n,\delta}^p(1-\eps_n/r_{n,\delta})^p
\geq
r_{n,\delta}^p(1-p\eps_n/r_{n,\delta})
\end{align*}
so that the right-hand side of the penultimate display is no smaller than
\begin{align*}
&\liminf_{n\to\infty}\P\del[1]{\bm{Z}_d+\bm{\theta}_n\in \mc{B}_p(r_{n,\delta}(1-p\eps_n/r_{n,\delta})^{1/p})} \\ = &
\liminf_{n\to\infty}\P \del[4]{\frac{\enVert[0]{\bm{Z}_d+\bm{\theta}_n}_p^p-\sum_{i=1}^{d}\lambda_p(\theta_{i,n})}{\sum_{i = 1}^{d} [\lambda_p(\theta_{i,n})-\lambda_p(0)]}<\delta -\frac{pr^p_{n,\delta}\eps_n}{r_{n,\delta}\sum_{i = 1}^{d} [\lambda_p(\theta_{i,n})-\lambda_p(0)]}}. 
\end{align*}
Because~$r_{n,\delta}^p/\sum_{i = 1}^{d} [\lambda_p(\theta_{i,n})-\lambda_p(0)]$ is bounded (recall that ~${d}^{-1} \sum_{i = 1}^{d} [\lambda_p(\theta_{i,n})-\lambda_p(0)]\to\infty$) and by~\eqref{eq:epsdivr}
\begin{align*}
\frac{pr^p_{n,\delta}\eps_n}{r_{n,\delta}\sum_{i = 1}^{d} [\lambda_p(\theta_{i,n})-\lambda_p(0)]}
=
O\del[1]{\eps_n/r_{n,\delta}}
=
o(1).
\end{align*}
Hence, by~\eqref{eq:op1} the penultimate display equals one and we conclude that
\begin{align}\label{eq:Part1}
\liminf_{n\to\infty}\P\del[4]{\frac{\enVert[1]{n^{-1/2}\sum_{i=1}^n\hat{\bm{\Sigma}}_n^{-1/2}\bm{Y}_{i,n}}_p^p-\sum_{i=1}^{d}\lambda_p(\theta_{i,n})}{\sum_{i = 1}^{d} [\lambda_p(\theta_{i,n})-\lambda_p(0)]}
\leq 
\delta}
\geq 
1-\rho.	
\end{align}

Next, since~$r_{n,-\delta}=\left(-\delta\sum_{i = 1}^{d} [\lambda_p(\theta_{i,n})-\lambda_p(0)]+\sum_{i=1}^{d}\lambda_p(\theta_{i,n})\right)^{1/p}>0$ for~$\delta\in(0,1)$,
\begin{align*}
	\cbr[4]{\frac{\enVert[1]{n^{-1/2}\sum_{i=1}^n\hat{\bm{\Sigma}}_n^{-1/2}\bm{Y}_{i,n}}_p^p-\sum_{i=1}^{d}\lambda_p(\theta_{i,n})}{\sum_{i = 1}^{d} [\lambda_p(\theta_{i,n})-\lambda_p(0)]}
		\leq 
		-\delta}
	=
	\cbr[2]{n^{-1/2}\sum_{i=1}^n\hat{\bm{\Sigma}}_n^{-1/2}\bm{Y}_{i,n}\in \mc{B}_p(r_{n,-\delta})},
\end{align*}
which is contained in
\begin{align*}
	\cbr[3]{\frac{1}{\sqrt{n}}\sum_{i=1}^n\bm{\Sigma}_n^{-1/2}\bm{Y}_{i,n}\in \mc{B}_p(r_{n,-\delta}+\eps_n)}\cup\cbr[0]{B_n > \eps_n}.
\end{align*}
By assumption
\begin{equation}\label{eqn:compeq}
	\limsup_{n\to\infty}\P\del[3]{\frac{1}{\sqrt{n}}\sum_{i=1}^n\bm{\Sigma}_n^{-1/2}\bm{Y}_{i,n}\in \mc{B}_p(r_{n,-\delta}+\eps_n)}
	=
	\limsup_{n\to\infty}\P\del[1]{\bm{Z}_d+\bm{\theta}_n\in \mc{B}_p(r_{n,-\delta}+\eps_n)}.
\end{equation}
Furthermore, since~$r_{n,-\delta}^p\geq (1-\delta) c_p^{-1}\sum_{i=1}^dg_p(\theta_{i,n})$ it holds that
\begin{align}\label{eq:epsdivr2}
	\frac{\eps_n}{r_{n,-\delta}}
	=
	O({d}^{\frac{1}{2}-\frac{1}{p}}a_{n})
	=
	o(1).
\end{align}
Thus, by the mean-value theorem, it eventually holds that
\begin{align*}
	(r_{n,-\delta}+\eps_n)^p
	=	
	r_{n,-\delta}^p(1+\eps_n/r_{n,-\delta})^p
	\leq 
	r_{n,-\delta}^p(1+2p\eps_n/r_{n,-\delta}),
\end{align*}
such that the right-hand side in~\eqref{eqn:compeq} is no greater than
\begin{align*}
	&\limsup_{n\to\infty}\P\del[1]{\bm{Z}_d+\bm{\theta}_n\in \mc{B}_p(r_{n,-\delta}(1+2p\eps_n/r_{n,-\delta})^{1/p}}\\
	=
	&\limsup_{n\to\infty}\P\del[4]{\frac{\enVert[1]{\bm{Z}_d+\bm{\theta}_n}_p^p-\sum_{i=1}^{d}\lambda_p(\theta_{i,n})}{\sum_{i = 1}^{d} [\lambda_p(\theta_{i,n})-\lambda_p(0)]}\leq -\delta + \frac{2pr_{n,-\delta}^p\eps_n}{r_{n,-\delta}\sum_{i = 1}^{d} [\lambda_p(\theta_{i,n})-\lambda_p(0)]}}.
\end{align*}
Because~$r_{n,-\delta}^p/\sum_{i = 1}^{d} [\lambda_p(\theta_{i,n})-\lambda_p(0)]$ is bounded (recall that ~${d}^{-1} \sum_{i = 1}^{d} [\lambda_p(\theta_{i,n})-\lambda_p(0)]\to\infty$)  and by~\eqref{eq:epsdivr2}
\begin{align*}
	\frac{2pr_{n,-\delta}^p\eps_n}{r_{n,-\delta}\sum_{i = 1}^{d} [\lambda_p(\theta_{i,n})-\lambda_p(0)]}
	=
	O\del[1]{\eps_n/r_{n,-\delta}}
	=
	o(1).
\end{align*}
Hence, the expression in the penultimate display equals zero by~\eqref{eq:op1} and we conclude that
\begin{equation}\label{eq:Part2}
	\limsup_{n\to\infty}\P\del[4]{\frac{\enVert[1]{n^{-1/2}\sum_{i=1}^n\hat{\bm{\Sigma}}_n^{-1/2}\bm{Y}_{i,n}}_p^p-\sum_{i=1}^{d}\lambda_p(\theta_{i,n})}{\sum_{i = 1}^{d} [\lambda_p(\theta_{i,n})-\lambda_p(0)]} \leq -\delta}
	\leq 
	\rho.
\end{equation}
The statement in \eqref{eq:pnormconshighlevel} of Assumption~\ref{ass:pnormconsistency} now follows, because Equations~\eqref{eq:Part1} and~\eqref{eq:Part2} imply that for every~$\delta \in (0, 1)$ and every $\rho\in(0,1)$ it holds that
\begin{equation*}
\limsup_{n \to \infty} \P\del[4]{\left | \frac{\enVert[1]{n^{-1/2}\sum_{i=1}^n\hat{\bm{\Sigma}}_n^{-1/2}\bm{Y}_{i,n}}_p^p-\sum_{i=1}^{d}\lambda_p(\theta_{i,n})}{\sum_{i = 1}^{d} [\lambda_p(\theta_{i,n})-\lambda_p(0)]} \right| > \delta} \leq 2 \rho.
\end{equation*}

\section{Proof of Corollary~\ref{cor:fourmoments}}

Concerning Condition 1 in Assumption~\ref{as:suffc}, consider first the case where $\hat{\bm{\Sigma}}_n(\bm{\beta}_n^*)$ is the estimator~$\hat{\bm{\Sigma}}_{\eta,\delta}$ from Theorem 1 of~\cite{abdalla2022covariance} for~$\eta=0$,~$\delta=1/d$, and~$p = 4$ applied to the i.i.d.~sample
\begin{equation*}
	(\bm{Y}_{2,n}-\bm{Y}_{1,n})/\sqrt{2},\hdots,(\bm{Y}_{2 \lfloor n/2 \rfloor,n}-\bm{Y}_{2 \lfloor n/2 \rfloor-1,n})/\sqrt{2}
\end{equation*}
of size~$\lfloor n /2 \rfloor$ with covariance matrix~$\bm{\Sigma}_n$ and mean vector zero; cf.~also Section 1.2 in~\cite{mendelson2020robust}. We apply Theorem 1 of~\cite{abdalla2022covariance} to show that~$\|\hat{\bm{\Sigma}}_{0,d^{-1}}-\bm{\Sigma}_{n}\|_2=O_\P\del[1]{\sqrt{d/n}}$, which then verifies Condition 1 of Assumption~\ref{as:suffc} with~$a_n=\sqrt{d/n}$, because~$d/n^{2/5}\to 0$ by assumption. To this end, we only need to show that the condition in Equation (1) of~\cite{abdalla2022covariance} holds for $p = 4$ (note that we only have to verify that condition for $p = q = 4$ due to Lyapunov's inequality): By the Binomial theorem and Assumption~\ref{ass:boundedkurt}, for every~$t\in\R^d$ and writing $L = L(\bm{\beta}^*)$,
\begin{align*}
\E \langle \bm{Y}_{2,n}-\bm{Y}_{1,n},t\rangle^4 
=
\E \langle (\bm{Y}_{2,n}-\bm{\mu}_n)-(\bm{Y}_{1,n}-\bm{\mu}_n),t\rangle^4
\leq 
(2L^4+6)\del[1]{\E \langle \bm{Y}_{1,n}-\bm{\mu}_n,t\rangle^2}^2.
\end{align*}
In addition,
\begin{align*}
\E \langle \bm{Y}_{1,n}-\bm{\mu}_n,t\rangle^2
=
t'\bm{\Sigma}_nt
=
\frac{1}{2}t'\E \del[1]{\bm{Y}_{2,n}-\bm{Y}_{1,n}}\del[1]{\bm{Y}_{2,n}-\bm{Y}_{1,n}}'t
=
\E\langle (\bm{Y}_{2,n}-\bm{Y}_{1,n})/\sqrt{2},t\rangle^2,
\end{align*}
such that~$(\bm{Y}_{2,n}-\bm{Y}_{1,n})/\sqrt{2}$ satisfies the condition in their Equation (1) with~$\kappa(4)=[(2L^4+6)/4]^{1/4}$. This also shows that Condition 1 of Assumption~\ref{as:suffc} holds if $\hat{\bm{\Sigma}}_n(\bm{\beta}_n^*)$ is the estimator in Theorem 1.3 of~\cite{oliveira2022improved} (with their $\eta = 0$, $\alpha = 1/d$ and $p = 4$ and applied to the auxiliary sample as above).

Concerning Condition 2 of Assumption~\ref{as:suffc}, we use Theorem 2.1 in~\cite{fang2020large} to establish that
\begin{align*}
\sup_{C\in\mc{C}_n}\envert[3]{\P\del[2]{n^{-1/2}\sum_{i=1}^n\bm{\Sigma}_n^{-1/2}[\bm{Y}_{i,n}-\bm{\mu}_n] \in  C}-\P\del[1]{\bm{Z}_d \in  C}}\to 0;
\end{align*}
that is their~$\xi_i=n^{-1/2}\bm{\Sigma}_n^{-1/2}[\bm{Y}_{i,n}-\bm{\mu}_n]$ and their~$\Sigma=\mathbf{I}_d$. Thus, the Cauchy-Schwarz inequality yields that their
\begin{align*}
\E\|\Sigma^{-1/2}\xi_1\|_2^4
=
\E\|\xi_1\|_2^4
=
\E \del[3]{\sum_{j=1}^d\xi_{1,j}^2}^2
\leq
d\sum_{j=1}^d\E\xi_{1,j}^4.
\end{align*}
In addition, letting~$e_j$ be the~$j$th canonical basis vector of~$\R^d$, it follows from Assumption~\ref{ass:boundedkurt} that for~$j=1,\hdots,d$ and $L = L(\bm{\beta}^*)$ as above,
\begin{align*}
\E\xi_{1,j}^4
=
n^{-2}\E\langle (\bm{Y}_{i,n}-\bm{\mu}_n), \bm{\Sigma}_n^{-1/2} e_j\rangle^4	
\leq
L^4n^{-2}\del[2]{\E\langle \bm{\Sigma}_n^{-1/2}(\bm{Y}_{i,n}-\bm{\mu}_n),e_j\rangle^2}^2
=
L^4n^{-2}.
\end{align*}
Therefore,~$\sum_{i=1}^n\E\|\Sigma^{-1/2}\xi_i\|_2^4\leq \frac{L^4d^2}{n}$ and Theorem 2.1 in~\cite{fang2020large} together with~$\frac{d}{n^{2/5}}[\log(n)]^{2/5}\to 0$ yields that Condition 2 of Assumption~\ref{as:suffc} is satisfied, since
\begin{align*}
\sup_{C\in\mc{C}_n}\envert[3]{\P\del[2]{n^{-1/2}\sum_{i=1}^n\bm{\Sigma}_n^{-1/2}[\bm{Y}_{i,n}-\bm{\mu}_n] \in  C}-\P\del[1]{\bm{Z}_d \in  C}}
=
O\del[3]{\frac{d^{5/4}}{n^{1/2}}\sqrt{\log(n)}}
=
o(1).
\end{align*}

\section{Gaussian approximation over hyperrectangles with estimated covariance matrix and~$d\to\infty$}

We begin by formulating a version of Lemma~\ref{lem:convgauss} for  hyperrectangles~$\mc{H}_n$, cf.~Section~\ref{sec:suptest}. Although a Gaussian approximation result over~$\mc{H}_n$ follows directly from Lemma~\ref{lem:convgauss} and~$\mc{H}_n\subseteq \mc{C}_n$, an approximation result applying to~$\bm{\theta}_n$ further from zero, necessary for our purposes, can be obtained for~$\mc{H}_n$. The proof strategy is identical to that of Lemma~\ref{lem:convgauss}, but as other anti-concentration inequalities are employed we provide the proof for completeness. For any $A\subseteq \R^d$ and~$\eps>0$, let $$A^{\eps,\infty}=\cbr[0]{\bm{x}\in\R^d:\inf_{\bm{y}\in A}\|\bm{x}-\bm{y}\|_\infty\leq \eps}.$$ Furthermore, $$A^{-\eps,\infty}=\cbr[0]{\bm{x}\in\R^d:\mc{B}_\infty(\bm{x},\eps) \subseteq A} \quad \text{ and } \quad \mc{B}_{\infty}(\bm{x},\eps)=\cbr[1]{\bm{y}\in\R^d:\|\bm{y}-\bm{x}\|_{\infty} \leq \eps}.$$

\begin{lemma}\label{lem:hypergauss}
Let~$n'$ be a subsequence of~$n$ and~$d'\to\infty$. Assume that the eigenvalues of~$\bm{\Sigma}_{n'}$ are (uniformly) bounded away from zero and from above and that $\enVert[0]{\hat{\bm{\Sigma}}_{n'}-\bm{\Sigma}_{n'}}_2=O_\P(a_{n'})$ with~${d'}^{3/4}a_{n'}\to 0$ as well as
\begin{align}\label{eq:hyperass}
\sup_{C\in\mc{C}_{n'}}\envert[3]{\P\del[3]{{n'}^{-1/2}\sum_{i=1}^{n'}\bm{\Sigma}_{n'}^{-1/2}\bm{Y}_{i,n'}\in C}-\P\del[1]{\bm{Z}_{d'}+\bm{\theta}_{n'}\in C}}
\to
0.
\end{align}
Then, if~$\frac{1}{{d'}^{3/2}/\log(d')} \sum_{i = 1}^{d'} g_2(\theta_{i,n'})=\frac{1}{{d'}^{3/2}/\log(d')}\|\bm{\theta}_{n'}\|_2^2$ is bounded, it also holds that
\begin{align}\label{eq:GAhyper}
\sup_{H\in\mc{H}_{n'}}\envert[3]{\P\del[3]{{n'}^{-1/2}\sum_{i=1}^{n'}\hat{\bm{\Sigma}}_{n'}^{-1/2}\bm{Y}_{i,n'}\in H}-\P\del[1]{\bm{Z}_{d'}+\bm{\theta}_{n'}\in H}}
\to
0.
\end{align}
\end{lemma}
\begin{proof}
Write~$n$ instead of~$n'$ and let
\begin{align*}
A_n &:=\enVert[3]{n^{-1/2}\sum_{i=1}^n\hat{\bm{\Sigma}}_n^{-1/2}\bm{Y}_{i,n}-n^{-1/2}\sum_{i=1}^n\bm{\Sigma}_{n}^{-1/2}\bm{Y}_{i,n}}_{\infty}\\
&\leq
\enVert[3]{n^{-1/2}\sum_{i=1}^n\hat{\bm{\Sigma}}_n^{-1/2}\bm{Y}_{i,n}-n^{-1/2}\sum_{i=1}^n\bm{\Sigma}_{n}^{-1/2}\bm{Y}_{i,n}}_2.
\end{align*}
By Lemma~\ref{eq:pnormerror} (with $p = 2$) it follows that $A_n=O_\P\del[1]{[d^{1/2}\vee \|\bm{\theta}_n\|_2]a_n}
$, where we used that~$\sum_{i = 1}^{d} g_2(\theta_{i,n})=\|\bm{\theta}_n\|_2^2$. Thus, for every~$\delta\in(0,1)$ there exists an~$M=M_\delta>0$ such that for~$\eps_n=M[d^{1/2}\vee \|\bm{\theta}_n\|_2]a_n$ one has~$\P(A_n> \eps_n)\leq \delta$ for all~$n\in\N$.  
Therefore, for any sequence~$H_n$ with~$H_n \in \mc{H}_n$ for every $n$ we have
\begin{align*}
\cbr[4]{{n^{-1/2}\sum_{i=1}^n\hat{\bm{\Sigma}}_n^{-1/2}\bm{Y}_{i,n}\in H_n}}
\subseteq
\cbr[4]{{n^{-1/2}\sum_{i=1}^n\bm{\Sigma}_n^{-1/2}\bm{Y}_{i,n}\in H_n^{\eps_n,\infty}}}
\cup \cbr[1]{A_n> \eps_n}.
\end{align*}	
Since~$H_n\in\mc{H}_n$, one has~$H_n=\prod_{j=1}^d[a_j,b_j]$ for~$-\infty\leq a_j\leq b_j\leq \infty,\ j=1,\hdots,d$, and so~$H_n^{\eps_n,\infty}=\prod_{j=1}^d[a_j-\eps_n,b_j+\eps_n]\in\mc{H}_n$. Thus, denoting the supremum in~\eqref{eq:hyperass} by~$r_n$,
\begin{align*}
\P\del[2]{{n^{-1/2}\sum_{i=1}^n\hat{\bm{\Sigma}}_n^{-1/2}\bm{Y}_{i,n}\in H_n}}
\leq
\P\del[1]{\bm{Z}_d+\bm{\theta}_n \in H_n^{\eps_n,\infty}}+r_n +\delta.	
\end{align*}
Next, by~$H_n^{\eps_n,\infty}-\bm{\theta}_n=\prod_{j=1}^d[a_j-\theta_{j,n}-\eps_n,b_j-\theta_{j,n}+\eps_n]$ it follows that
\begin{align*}
&\P\del[1]{\bm{Z}_d+\bm{\theta}_n \in H_n^{\eps_n,\infty}}\\
=	
&\P\del[2]{\bm{Z}_d \in \prod_{j=1}^d[a_j-\theta_{j,n}-\eps_n,b_j-\theta_{j,n}+\eps_n]}\\
=
&\P\del[2]{(\bm{Z}'_d,-\bm{Z}'_d)' \in \prod_{j=1}^d(-\infty,b_j-\theta_{j,n}+\eps_n]\times \prod_{j=1}^d(-\infty,-a_j+\theta_{j,n}+\eps_n]}.
\end{align*}
Lemma A.1 in~\cite{chernozhukov2017central} (which is attributed to~\cite{nazarov2003maximal}, cf.~also \cite{chernozhukov2017detailed}) now yields that the far right-hand side of the previous display can (eventually) be bounded from above by
\begin{align*}
&\P\del[2]{(\bm{Z}'_d,-\bm{Z}'_d)' \in \prod_{j=1}^d(-\infty,b_j-\theta_{j,n}]\times \prod_{j=1}^d(-\infty,-a_j+\theta_{j,n}]}+C\sqrt{\log(d)}\eps_n\\
=
&\P\del[1]{\bm{Z}_d\in H_n-\bm{\theta}_n}+C\sqrt{\log(d)}\eps_n\\
=
&
\P\del[1]{\bm{Z}_d+\bm{\theta}_n\in H_n}+C\sqrt{\log(d)}\eps_n.
\end{align*}
Therefore,
\begin{align*}
\limsup_{n\to\infty}\sbr[3]{\P\del[3]{{n^{-1/2}\sum_{i=1}^n\hat{\bm{\Sigma}}_n^{-1/2}\bm{Y}_{i,n}\in H_n}}-\P\del[1]{\bm{Z}_d+\bm{\theta}_n \in H_n}}
\leq
\delta.
\end{align*}
On the other hand, it also holds that
\begin{align*}
\cbr[4]{{n^{-1/2}\sum_{i=1}^n\bm{\Sigma}_n^{-1/2}\bm{Y}_{i,n}\in H_n^{-\eps_n,\infty}}}	
\subseteq
\cbr[4]{{n^{-1/2}\sum_{i=1}^n\hat{\bm{\Sigma}}_n^{-1/2}\bm{Y}_{i,n}\in H_n}}	
\cup
\cbr[1]{A_n> \eps_n}.
\end{align*}
Since~$H_n^{-\eps_n,\infty}=\prod_{j=1}^d[a_j+\eps_n,b_j-\eps_n]\in\mc{H}_n$, similar arguments to the above, noting also that $$|\P(\bm{Z}_d + \bm{\theta}_n \in H_n) - \P(\bm{Z}_d + \bm{\theta}_n \in H_n^{-\varepsilon_n, \infty})| \leq \Phi(\varepsilon_n) - \Phi(-\varepsilon_n) \quad \text{ if } H_n^{-\eps_n,\infty} = \emptyset,$$ imply 
\begin{align*}
\liminf_{n\to\infty}\sbr[3]{\P\del[3]{{n^{-1/2}\sum_{i=1}^n\hat{\bm{\Sigma}}_n^{-1/2}\bm{Y}_{i,n}\in H_n}}-
\P\del[1]{\bm{Z}_d+\bm{\theta}_n \in H_n}}
\geq 
-\delta.
\end{align*}
Thus,~\eqref{eq:GAhyper} follows since the sequence~$H_n$ and~$\delta\in(0,1)$ were arbitrary.
\end{proof}

\begin{remark}\label{rem:arbitraryhyper}
Under the assumptions of Lemma~\ref{lem:hypergauss} the supremum in~\eqref{eq:GAhyper} can be extended to also include all open rectangles of the form~$H_n^o=\prod_{j=1}^d(a_j,b_j)$ (where without loss of generality~$-\infty<a_j<b_j<\infty$) by the following squeezing argument: write~$H_n=\prod_{j=1}^d[a_j,b_j]$ and $\bm{S}_n=n^{-1/2}\sum_{i=1}^n\hat{\bm{\Sigma}}_n^{-1/2}\bm{Y}_{i,n}$. Using that~$\bm{Z}_d$ has no mass at~$H_n\setminus H_n^o$, one has (along any subsequence of~$n$)
\begin{align*}
\P\del[1]{\bm{S}_n\in H_n^o}-\P\del[1]{\bm{Z}_d+\bm{\theta}_n\in H_n^o}
\leq 
\P\del[1]{\bm{S}_n\in H_n}-\P\del[1]{\bm{Z}_d+\bm{\theta}_n\in  H_n};
\end{align*}
by Lemma~\ref{lem:hypergauss} the right-hand side tends to zero. Similarly, setting~$c_n=[\log(d)]^{-1}\wedge \min_{j=1\hdots,d}(b_j-a_j)/4$%
\begin{align*}
\P\del[1]{\bm{S}_n\in H_n^o}-\P\del[1]{\bm{Z}_d+\bm{\theta}_n\in H_n^o}
\geq
\P\del[2]{\bm{S}_n\in \prod_{j=1}^d[a_j+c_n,b_j-c_n]}-\P\del[1]{\bm{Z}_d+\bm{\theta}_n\in H_n}.
\end{align*}
Using Lemma A.1 in~\cite{chernozhukov2017central} as in the proof of Lemma~\ref{lem:hypergauss} it follows that there exists a constant~$C\in(0,\infty)$ such that
\begin{align*}
\P\del[1]{\bm{Z}_d+\bm{\theta}_n\in H_n}
\leq
\P\del[2]{\bm{Z}_d+\bm{\theta}_n\in \prod_{j=1}^d[a_j+c_n,b_j-c_n]}+C\sqrt{\log(d)}c_n,
\end{align*}
such that Lemma~\ref{lem:hypergauss} can again be applied yielding the claimed uniform convergence as~$H_n^o$ was an arbitrary open rectangle. As any rectangle can be squeezed between its counterpart with all endpoints excluded and included, respectively, Lemma~\ref{lem:hypergauss} also applies to arbitrary rectangles.
\end{remark}

\section{Proof of Theorem~\ref{thm:supnorm}}
The size control in part a) follows by combining Lemma~\ref{lem:hypergauss} (using ``$n'=n$'' and~$\bm{\theta}_n=\bm{0}_d$) and Lemma A.2 in~\cite{kp2021consistency}, cf.~also Remark~\ref{rem:arbitraryhyper} above (which is invoked tacitly in the remainder of the proof). 

Regarding part b), suppose that $b_n:=\sum_{i = 1}^d \frac{\overline{\Phi}\left(\mathfrak{c}_d - |\theta_{i,d}|\right)}{\Phi\left(\mathfrak{c}_d - |\theta_{i,d}|\right)} \to \infty$. By Theorem 3.5 in~\cite{kp2021consistency} one has~$\P(\|\bm{Z}_d+\bm{\theta}_n\|_\infty\geq  \kappa_{n,\infty})\to 1$. Let~$n'$ be a subsequence of~$n$. We show that it possesses a further subsequence along which~$\P\del[1]{S_{n,\infty}\geq \kappa_{n,\infty}}\to 1$. If~$\frac{1}{{d}^{3/2}/\log(d)}\|\bm{\theta}_{n}\|_2^2$ is bounded along~$n'$, Lemma~\ref{lem:hypergauss} implies~$\P\del[1]{S_{n',\infty}\geq \kappa_{n',\infty}}\to 1$. 

If, on the other hand,~$\frac{1}{{d}^{3/2}/\log(d)}\|\bm{\theta}_{n}\|_2^2$ is unbounded there exists a further subsequence~$n''$ along which~$\frac{1}{{d}^{3/2}/\log(d)}\|\bm{\theta}_{n}\|_2^2\to \infty$.  Furthermore, note that~$\P\del[1]{S_{n,\infty}\geq \kappa_{n,\infty}}$ equals
\begin{align}\label{eq:supteststart}
\P\del[4]{\enVert[3]{n^{-1/2}\sum_{i=1}^n\hat{\bm{\Sigma}}_n^{-1/2}(\bm{Y}_{i,n}-\bm{\mu}_n)+n^{1/2}\hat{\bm{\Sigma}}_n^{-1/2}\bm{\mu}_n}_\infty \geq \kappa_{n,\infty}}.
\end{align}
Because~$\bm{Y}_{i,n}-\bm{\mu}_n$ is centered and has the same covariance matrix as~$\bm{Y}_{i,n}$, we can apply Lemma~\ref{lem:hypergauss} (with~$\bm{\theta}_n=\bm{0}_d$ there) in conjunction with part a) to conclude that (along the entire sequence~$n$)
\begin{align}\label{eq:suptestsmallterm}
\enVert[3]{n^{-1/2}\sum_{i=1}^n\hat{\bm{\Sigma}}_n^{-1/2}(\bm{Y}_{i,n}-\bm{\mu}_n)}_\infty
=
O_\P\del[1]{\sqrt{\log(d)}}. 
\end{align}
Next, 
\begin{align*}
\enVert[1]{n^{1/2}\hat{\bm{\Sigma}}_n^{-1/2}\bm{\mu}_n}_\infty
\geq
\frac{\enVert[1]{n^{1/2}\hat{\bm{\Sigma}}_n^{-1/2}\bm{\mu}_n}_2}{\sqrt{d}}
\geq
\frac{\enVert[1]{n^{1/2}\bm{\Sigma}_n^{-1/2}\bm{\mu}_n}_2-\enVert[1]{n^{1/2}[\hat{\bm{\Sigma}}_n^{-1/2}-\bm{\Sigma}_n^{-1/2}]\bm{\mu}_n}_2}{\sqrt{d}}
\end{align*}
and the numerator on the far right-hand side equals~$\enVert[1]{n^{1/2}\bm{\Sigma}_n^{-1/2}\bm{\mu}_n}_2=\|\bm{\theta}_n\|_2$ (which is eventually positive along~$n''$) times
\begin{align*}
1-\frac{\enVert[0]{n^{1/2}[\hat{\bm{\Sigma}}_n^{-1/2}-\bm{\Sigma}_n^{-1/2}]\bm{\mu}_n}_2}{\enVert[1]{n^{1/2}\bm{\Sigma}_n^{-1/2}\bm{\mu}_n}_2}
\geq
1-\frac{\enVert[0]{\hat{\bm{\Sigma}}_n^{-1/2}-\bm{\Sigma}_n^{-1/2}}_2\|\bm{\mu}_n\|_2}{c \|\bm{\mu}_n\|_2}
=1-O_\P(a_n),
\end{align*}
for some real $c > 0$, due to the uniform bound on the eigenvalues of $\bm{\Sigma}_n$,
the last equality following from Lemma~\ref{lem:covestim}. The previous two displays thus yield that along~$n''$ there exists an~$M\in(0,\infty)$ such that
\begin{align}\label{eq:suptestlargeterm}
\enVert[1]{n^{1/2}\hat{\bm{\Sigma}}_n^{-1/2}\bm{\mu}_n}_\infty
\geq
\frac{\|\bm{\theta}_n\|_2}{\sqrt{d}}	\del[1]{1-O_\P(a_n)}
\geq
M\frac{d^{1/4}}{\sqrt{\log(d)}}	\del[1]{1-O_\P(a_n)}
\end{align}
Finally,~\eqref{eq:supteststart} and hence~$\P\del[1]{S_{n,\infty}(\bm{\beta}_n^*)\geq \kappa_{n,\infty}}$ is no smaller than
\begin{align*}
\P\del[4]{\enVert[1]{n^{1/2}\hat{\bm{\Sigma}}_n^{-1/2}\bm{\mu}_n}_\infty-\enVert[3]{n^{-1/2}\sum_{i=1}^n\hat{\bm{\Sigma}}_n^{-1/2}(\bm{Y}_{i,n}-\bm{\mu}_n)}_\infty\geq  \kappa_{n,\infty}}\to 1	
\end{align*}
along~$n''$; the convergence following from~\eqref{eq:suptestsmallterm}, \eqref{eq:suptestlargeterm},~$\kappa_{n,\infty}=O(\sqrt{\log(d)})$, and~$a_n\to 0$.

Suppose now that~$\P\del[1]{S_{n,\infty}(\bm{\beta}_n^*)\geq \kappa_{n,\infty}}\to 1$ and choose a subsequence~$n'$. We show that~$n'$ possesses a further subsequence along which~$b_n\to\infty$. If~$\frac{\log(d)}{d^{3/2}}\|\bm{\theta}_n\|_2^2$ is bounded along~$n'$, then Lemma~\ref{lem:hypergauss} and an application of Theorem 3.5 in \cite{kp2021consistency} along~$n'$ imply that~$b_{n'}\to\infty$.
If, on the other hand,~$\frac{\log(d)}{d^{3/2}}\|\bm{\theta}_n\|_2^2$ is unbounded along~$n'$ then there is a further subsequence~$n''$ along which~$\frac{\log(d)}{d^{3/2}}\|\bm{\theta}_n\|_2^2\to\infty$. Since~$\|\bm{\theta}_n\|_\infty\geq d^{-1/2}\|\bm{\theta}_n\|_2$ and~$\mathfrak{c}_d=O(\sqrt{\log(d)})$, this implies along~$n''$ that~$$b_n  \geq \frac{\overline{\Phi}\left(\mathfrak{c}_d - \|\bm{\theta}_n\|_\infty\right)}{\Phi\left(\mathfrak{c}_d - \|\bm{\theta}_n\|_\infty\right)} \to\infty.$$

\section{Proof of Theorem~\ref{thm:domtest}}
For~$p\in[2,\infty]$ and~$r\in(0,\infty)$ define $\mathbb{B}_p(r)=\cbr[0]{\bm{x}\in\R^d:\|\bm{x}\|_p<r}$. Concerning the asymptotic size control in Part 1, note that by convexity of~$\bigcap_{p\in\mathfrak{P}_n}\mathbb{B}_p(c_n\kappa_{n,p})$,  Lemma~\ref{lem:convgauss} yields 
\begin{align*}
\E\psi_n
=
1-\P\del[3]{\hat{\bm{\Sigma}}_n^{-1/2}\bm{H}_n\in\bigcap_{p\in\mathfrak{P}_n}\mathbb{B}_p(c_n\kappa_{n,p})}
= 1-\P\del[2]{\bm{Z}_d\in \bigcap_{p\in\mathfrak{P}_n}\mathbb{B}_p(c_n\kappa_{n,p})}+r_n;
\end{align*} 
for a sequence~$r_n\to 0$. Finally, by~\eqref{eq:kappachoice},
\begin{align*}
1-\P\del[2]{\bm{Z}_d\in \bigcap_{p\in\mathfrak{P}_n}\mathbb{B}_p(c_n\kappa_{n,p})}
=
\P\del[3]{\max_{p \in \mathfrak{P}_n} \kappa_{n,p}^{-1}
\| \bm{Z}_{d} \|_{p} \geq c_n}
=
\alpha.
\end{align*}
Concerning Part 2,~first let~$p\in[2,\infty)$ and consider a test~$\mathds{1}\cbr[1]{S_{n,p}\geq \bar{\kappa}_{n,p}}$ as in the statement of the theorem (where we drop the dependence of $S_{n,p}$ on~$\bm{\beta}^*_n$ and of $\overline{\kappa}_{n,p}$ on $\overline{\alpha}_p$ for notational convenience). Observe that since~$m_n\to\infty$ and~$p_n$ is strictly increasing, there exists a~$J\in\mathbb{N}$ such that eventually~$p\leq p_J\leq p_{m_n}<\infty$ and
\begin{align*}
\psi_n
\geq
\mathds{1}\cbr[1]{S_{n,p_J}\geq c_n\kappa_{n,p_J}}
\geq
\mathds{1}\cbr[1]{S_{n,p_J}\geq \kappa_{n,p_J}}.
\end{align*}
By Lemma~\ref{lem:convgauss}, the test~$ \mathds{1}\cbr[1]{S_{n,p_J}\geq \kappa_{n,p_J}}$ has asymptotic null rejection probability equal to $\lim_{n \to \infty} \alpha_{n,p_J}\in(0,\alpha_2)$. Hence, by Theorem~\ref{thm:pnorm-char_GMM} (see also~\eqref{eq:dom}), if~$\P\del[1]{S_{n,p}\geq \bar{\kappa}_{n,p}}\to 1$ it also follows that~$\P\del[1]{S_{n,p_J}\geq \kappa_{n,p_J}}\to 1$, so that~$\E\psi_n\to 1$ follows from the previous display.

Next, let~$p=\infty$ and consider a test~$\mathds{1}\cbr[1]{S_{n,\infty}\geq \bar{\kappa}_{n,\infty}}$ as in the statement of the theorem (again dropping some dependencies for notational convenience). Clearly,
\begin{align*}
\psi_n
\geq
\mathds{1}\cbr[1]{S_{n,\infty}\geq c_n\kappa_{n,\infty}}
\geq
\mathds{1}\cbr[1]{S_{n,\infty}\geq \kappa_{n,\infty}}.
\end{align*}
By Lemma~\ref{lem:hypergauss}, the test~$\mathds{1}\cbr[1]{S_{n,\infty}\geq \kappa_{n,\infty}}$ has asymptotic null rejection frequency~$\alpha_\infty\in(0,1)$, cf.~also Remark~\ref{rem:arbitraryhyper}. Hence, by Theorem~\ref{thm:supnorm}, if $\P\del[1]{S_{n,\infty}\geq \bar{\kappa}_{n,\infty}}\to 1$ then also $\P\del[1]{S_{n,\infty}\geq \kappa_{n,\infty}}\to 1$ and~$\E\psi_n\to 1$ follows from the previous display.

\section{Proof of Theorem \ref{thm:boundedloss}}

Also in this proof, we do not signify the dependence of several quantities on~$\bm{\beta}^*$. Let~$n$ be large enough to ensure that~$p\in\mathfrak{P}_n$. The definition of~$\psi_n$ in~\eqref{eq:psidef} along with~$c_n\in(0,1]$ implies that 		
\begin{equation}\label{eqn:pjupbd}
\begin{aligned}
\P\del[1]{S_{n,p} \geq\bar{\kappa}_{n,p}(\alpha)}
- 
\E\psi_n
\leq  
\P\del[1]{S_{n,p} \geq\bar{\kappa}_{n,p}(\alpha)}			-
\P\del[1]{S_{n,p} \geq\kappa_{n,p}}.
\end{aligned}
\end{equation}
By the definition of~$\kappa_{n,p}$ in the construction of~$\psi_n$ and Lemma~\ref{lem:convgauss} it holds under the null that~$\lim_{n\to\infty}\P\del[1]{S_{n,p} \geq\kappa_{n,p}}=\alpha_p\in(0,1)$.	For clarity, we thus write~$\kappa_{n,p}(\alpha_p)$ in what follows. Next, denote by~$\mathfrak{s}$ the limit superior of the left-hand side in~\eqref{eqn:pjupbd} and let~$n'$ be a subsequence satisfying
\begin{equation*}\label{eqn:limsuprepl}
\lim_{n'\to\infty}  \left[\P \del[1]{S_{n',p}\geq \bar{\kappa}_{n',p}(\alpha)}-\E\psi_{n'}\right] =\mathfrak{s}.
\end{equation*}
We shall also use that the right hand side of~\eqref{eqn:pjupbd} equals
\begin{align}
&\sbr[1]{\P\del[1]{S_{n,p} \geq\bar{\kappa}_{n,p}(\alpha)}-\P\del[1]{\|\bm{Z}_d+\bm{\theta}_n\|_p \geq\bar{\kappa}_{n,p}(\alpha)}}\notag\\
-&
\sbr[1]{\P\del[1]{S_{n,p} \geq\kappa_{n,p}(\alpha_p)}-\P\del[1]{\|\bm{Z}_d+\bm{\theta}_n\|_p \geq\kappa_{n,p}(\alpha_p)}}\notag\\
+&			
\sbr[1]{\P\del[1]{\|\bm{Z}_d+\bm{\theta}_n\|_p \geq\bar{\kappa}_{n,p}(\alpha)}-\P\del[1]{\|\bm{Z}_d+\bm{\theta}_n\|_p \geq\kappa_{n,p}(\alpha_p)}}\label{eq:bounddecomp}.
\end{align} 

Consider first the case of~$p\in \mathfrak{P}_n\cap [2,\infty)$: If necessary, pass to a further subsequence~$n''$ along which~$d^{-1/2}\sum_{i=1}^dg_p(\theta_{i,n})$ converges to some~$b\in[0,\infty]$. If~$b=\infty$, applying Theorems~\ref{thm:pnorm-char_GMM} (Part 3) and~\ref{thm:prim} along this subsequence\footnote{Inspection of the proof of these theorems shows that they remain valid along subseqeunces satisfying the conditions of the theorem.} yields
\begin{align*}
\mathfrak{s}
\leq
\lim_{n'\to\infty}\sbr[2]{\P\del[1]{S_{n',p} \geq\bar{\kappa}_{n',p}(\alpha)}			-
\P\del[1]{S_{n',p} \geq\kappa_{n',p}(\alpha_p)}}	
=
1-1
=
0.
\end{align*}
If, on the other hand,~$b\in[0,\infty)$, then~$d^{-1}\sum_{i=1}^dg_p(\theta_{i,n})\to 0$. Hence, also~$d^{-1}\sum_{i=1}^dg_2(\theta_{i,n})\to 0$ and so~\eqref{eq:bounddecomp} and Lemma~\ref{lem:convgauss} together with the argument following (C.61) in the proof of Theorem C.10 in~\cite{kp2021consistency} yield the desired bound on~$\mathfrak{s}$.

Next, consider the case where~$p=\infty$: If necessary, pass to a further subsequence~$n''$ along which~$\sum_{i = 1}^d \frac{\overline{\Phi}\left(\mathfrak{c}_d - |\theta_{i,d}|\right)}{\Phi\left(\mathfrak{c}_d - |\theta_{i,d}|\right)}$ converges to some~$c\in[0,\infty]$. If~$c=\infty$, applying Part 2 of Theorem~\ref{thm:supnorm} along this subsequence yields
\begin{align*}
\mathfrak{s}
\leq
\lim_{n'\to\infty}\sbr[2]{\P\del[1]{S_{n',\infty} \geq\bar{\kappa}_{n',\infty}(\alpha)}			-
\P\del[1]{S_{n',\infty} \geq\kappa_{n',\infty}(\alpha_\infty)}}	
=
1-1
=
0.
\end{align*}
If, on the other hand~$c\in[0,\infty)$, then~$\frac{1}{{d'}^{3/2}/\log(d')} \sum_{i = 1}^{d'} g_2(\theta_{i,n'})=\frac{1}{{d'}^{3/2}/\log(d')}\|\bm{\theta}_{n'}\|_2^2$ is bounded (cf.~the argument at the end of the proof of Theorem~\ref{thm:supnorm}). Therefore, by~\eqref{eq:bounddecomp} and Lemma~\ref{lem:hypergauss} it suffices to study the limiting behaviour of
\begin{align*}
	&\P\del[1]{\|\bm{Z}_d+\bm{\theta}_n\|_\infty \leq\kappa_{n,\infty}(\alpha_\infty)}-\P\del[1]{\|\bm{Z}_d+\bm{\theta}_n\|_\infty \leq\bar{\kappa}_{n,\infty}(\alpha)}\\
	=&
	\P\del[1]{\|\bm{Z}_d+\bm{\theta}_n\|_\infty \leq\bar{\kappa}_{n,\infty}(\alpha)+[\kappa_{n,\infty}(\alpha_\infty)-\bar{\kappa}_{n,\infty}(\alpha)]}-\P\del[1]{\|\bm{Z}_d+\bm{\theta}_n\|_\infty \leq\bar{\kappa}_{n,\infty}(\alpha)},
\end{align*}
along~$n'$. By~\eqref{eq:sizesup}
\begin{align*}
\kappa_{n,\infty}(\alpha_\infty)-\bar{\kappa}_{n,\infty}(\alpha)
=
\frac{\log\del[0]{-\log(1-\alpha)/2}}{\sqrt{2\log(d)}}-\frac{\log\del[0]{-\log(1-\alpha_\infty)/2}}{\sqrt{2\log(d)}}+o\del[2]{\frac{1}{\sqrt{\log(d)}}}.
\end{align*}
Therefore, by Theorem 1 in~\cite{chernozhukov2017detailed} and arguments similar to those in the proof of Lemma~\ref{lem:hypergauss} the limit superior of the right-hand side of the penultimate display along~$n'$ is no greater than~$\log\del[0]{-\log(1-\alpha)/2}-\log\del[0]{-\log(1-\alpha_\infty)/2}$, which is the claimed upper bound.

\section{Sample split}\label{app:samplesplit}
To state and prove the asymptotic size control of the proposed sample split, let~$I \subseteq \cbr[0]{1,\hdots,D}$ be non-empty. For any~$D\times 1$ vector~$\bm{a}$ and any~$D\times D$ matrix $\bm{A}$, let~$\bm{a}_I$ and~$\bm{A}_{I}$, respectively, denote the sub-vector (sub-matrix) with rows (and columns) indexed by~$I$. In this section~$h_{n}=(h_{1,n},\hdots,h_{D,n})'$ denotes the total number,~$D$, of moment restriction functions. In addition,
\begin{align*}
\bm{H}(\bm{\beta}_n^*,N)&:=\frac{1}{\sqrt{|N|}}\sum_{i\in N}h_{n}(\bm{X}_{i,n},\bm{\beta}_n^*)\qquad \text{where }N\subseteq\cbr[0]{1,\hdots,n},\ N\neq \emptyset\\ 
\hat{\bm{\Sigma}}(\bm{\beta}_n^*,N,I)&\text{\hspace{0.6cm}is a symmetric, positive semi-definite estimator of the covariance matrix}\\
&\hspace{0.6cm}\text{of~$\bm{H}_I(\bm{\beta}_0,N)$ using observations indexed by~$N$ only.}\\
S_{p}(\bm{\beta}_n^*,N,I)&:=\enVert[2]{\sbr[1]{\hat{\bm{\Sigma}}(\bm{\beta}_n^*,N,I)}^{-1/2}\bm{H}_{I}(\bm{\beta}_n^*,N)}_p,\qquad p\in[2,\infty]\\
\psi(\bm{\beta}_n^*,N,I)&:=\mathds{1}\cbr[2]{\max_{p\in\mathfrak{P}_{|N|}} \kappa_{|N|,p}^{-1}
S_{p}(\bm{\beta}_n^*,N,I)\geq c_{|N|}},
\end{align*}
with all other quantities as in the construction of~$\psi_n(\bm{\beta}_n^*)$ in Section~\ref{sec:domtestconstruct}, but with the~$|N|$ replacing~$n$. Thus, the test (statistics) based on the observations in~$N_2$ and the selected moments~$S$ (depending on observations in~$N_1$ only) are $S_{p}(\bm{\beta}_n^*, N_2, S)$,~$p\in[2,\infty]$, and~$\psi(\bm{\beta}_n^*,N_2,S)$. The allowed growth rate of~$d=d(n_2)$ is now dependent on size~$n_2$ of the second fold. Finally, let~$J_{n_2} \subseteq J_{n_2}^* := \cbr[1]{I\subseteq \cbr[0]{1,\hdots,D}:|I|=d}$ be non-empty.

For concreteness, we work under the appropriate ``sample split version'' of the assumptions of Theorem~\ref{thm:prim}. These essentially impose the assumptions of that theorem to be valid under the null only, but uniformly over all possible sets of selected moments (of size~$d$) replacing the sample size~$n$ by the sample size~$n_2$ in the second step (recall that~$d$ is now the number \emph{selected} moment restrictions out of a total of~$D$ moment restrictions). Under independent sampling, appropriate versions of all other asymptotic size control results that we have stated can also be adjusted to incorporate a sample split. Recall that~$d=d(n_2)$.
\begin{theorem}\label{thm:dimred}
Let~$\bm{\beta}^* \in\mathbf{B}^{(0)}$,~$n_2\to\infty$, $d\to\infty$ and~$\bm{\Sigma}_{n,I}(\bm{\beta}_n^*)$ be invertible for all~$I\in J_{n_2}$ with~$\max_{I\in J_{n_2}}\|\bm{\Sigma}_{n,I}(\bm{\beta}_n^*)\|_2$ as well as $\max_{I\in J_{n_2}}\|\bm{\Sigma}^{-1}_{n,I}(\bm{\beta}_n^*)\|_2$ (uniformly) bounded from above. Furthermore,~$\bm{X}_{1,n},\hdots,\bm{X}_{n,n}$ is i.i.d.~for each~$n\in\N$ and, as~$n_2\to\infty$, 
\begin{enumerate}
\item $\max_{I\in J_{n_2}}\|\hat{\bm{\Sigma}}(\bm{\beta}_n^*,N_2,I)-\bm{\Sigma}_{n_,I}(\bm{\beta}_n^*)\|_2=O_\P(a_{n_2})$ with~$d^{3/4}a_{n_2}\to 0$ and
\item $\sup_{C\in\mc{C}_{n_2}}\max_{I\in J_{n_2}}\envert[1]{\P\del[1]{[\bm{\Sigma}(\bm{\beta}_n^*,N_2,I)]^{-1/2}\bm{H}_I(\bm{\beta}_n^*,N_2)\in  C}-\P\del[1]{\bm{Z}_d \in  C}}\to 0$.
\end{enumerate}
In addition, let~$\alpha\in(0,1)$ and~$S = S\del[1]{n_1,n_2,\cbr[0]{\bm{X}_{i,n}}_{i\in N_1}} \in J_{n_2}$ be measurable. Then:
\begin{enumerate}
\item For~$p\in[2,\infty]$ a sequence of real numbers~$(\kappa_{n_2,p})_{n_2\in\N}$ satisfies
\begin{align*}
\P\del[1]{S_{p}(\bm{\beta}_n^*,N_2,S)\geq \kappa_{n_2,p}}\to \alpha, 
\end{align*}
if and only if~$\kappa_{n_2,p}$ is as in~\eqref{eq:sizepnorm} if~$p\in[2,\infty)$ or as in~\eqref{eq:sizesup} if~$p=\infty$.	
\item $\E\psi(\bm{\beta}_n^*,N_2,S)\to \alpha$. 
\end{enumerate}
\end{theorem}
Theorem~\ref{thm:dimred} shows that \emph{irrespectively} of how the set of  moments~$S$ to be tested in the second step is selected --- as long as this selection depends (measurably) on observations in~$N_1$ only --- all tests studied have the desired asymptotic size if~$d$ satisfies the growth conditions of our previous theorems relatively to the second step sample size~$n_2$. Thus, one can use \emph{any} state-of-the-art procedure to select the moments to be tested (some examples are given in Section~\ref{ex:momsel} below). We also reiterate that no conditions need to be imposed on the total number of available moment conditions,~$D$.

\begin{proof}[Proof of Theorem~\ref{thm:dimred}]
Fix~$\alpha\in(0,1)$ and let~$\varphi_{n_2}:\mc{X}_n^{n_2}\times J_{n_2}\to [0,1]$ (measurable) be a test. Then, by~$N_1\cap N_2=\emptyset$ and independent sampling,
\begin{align*}
\E\varphi\del[1]{\cbr[0]{\bm{X}_{i,n}}_{i\in N_2},S}
&=
\sum_{I\in J_{n_2}}\E\sbr[1]{\varphi_{n_2}\del[1]{\cbr[0]{\bm{X}_{i,n}}_{i\in N_2},I}\mathds{1}(S=I)}\\
&=
\sum_{I\in J_{n_2}}\E\varphi_{n_2}\del[1]{\cbr[0]{\bm{X}_{i,n}}_{i\in N_2},I}\P(S=I).
\end{align*}	
Thus, to show that~$\lim_{n\to\infty}\E\varphi_{n_2}\del[1]{\cbr[0]{\bm{X}_{i,n}}_{i\in N_2},S}=\alpha$,
it suffices to establish (as~$n_2 \to \infty$)
\begin{align}\label{eq:auxunif}
\max_{I\in J_{n_2}}\envert[1]{\E\varphi_{n_2}\del[1]{\cbr[0]{\bm{X}_{i,n}}_{i\in N_2},I}-\alpha}\to 0	.
\end{align}
We now argue that~\eqref{eq:auxunif} holds for~$\varphi_{n_2}$ a $p$-norm based tests with~$p\in[2,\infty)$. Thus, let 
\begin{align*}
\varphi_{n_2}\del[1]{\cbr[0]{\bm{X}_{i,n}}_{i\in N_2},I}=
\mathds{1}\del[1]{S_{p}(\bm{\beta}_n^*,N_2,I)\geq \kappa_{n_2,p}}.
\end{align*}
Then~\eqref{eq:auxunif} follows from Part 1~of Theorem~\ref{thm:pnorm-char_GMM} since the sufficient conditions of Theorem~\ref{thm:prim} are satisfied on the second subsample~$N_2$ uniformly in~$J_{n_2}$ with ``$\cbr[0]{1,\hdots,n}=N_2$'' and sample size ``$n=n_2=|N_2|$''. Also,~``$d=d(n_2)$''.
The cases of~$p=\infty$ and~$\psi(\bm{\beta}_n^*,N_2,S)$ are handled similarly via Theorems~\ref{thm:supnorm} and~\ref{thm:domtest}, respectively.
\end{proof}

\subsection{Choosing the set of moment restrictions in order to maximize power}\label{ex:momsel}
\begin{itemize}
\item \emph{Maximizing the test statistic in the first step:}
If one intends to use~$S_{p}(\bm{\beta}_n^*,N_2,S)$ in the second step, then one may choose
\begin{align*}
S\in \argmax_{I\in J_{n_2}}S_{p}(\bm{\beta}_n^*,N_1,I),
\end{align*} 
with~$d$ satisfying the growth conditions of Theorem~\ref{thm:dimred}. Since choosing~$S$ in this fashion may be computationally burdensome for very large~$D$, one can alternatively add moments one-by-one in a step-up approach (one could also use a step-down procedure) until arriving at~$d$ moments, in each step adding the moment that maximizes the value of the test statistic given the previously selected moments. A similar procedure can be used for~$\psi(\bm{\beta}_n^*,N_2,S)$. In this case, one chooses 
\begin{equation*}
S \in \argmax_{I\in J_{n_2}} \left[ \max_{p\in\mathfrak{P}_{n_2}} \kappa_{n_1,p}^{-1}
S_{p}(\bm{\beta}_n^*,N_1,I)\right].
\end{equation*}
\item \emph{The~$d$ largest scaled moments:} Let~$S$ be the indices of the~$d$ largest elements of~$$\cbr[3]{\envert[2]{\sbr[1]{\hat{\bm{\Sigma}}(\bm{\beta}_n^*,N_1,\cbr[0]{j})}^{-1/2}\bm{H}_{\cbr[0]{j}}(\bm{\beta}_n^*,N_1)},\ j=1,\hdots, D},\qquad N_1\neq \emptyset.$$
This amounts to testing those~$S$ individual moments that most strongly indicate a violation of~\eqref{eq:GMMlarge} (taking into account the scale of these, but ignoring their correlation structure). It is computationally cheaper than the method in the previous point.
\item One can use \emph{any} state-of-the-art procedure to select the moments to be tested tailored to the specific problem at hand. For example, in the context of hypothesis testing in the presence of many instruments, cf.~Sections~\ref{sec:IVexample} and~\ref{ex:IVdom}, one can use the Lasso in the first step with a choice of penalty parameter ensuring that the number of selected instruments~$d$ satisfies the growth conditions imposed in the previous sections.
\end{itemize}

\section{Further numerical evidence}\label{sec:furtherplots}
Figures \ref{fig:IV} and~\ref{fig:Gauss} below contain the raw power for the simulations in Sections~\ref{sec:simIV} and~\ref{sec:simGauss}, respectively.

\begin{figure}

\begin{center}
	\footnotesize $n=5{,}000$ and~$d=100$
\end{center}

\vspace{-0.5cm}
\includegraphics[width=5.2cm]{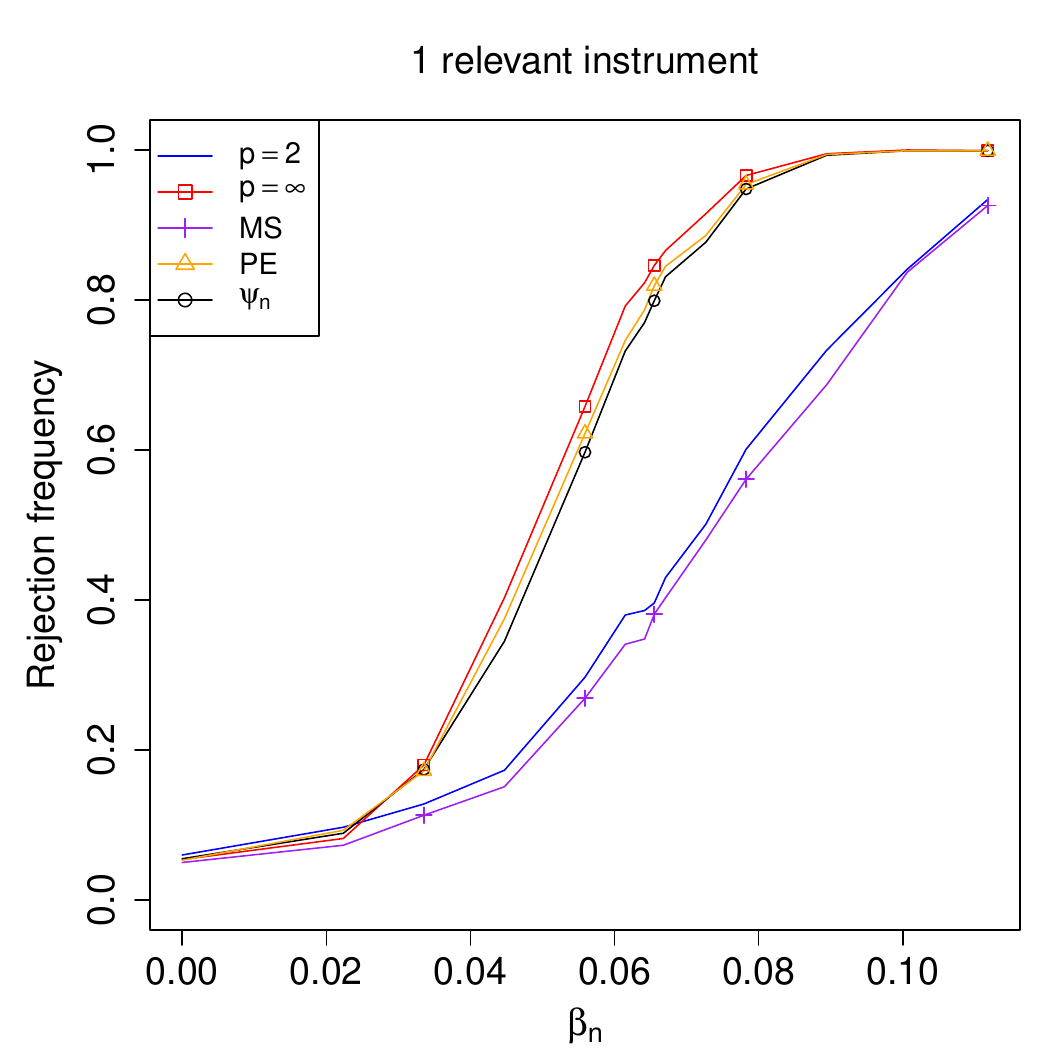}
\hspace{-0.5cm}
\includegraphics[width=5.2cm]{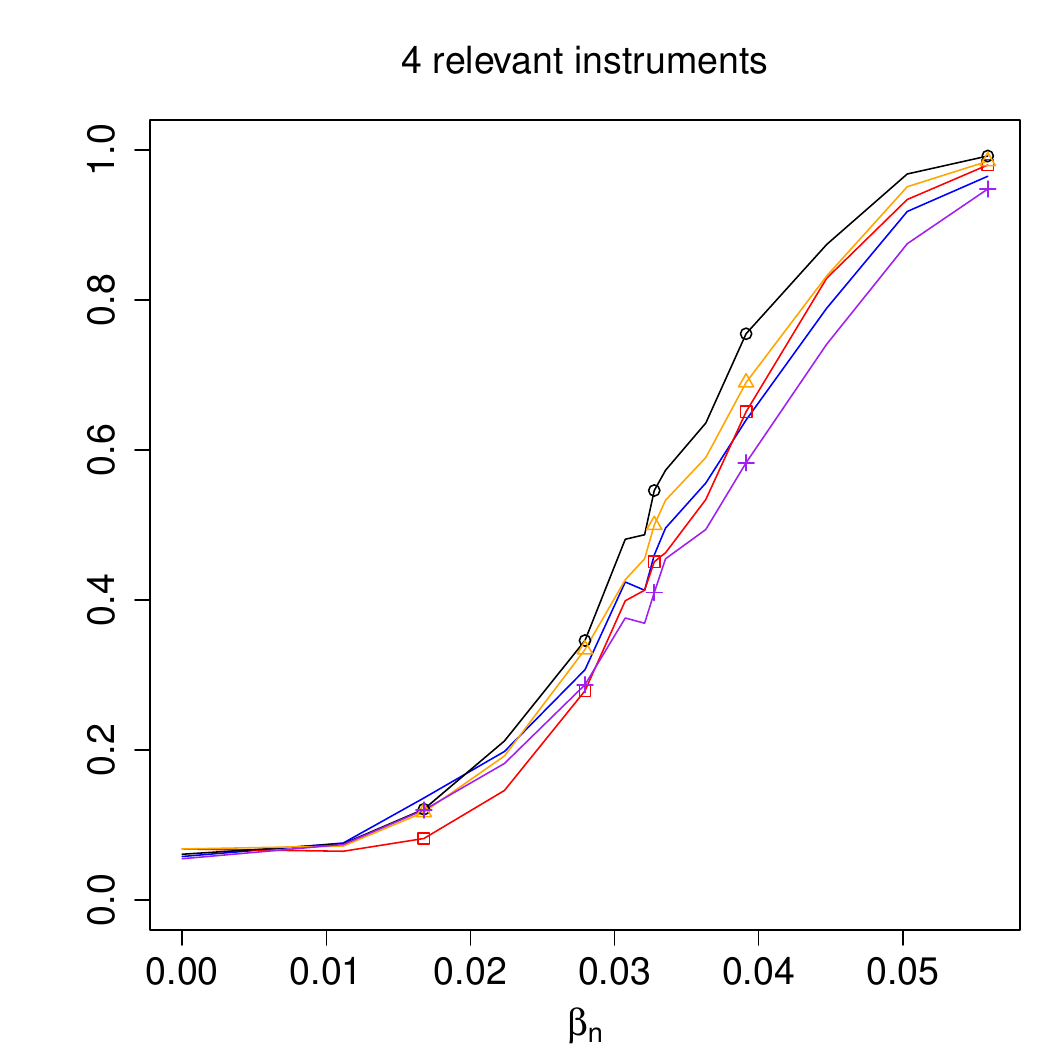}
\hspace{-0.5cm}
\includegraphics[width=5.2cm]{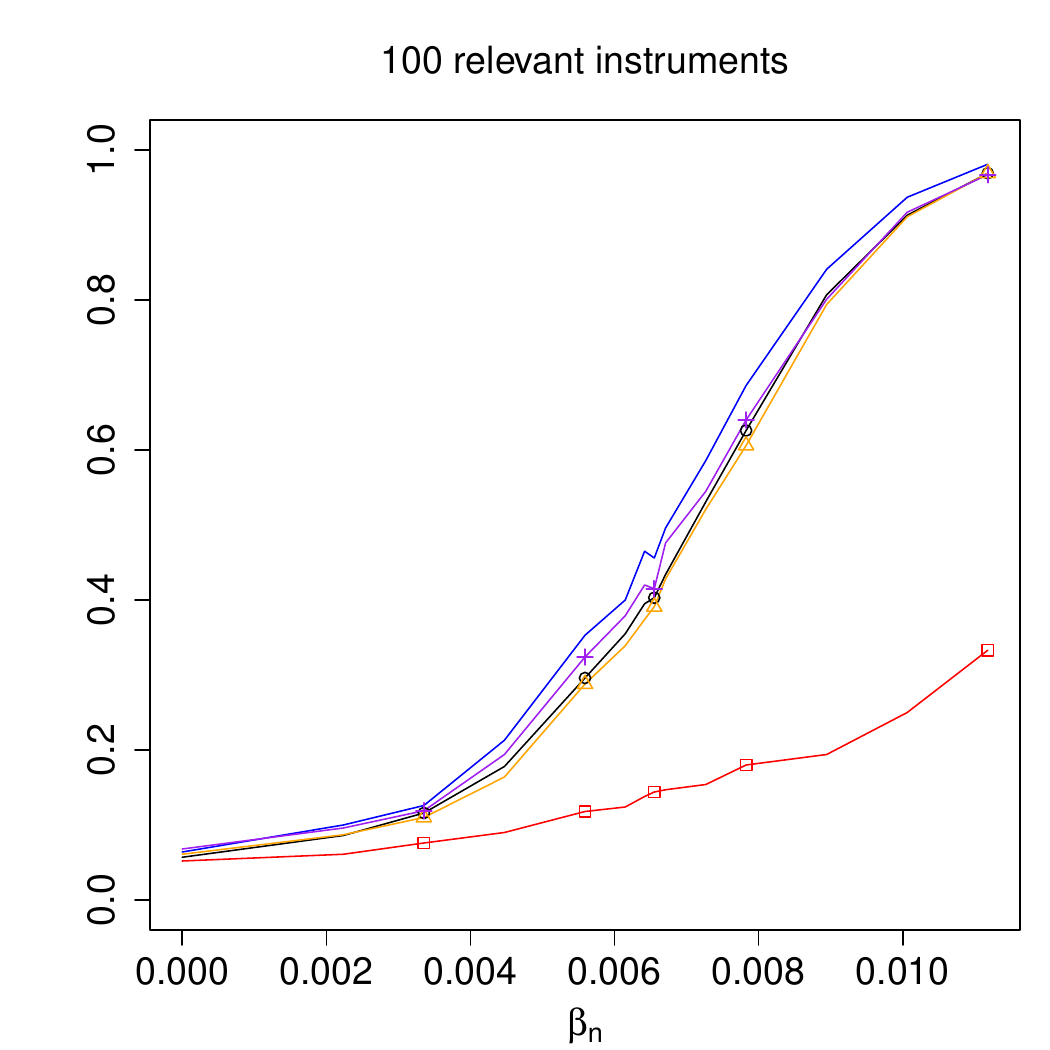}
\vspace{-0.2cm}

\begin{center}
	\footnotesize $n=25{,}000$ and~$d=500$
\end{center}

\vspace{-0.5cm}
\includegraphics[width=5.2cm]{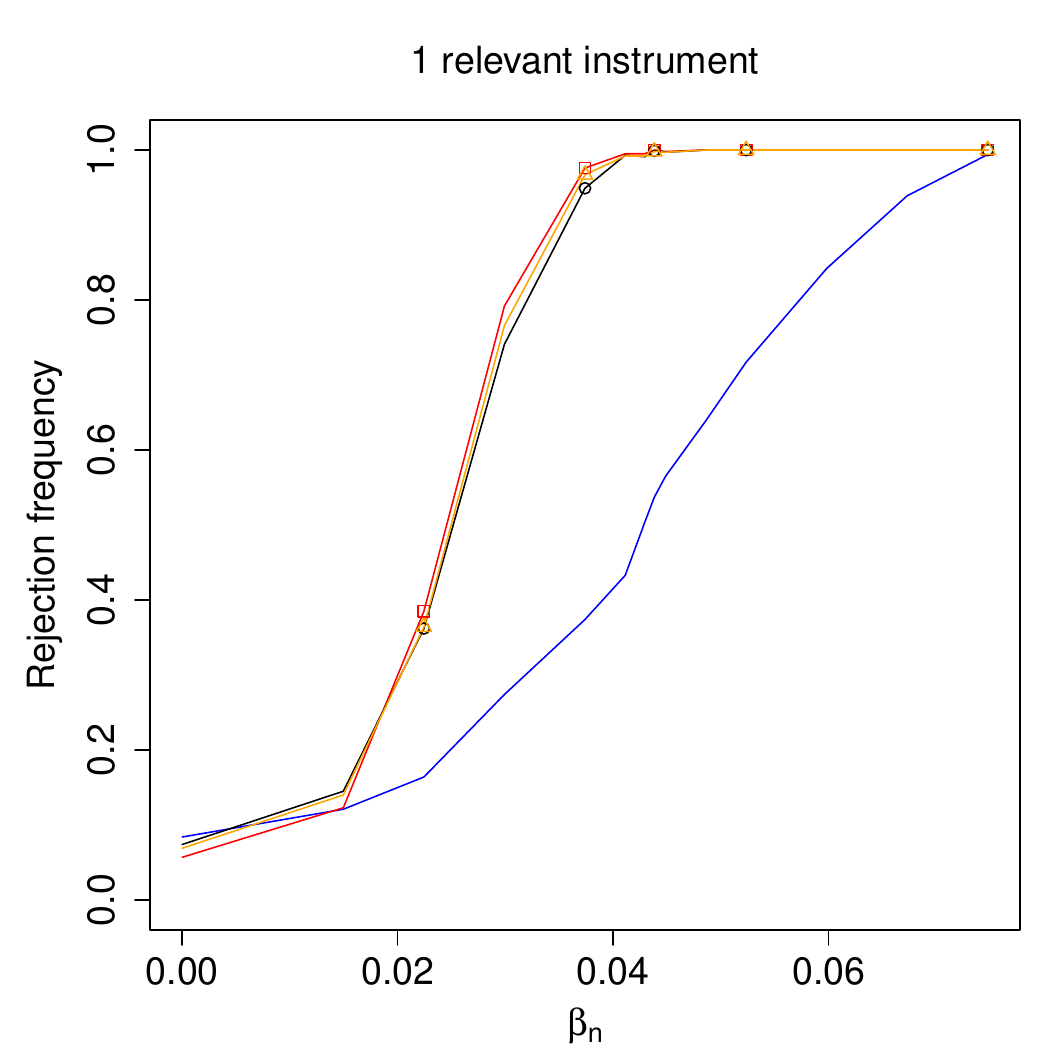}
\hspace{-0.5cm}
\includegraphics[width=5.2cm]{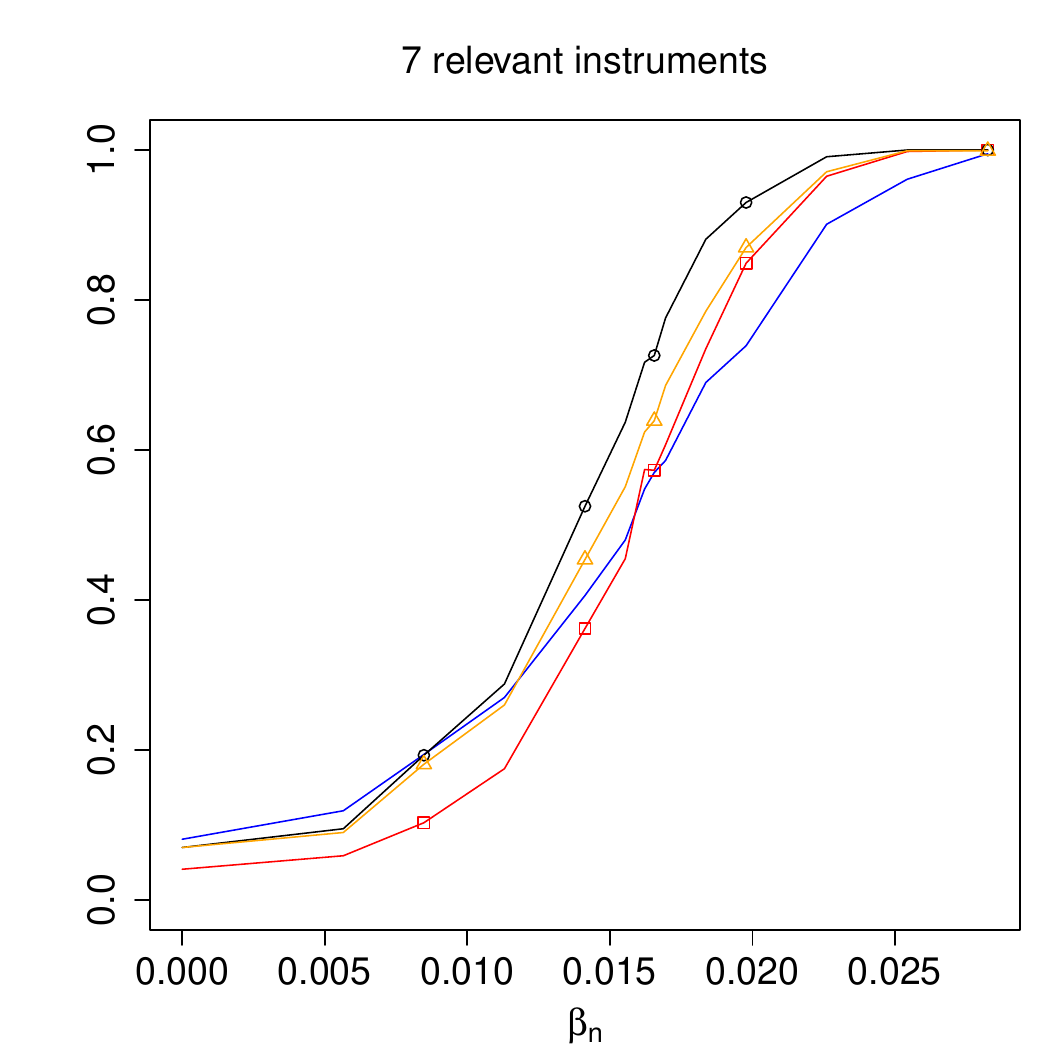}
\hspace{-0.5cm}
\includegraphics[width=5.2cm]{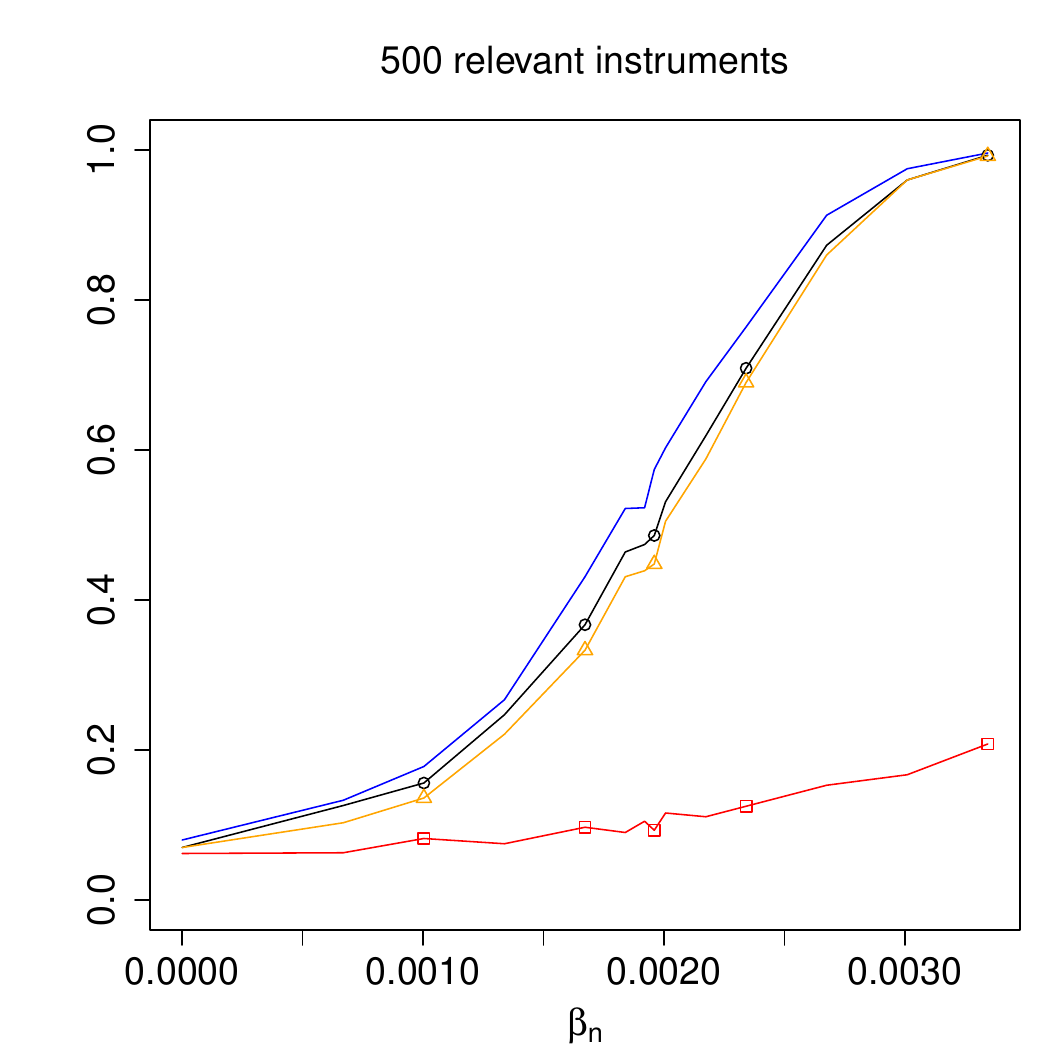}
\vspace{-0.2cm}

\begin{center}
	\footnotesize  $n=100{,}000$ and~$d=1{,}000$
\end{center}

\vspace{-0.5cm}
\includegraphics[width=5.2cm]{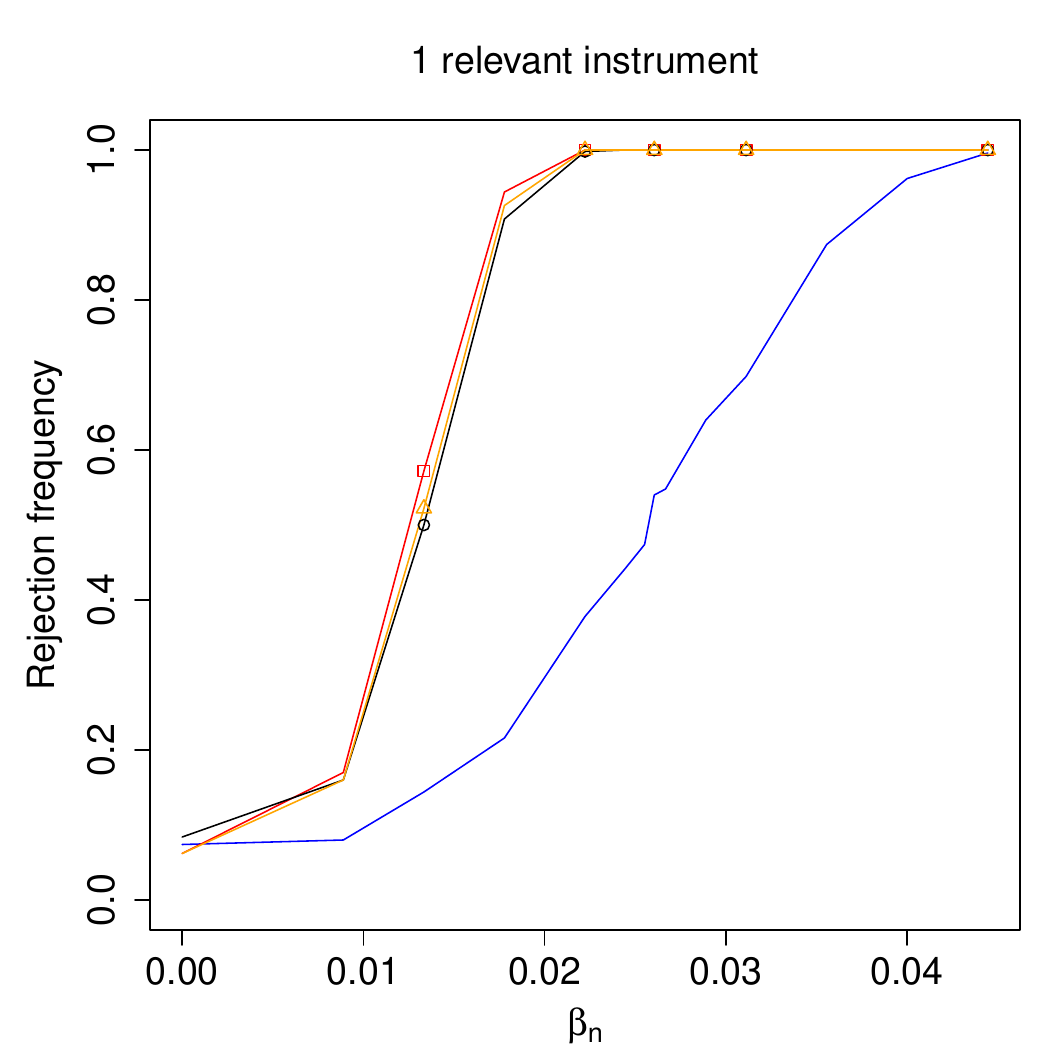}
\hspace{-0.5cm}
\includegraphics[width=5.2cm]{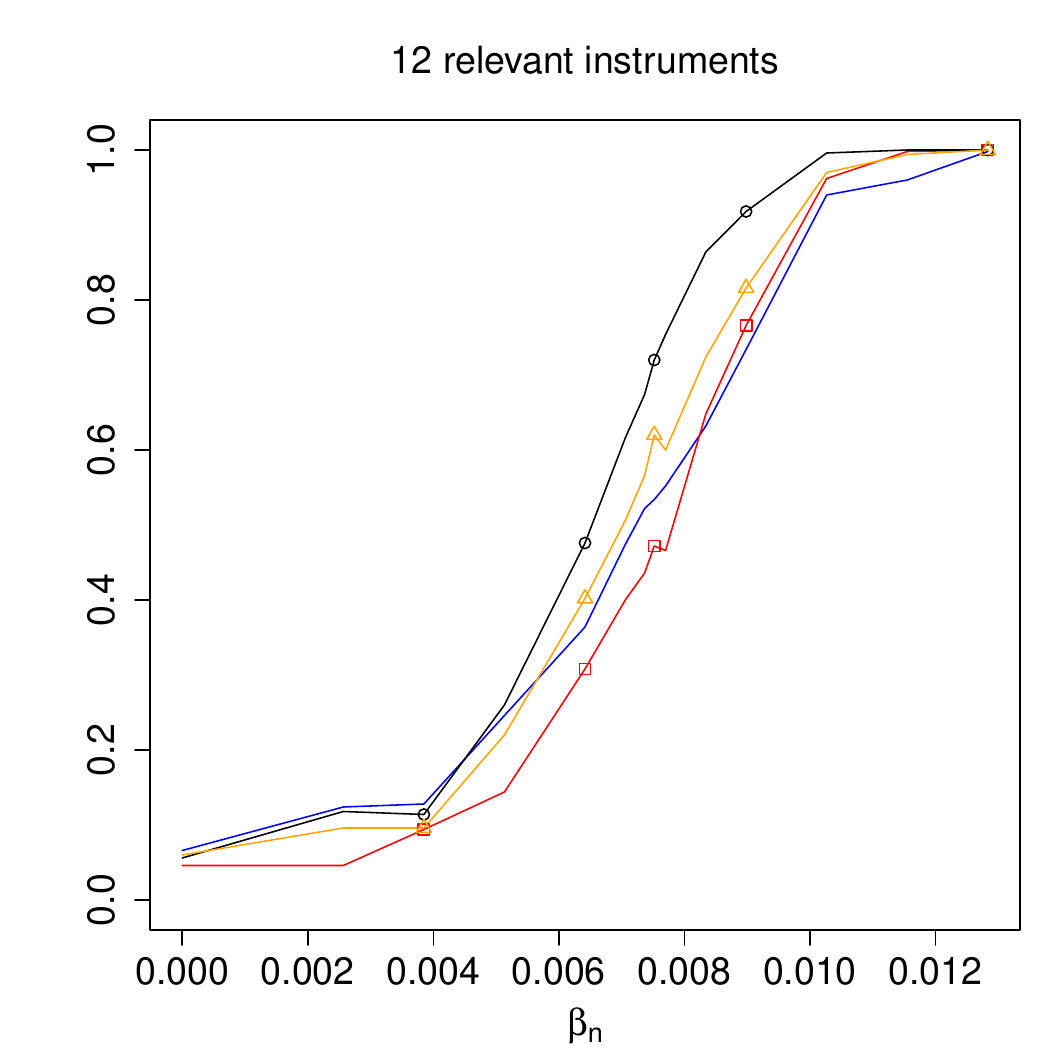}
\hspace{-0.5cm}
\includegraphics[width=5.2cm]{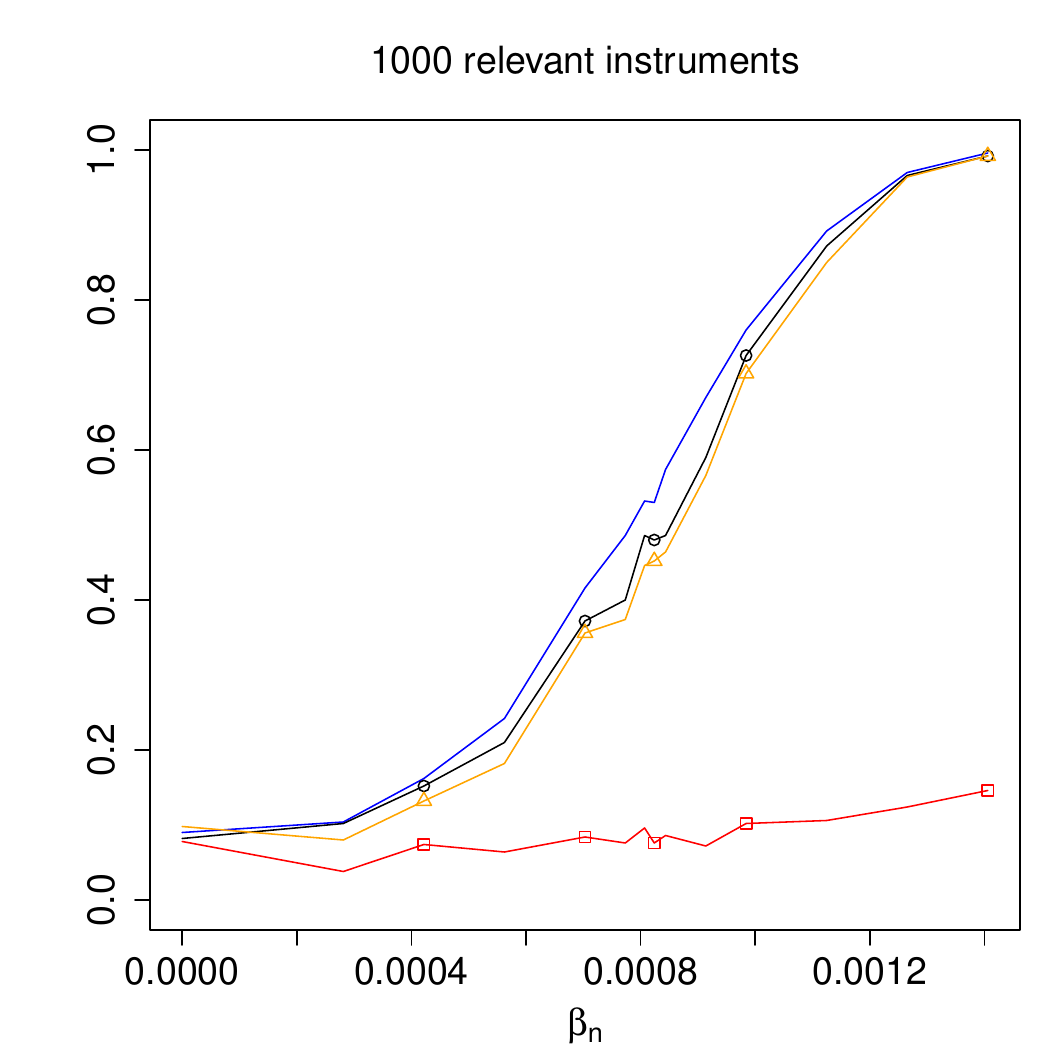}
\vspace{-0.6cm}
\caption{\footnotesize Rejection frequencies for~$(n,d)=(5{,}000,100)$ [first row], $(25{,}000, 500)$ [second row], $(100{,}000, 1{,}000)$ [third row] for sparse [first column], semi-sparse [second column] and dense [third column] alternatives. MS is the jackknifed Anderson-Rubin test of~\cite{mikusheva2020inference} and PE is the power enhancement principle. Full implementation details are given in the body text.}
\label{fig:IV}	
\end{figure}

\begin{figure}
\begin{center}
\includegraphics[width=5.2cm]{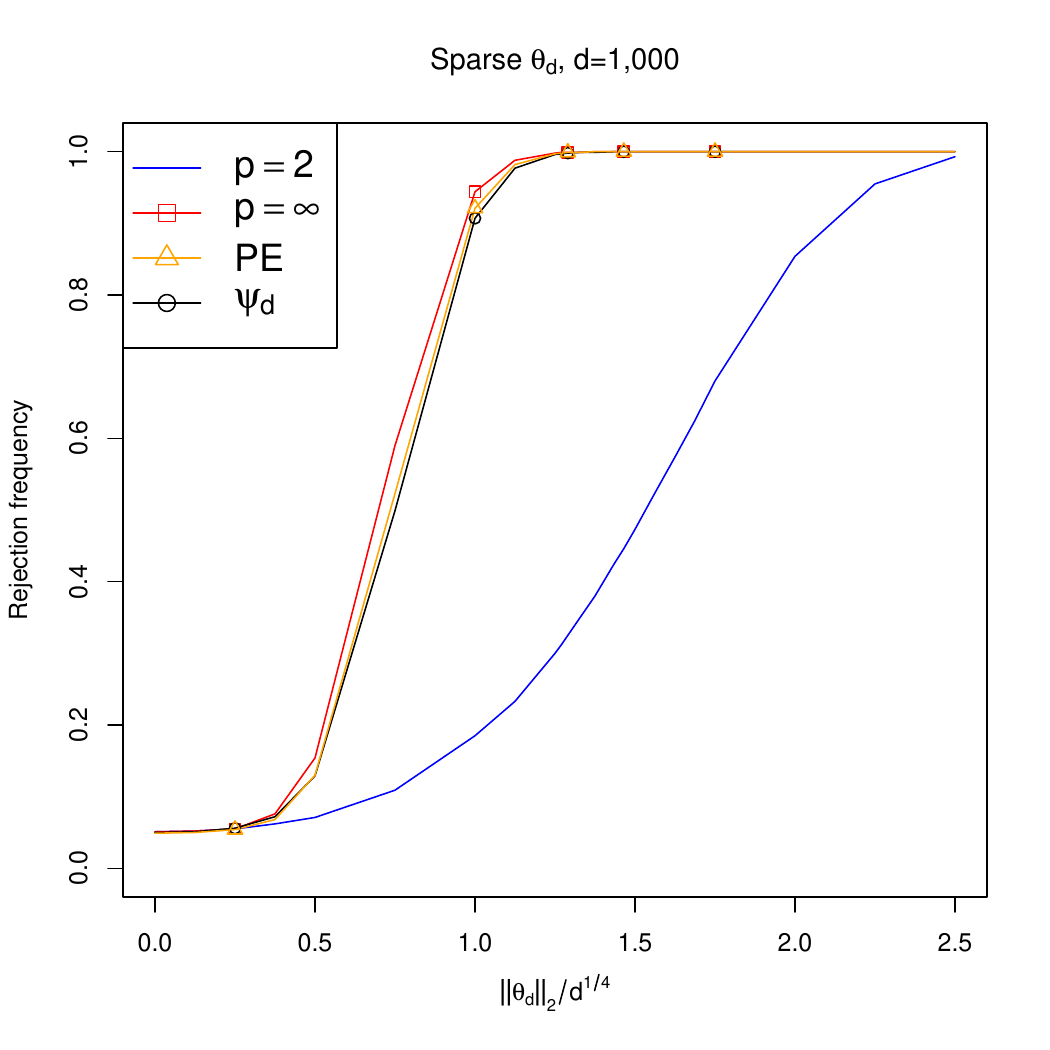}
\hspace{-0.6cm}
\includegraphics[width=5.2cm]{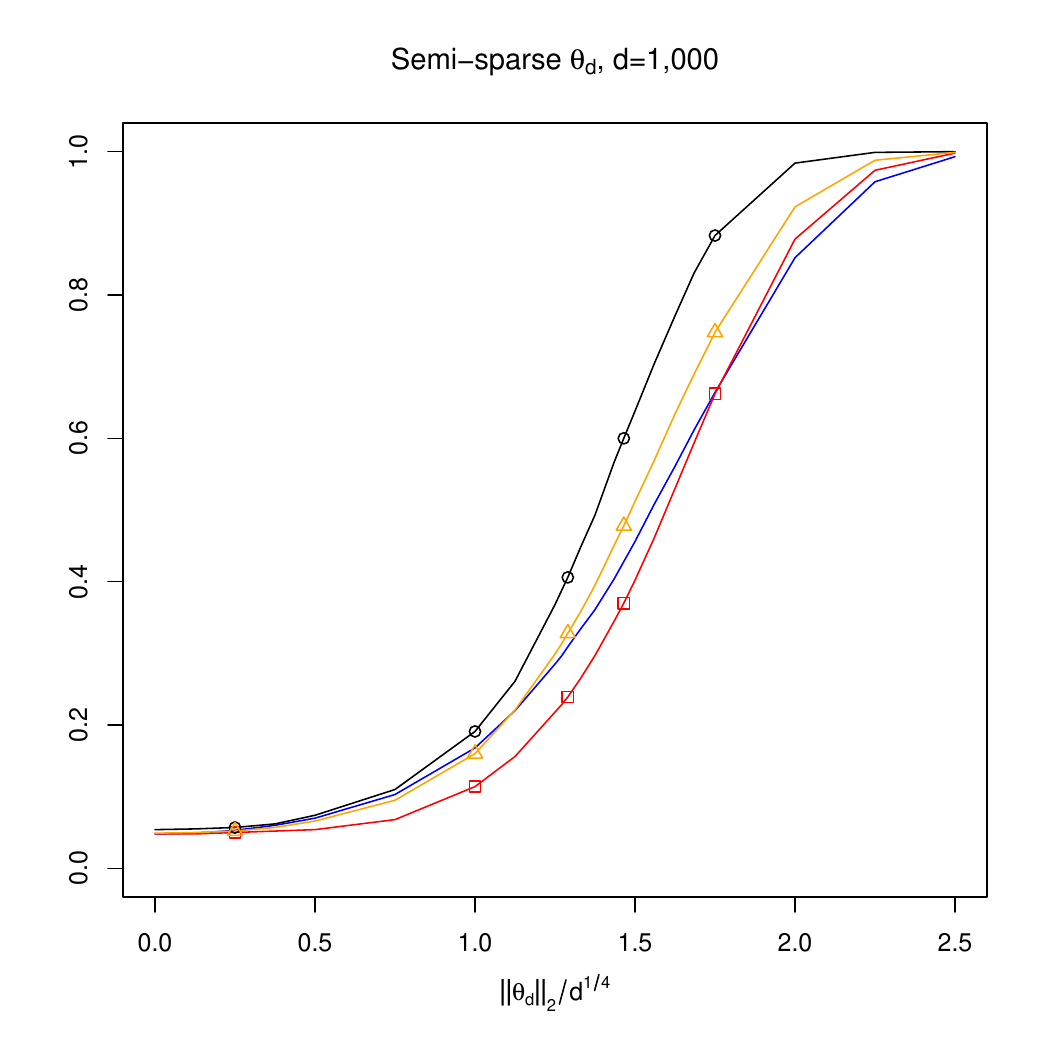}
\hspace{-0.6cm}
\includegraphics[width=5.2cm]{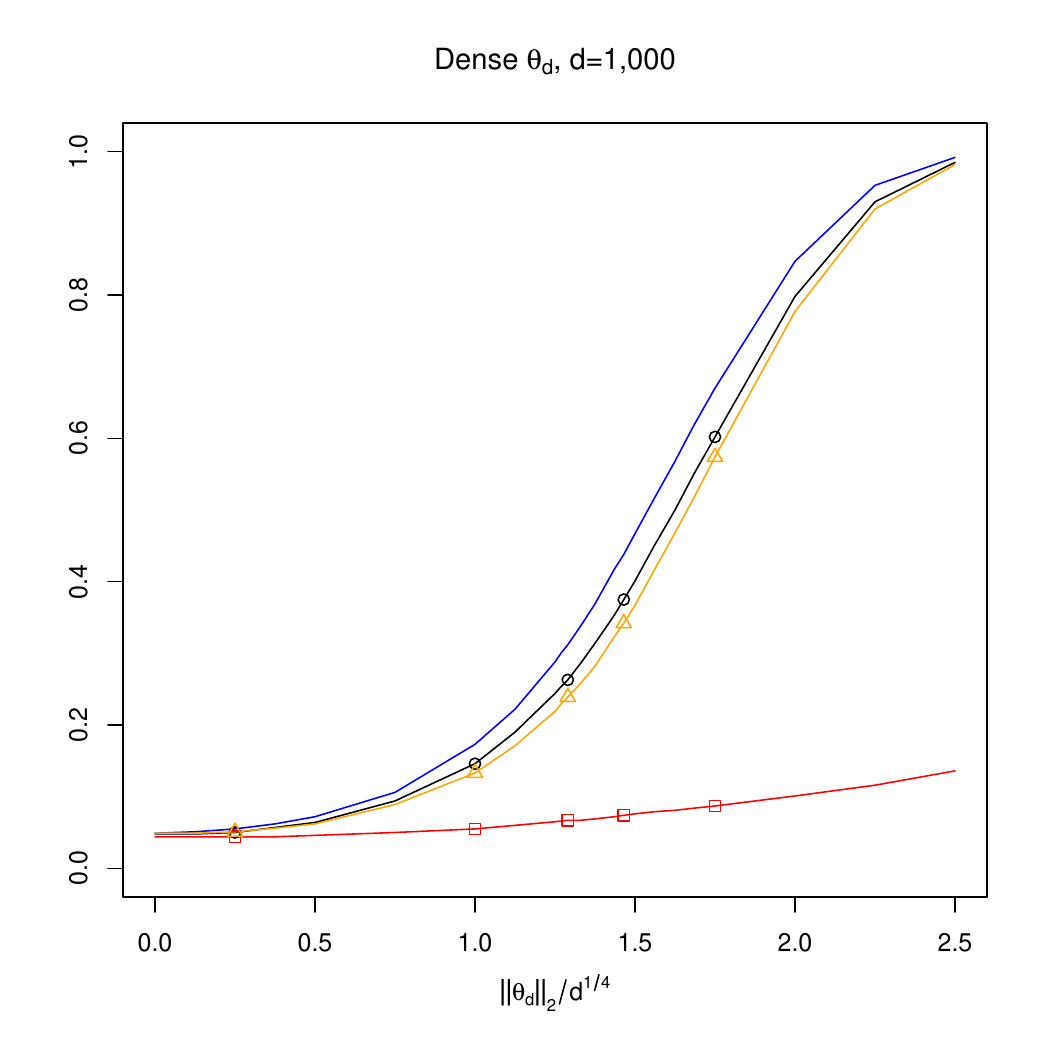}\null 
\vspace{-0.3cm}
\includegraphics[width=5.2cm]{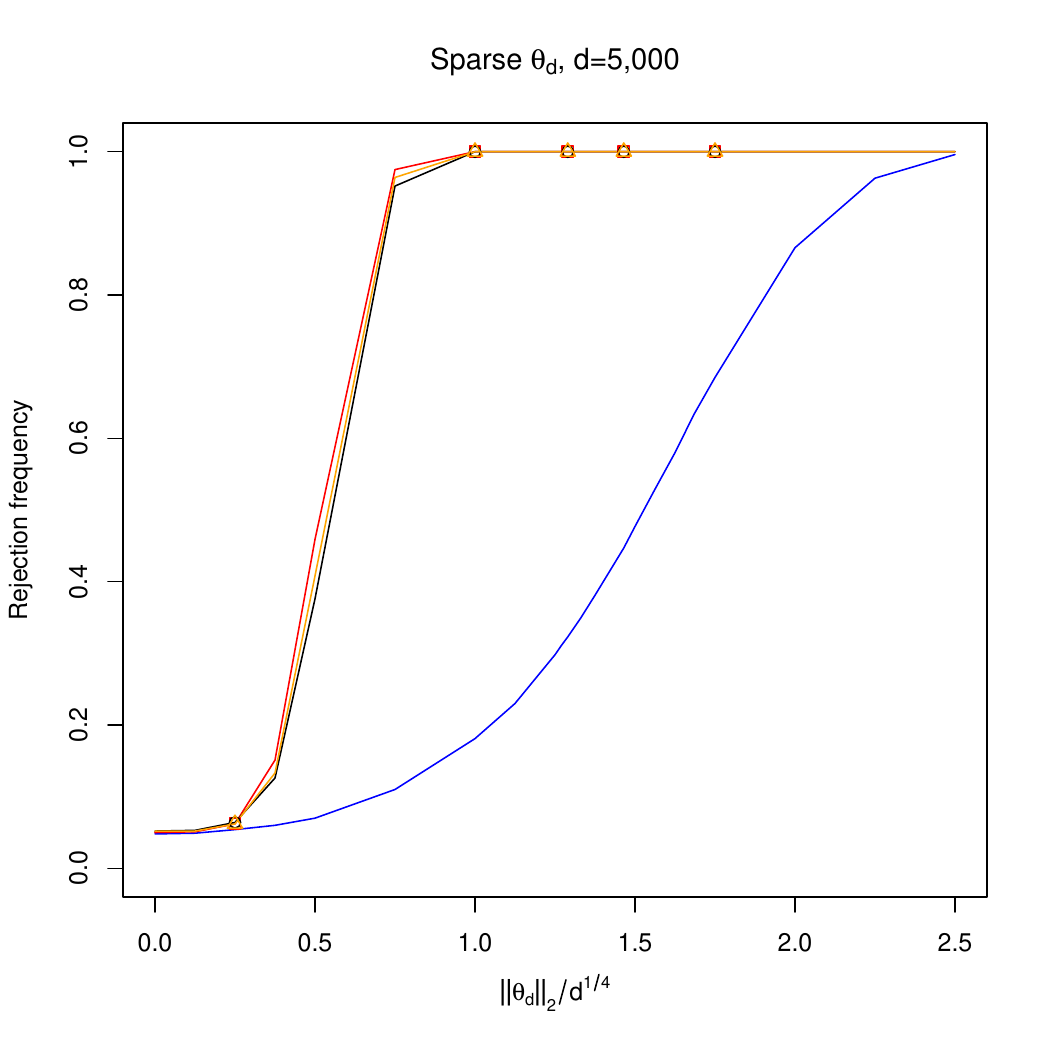}
\hspace{-0.6cm}
\includegraphics[width=5.2cm]{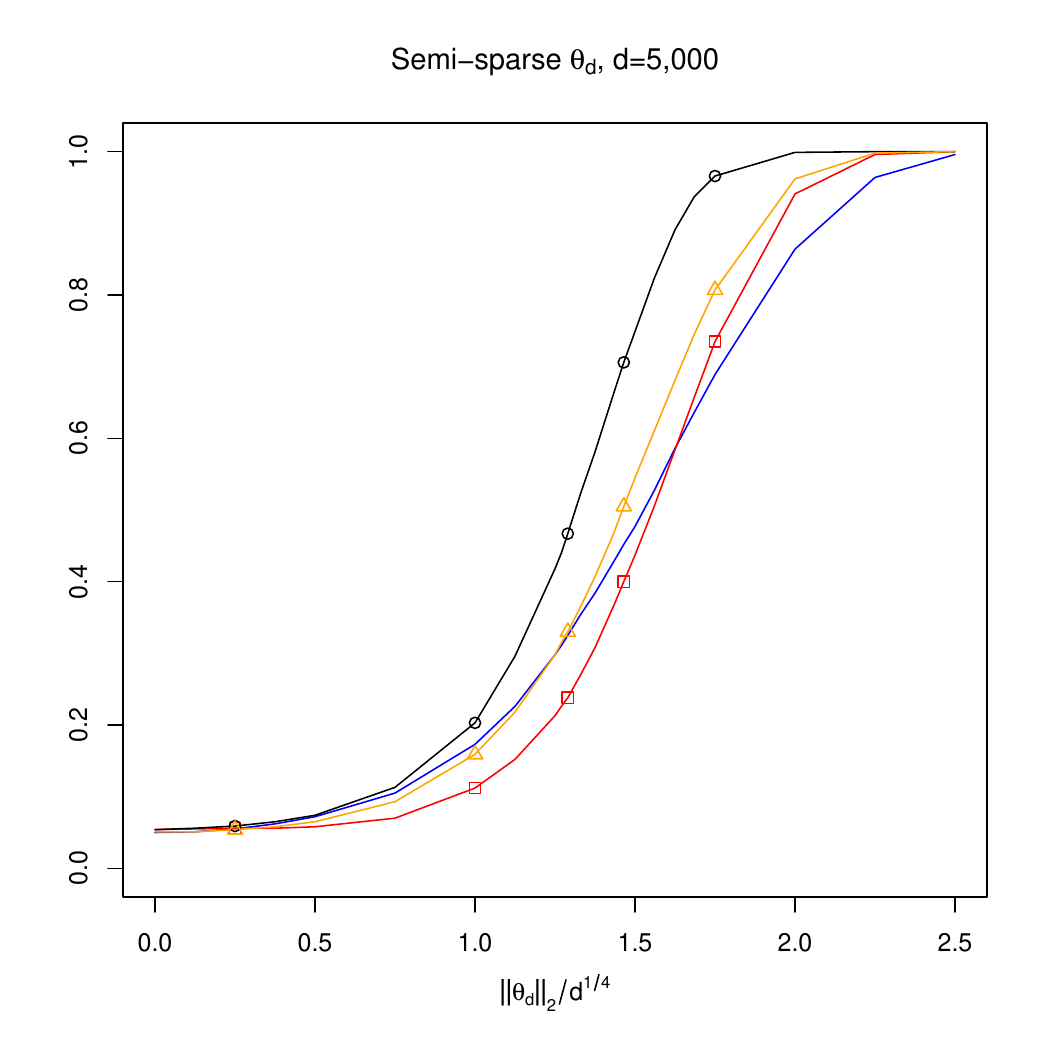}
\hspace{-0.6cm}
\includegraphics[width=5.2cm]{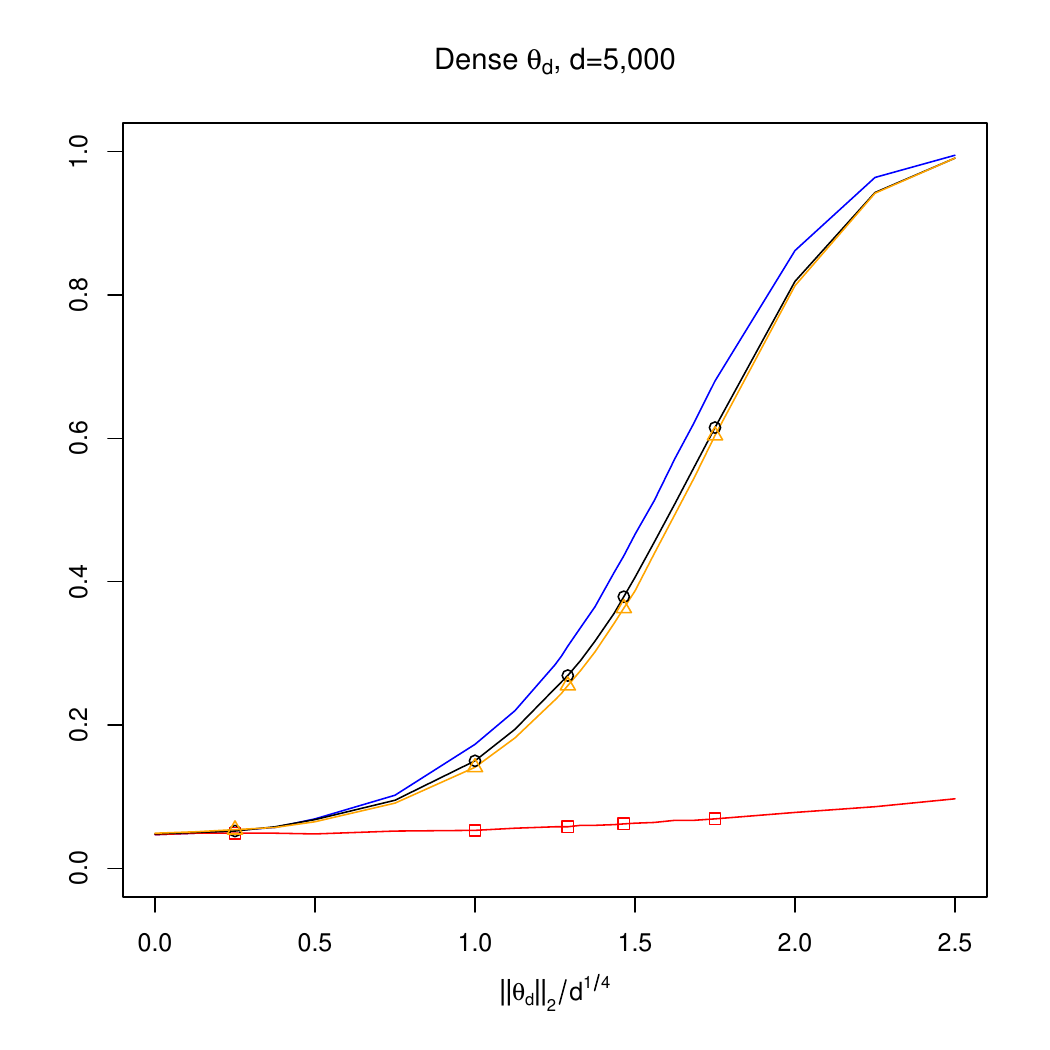}\null
\vspace{-0.3cm}
\includegraphics[width=5.2cm]{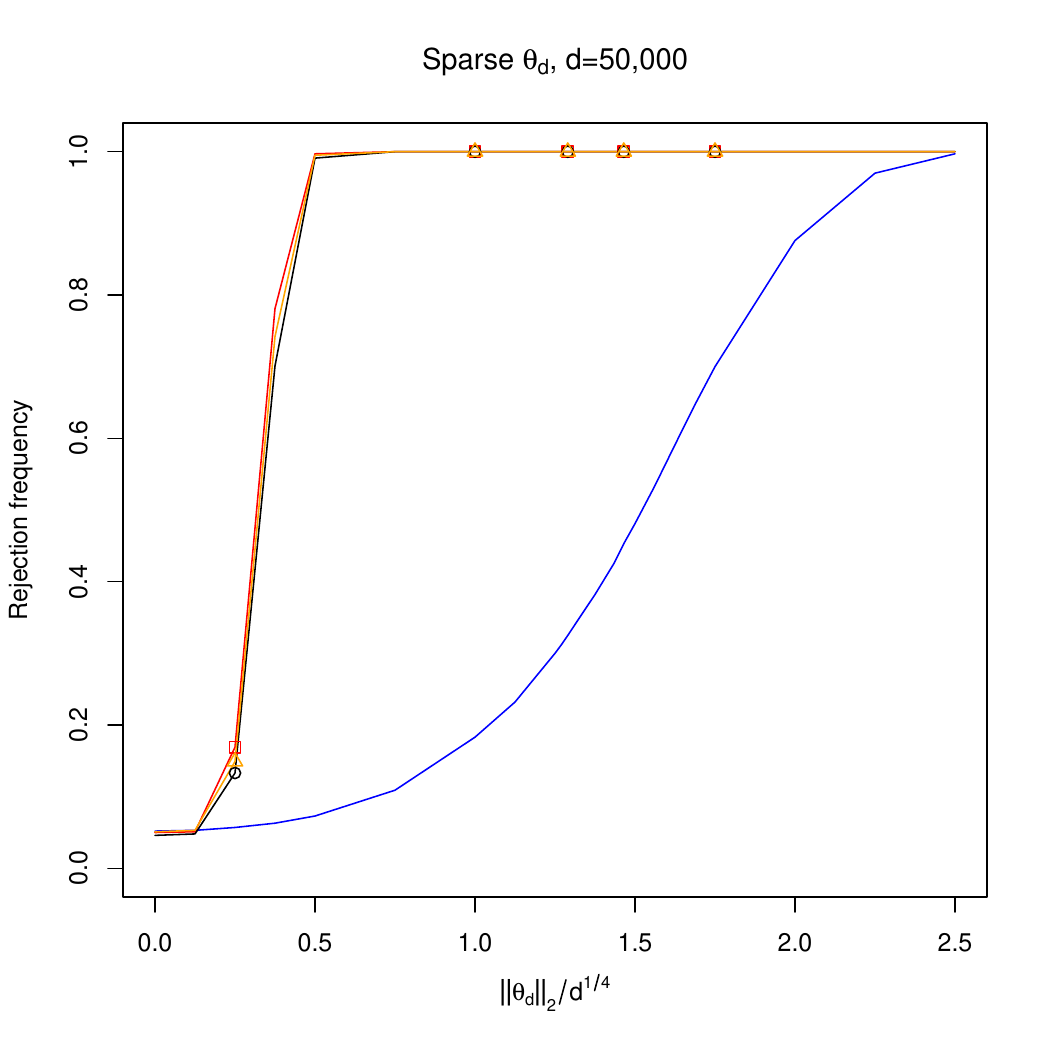}
\hspace{-0.6cm}
\includegraphics[width=5.2cm]{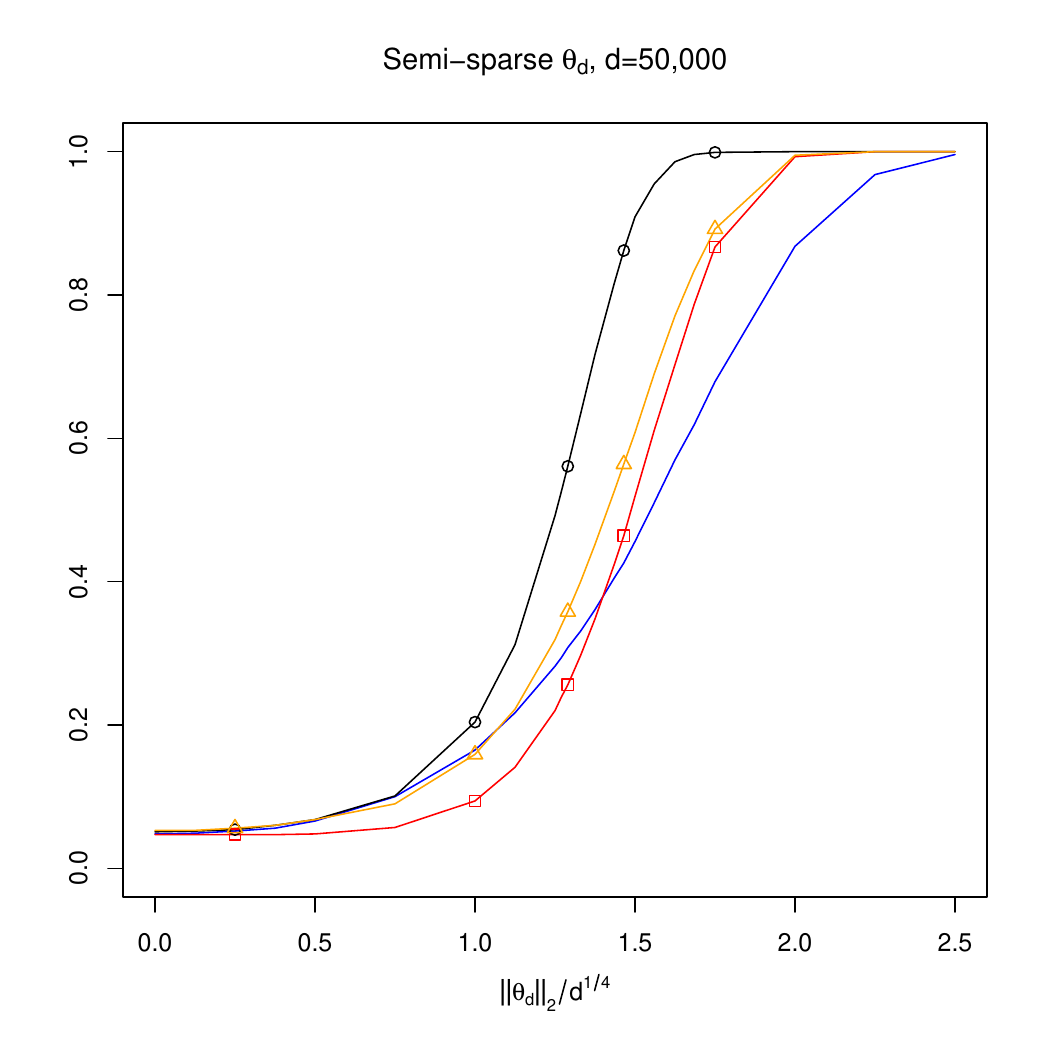}
\hspace{-0.6cm}
\includegraphics[width=5.2cm]{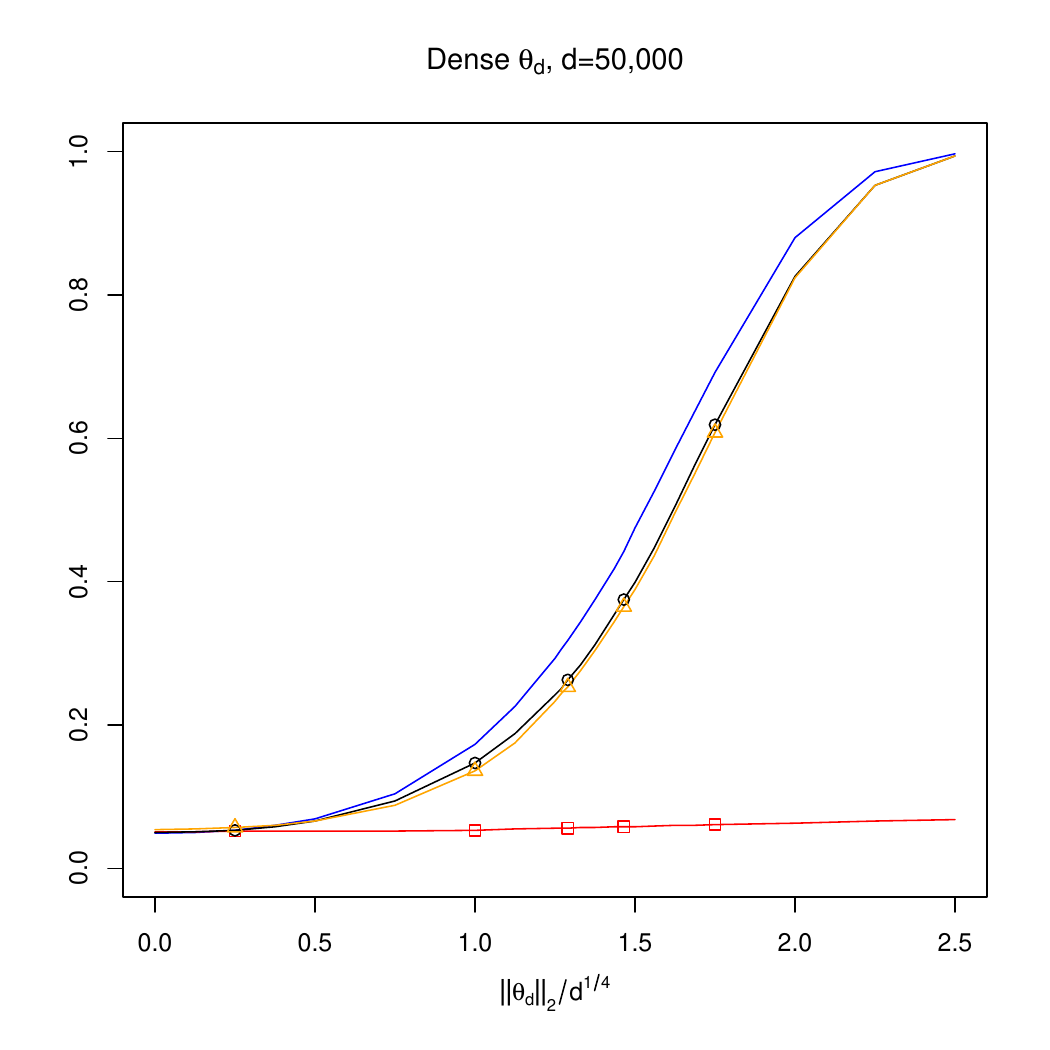}
\null 
\vspace{-0.3cm}
\includegraphics[width=5.2cm]{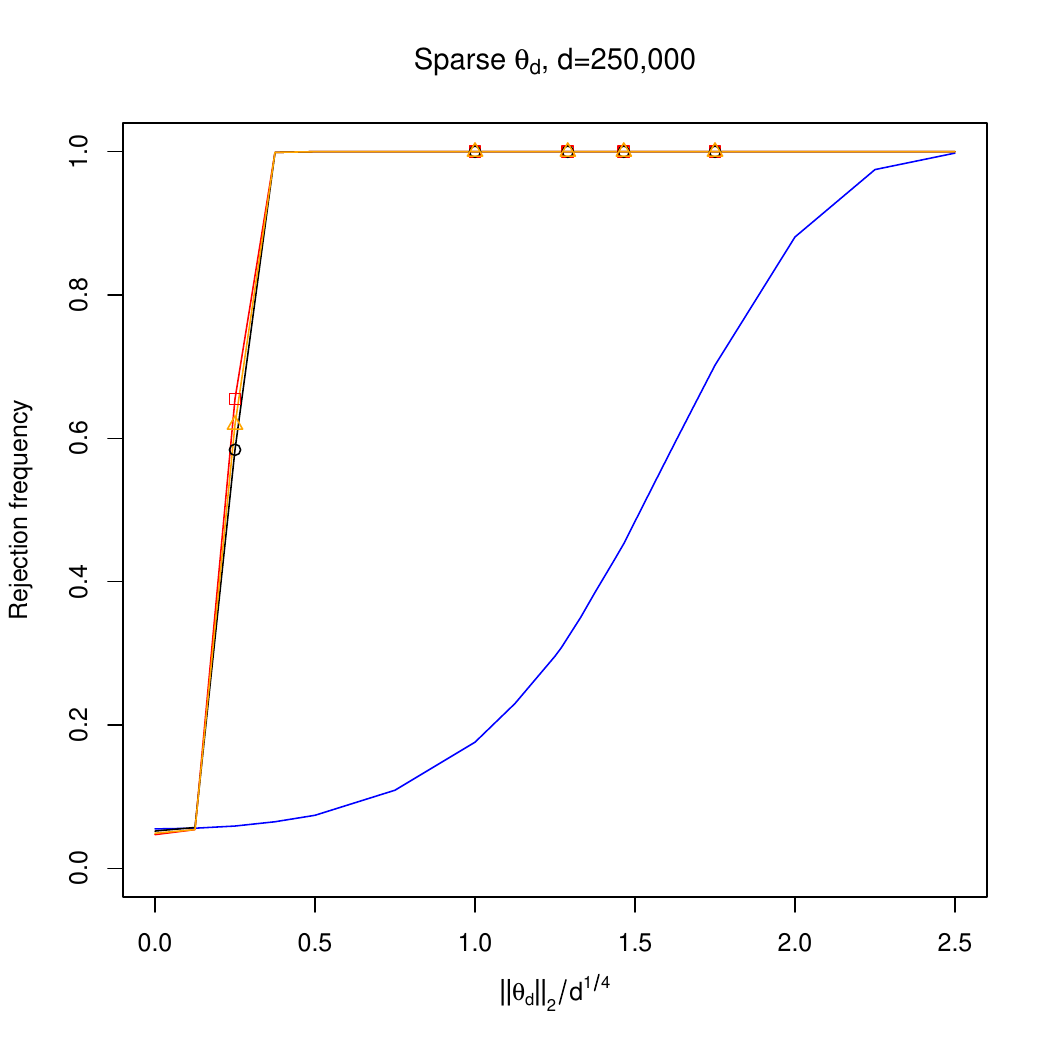}
\hspace{-0.6cm}
\includegraphics[width=5.2cm]{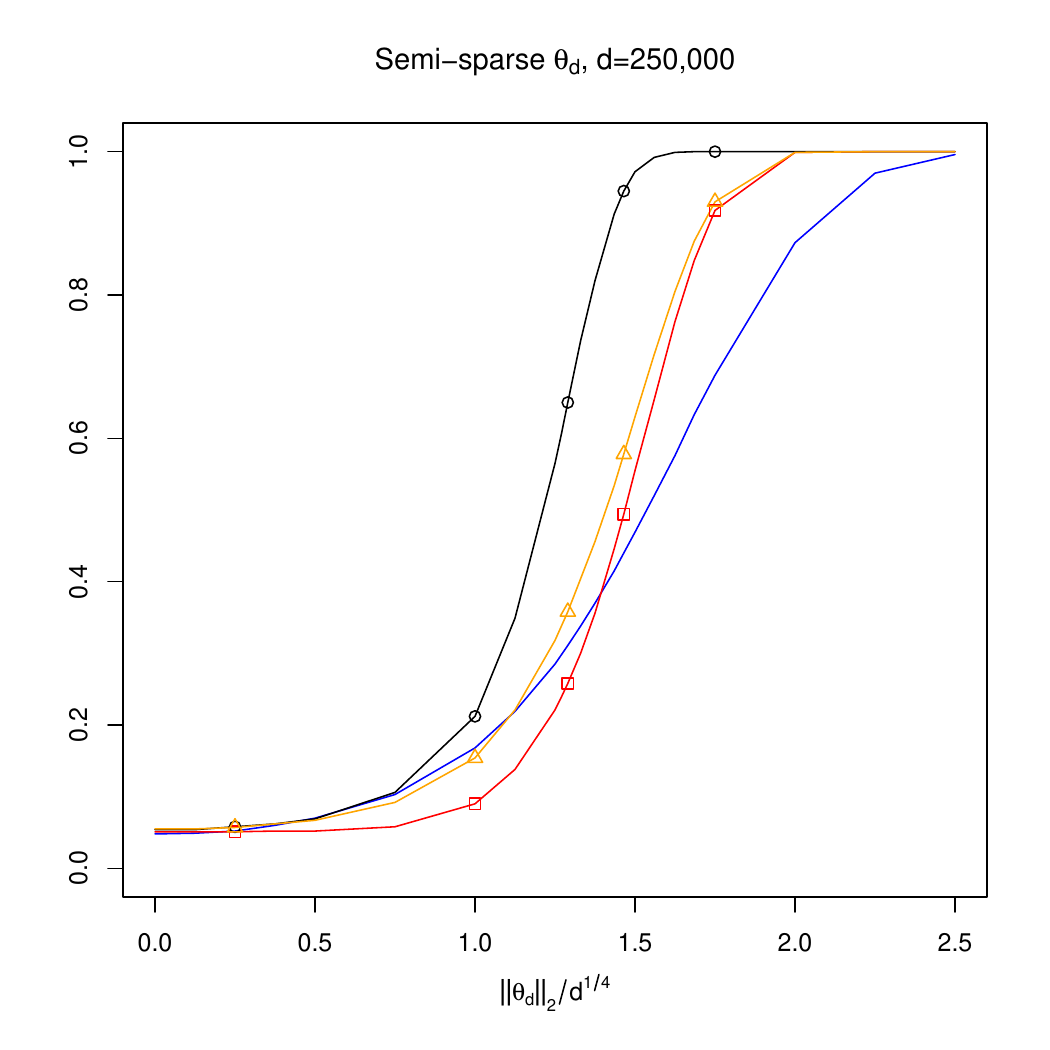}
\hspace{-0.6cm}
\includegraphics[width=5.2cm]{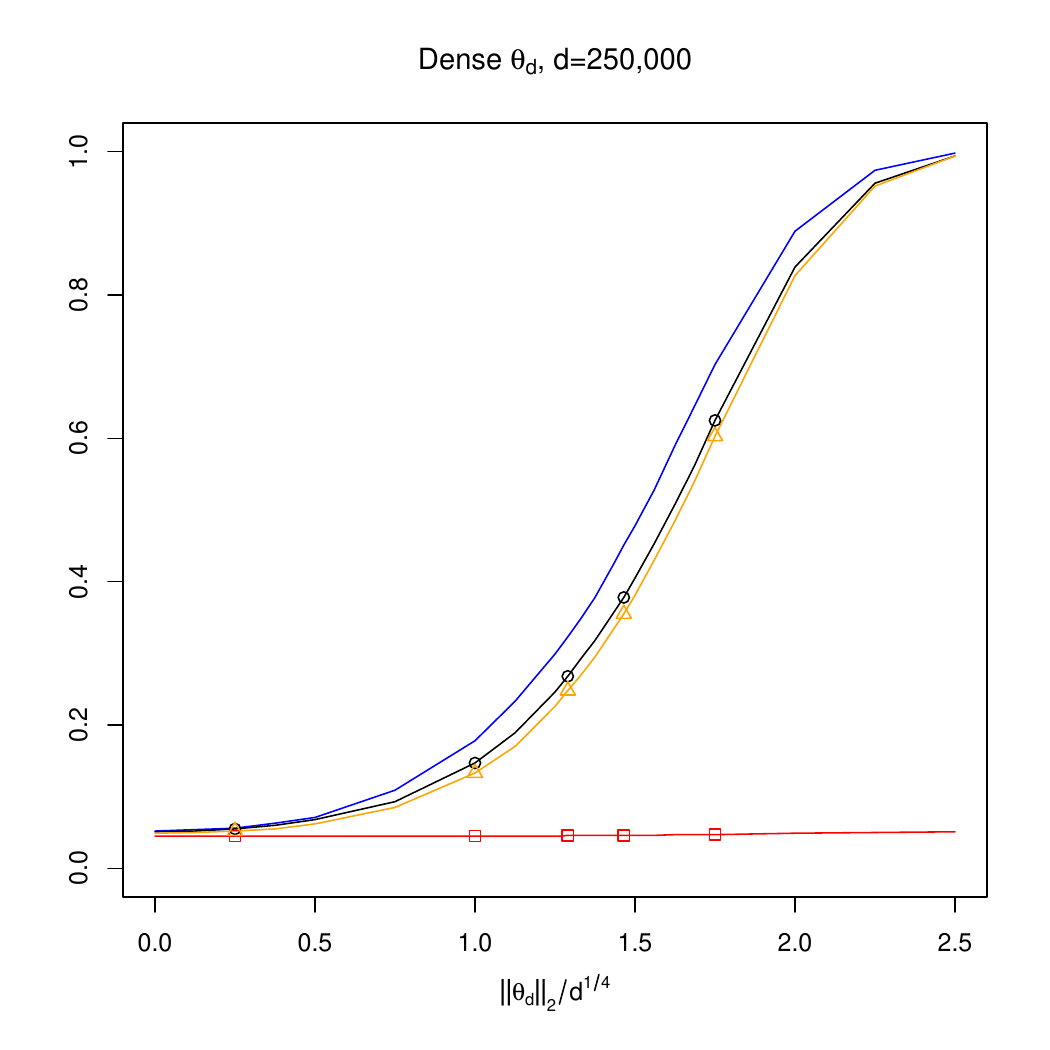}

\end{center}
\caption{\footnotesize Rejection frequencies as a function of~$\enVert[0]{\bm{\theta}_d}_2/d^{1/4}$ for~$d=1{,}000$ (first row), $5{,}000$ (second row), $50{,}000$ (third row) and $250{,}000$ (fourth row) for sparse (first column), semi-sparse (second column) and dense (third column) alternatives. The tests considered are~$\mathds{1}\del[0]{\enVert[0]{\bm{Z}_d}_p\geq \kappa_{d,p}}$ for~$p\in\cbr[0]{2,\infty}$, the PE based test, and~$\psi_d$. Full implementation details are given in the body text.}
\label{fig:Gauss}
\end{figure}

\end{appendix}
\newpage

\bibliographystyle{ecta} 
\bibliography{ref}		

\newpage

\end{document}